\documentclass[11pt]{article}

\usepackage{amsmath,amssymb,amsthm}
\usepackage{graphicx}
\usepackage{booktabs}
\usepackage{natbib}
\usepackage{hyperref}
\usepackage{geometry}
\usepackage{enumerate}

\newtheorem{theorem}{Theorem}[section]

\newtheorem{proposition}[theorem]{Proposition}
\newtheorem{lemma}[theorem]{Lemma}
\newtheorem{corollary}[theorem]{Corollary}

\newcommand{\blind}{0}

\newcommand{\tr}{\mathrm{tr}}

\newcommand{\Var}{\mathbb{Var}}

\hypersetup{
  colorlinks=true,
  linkcolor=blue,
  citecolor=blue,
  urlcolor=blue,
}

\begin{document}

\if0\blind
  \title{Density-valued VAR Models with Latent Factors}
  \author{
    Yasumasa Matsuda\thanks{\texttt{yasumasa.matsuda.a4@tohoku.ac.jp}}\\
    Graduate School of Economics and Management\\ 
    Tohoku University, Japan
    \and
    Michel F. C. Haddad\thanks{\texttt{m.haddad@qmul.ac.uk}}\\
    Department of Business Analytics and Applied Economics\\ 
    Queen Mary University of London, UK
  }
\else
  \title{Density-valued VAR Models with Latent Factors}
  \author{}
\fi

\date{}
\maketitle

\begin{abstract}
We propose a density-valued vector autoregressive model with latent factors for multivariate time series of density functions. Motivated by weekly regional distributions of SARS-CoV-2 cycle threshold (Ct) values in Brazil, we study their distributional dynamics across regions. The Ct value is the number of amplification cycles required for the viral signal to cross a detection threshold (lower Ct values correspond to higher viral load). We estimate each regional density by a B-spline mixture, mapping the mixture weights to a Euclidean space by a generalized logit transform equipped with an isometric inner product, and model the transformed series by a cross-regional VAR with latent factors. This decomposition allows for the separation between strong common movements and directed idiosyncratic dynamics. Directed edges are identified from the idiosyncratic VAR component using one-sided tests with Benjamini--Yekutieli false discovery rate control. Simulations show that increasing the number of estimated factors does not mechanically eliminate genuine idiosyncratic dependence; rather, it mainly removes spuriously detected edges driven by common factor movements. In the real-world data application, the full sample yields only a weak directed network, whereas a substantial network emerges once the first six months are excluded and the density prior is kept weak. The estimated links suggest directed predictive relations from the northern region toward southeastern metropolitan areas.
\end{abstract}


\section{Introduction}
\label{sec:intro}
The cycle threshold (Ct) value from a polymerase chain reaction (PCR) test is the number of amplification cycles required for the viral signal to cross a detection threshold. It is widely used as a proxy for viral load, with lower Ct values corresponding to higher viral load \citep{jones2021estimating,lin2022incorporating}.
In epidemic surveillance, however, the object of interest is often not a single summary such as the weekly mean Ct value, but its entire distribution within a region over time. Changes in this distribution may reflect shifts in the composition of tested individuals, the timing of testing, and the intensity of local transmission. When such distributional information is observed for many regions over time, a crucial empirical question is whether one can detect directional lead--lag dependence in the evolution of regional distributions.

Ct values have increasingly been used in epidemic surveillance as population-level indicators that may contain information beyond the diagnosis of individual cases. In particular, aggregate Ct measures have been studied as early signals of variant emergence and changing transmission intensity. For example, \citet{harrison2023ctvariants} reported that shifts in population-level Ct values preceded subsequent increases in case counts during the emergence of major SARS-CoV-2 variants, suggesting that Ct-based summaries may complement conventional surveillance indicators. The epidemiological relevance of Ct values is also supported by their association with transmission dynamics and disease severity. \citet{baggio2020sars} documented comparable SARS-CoV-2 viral loads in children and adults during early infection, while \citet{liu2020viral} found that severe COVID-19 cases tended to exhibit lower Ct values and longer viral shedding; see also \citet{rao2020ctreview} for a broader review of the relation between Ct values and clinical outcomes. At the same time, most existing work relies on scalar summaries of Ct values, such as means or other aggregate indicators. When the aim is to study how viral load evolves and propagates across regions, such descriptive statistics may be too restrictive. This is because they discard potentially important information about dispersion, skewness, and tail behavior in the regional Ct-value distributions.

A further relevant challenge arises in the statistical analysis of multivariate time series of density functions. A density function is not an ordinary vector-valued observation. Instead, it must satisfy nonnegativity and unit-integral constraints, and therefore does not belong to a linear space under the usual operations. In regional epidemic surveillance one often expects at least two distinct sources of dependence. One is pervasive nationwide comovement driven by common shocks, such as aggregate epidemic waves or nationwide policy changes. The other is the directed lagged dependence across regions, reflecting heterogeneous local propagation. A useful statistical framework must therefore handle both the nonlinear structure of density-valued observations and the separation of common and idiosyncratic dynamic dependence.

Methods for density-valued data \citep{PetersenMuller2016}, factor models \citep{Bai2003,Bai2009}, and Granger-causality networks \citep{Granger1969} have largely developed in parallel. Consequently, there is still no unified methodology for multivariate time series of density functions that separates strong common movements from directed idiosyncratic dependence across regions. This separation is essential in practice, because without controlling for pervasive common factors, apparent cross-regional predictability may simply reflect nationwide shocks rather than genuine directional dependence. Related work is reviewed in Section~\ref{sec:lit}.

In this paper, we propose a density-valued vector autoregressive (VAR) model with latent factors for multivariate time series of density functions observed across regions, motivated by regional Ct-value distributions. For each week and region, we first estimate the Ct-value density by a B-spline mixture and represent it by its mixture-weight vector. We then map these weights from the simplex to a Euclidean coordinate space using a generalized logit transform, equipping that space with an inner product that is isometric to the corresponding spline-mixture function space under the \(L^2\) geometry.
This representation enables us to model the transformed series by a cross-regional VAR with latent factors.

In this specification, the factor part captures common nationwide movements, while the VAR part describes directed lagged dependence in the idiosyncratic distributional dynamics. Based on the estimated VAR coefficients, we construct a directed network using one-sided \(t\)-tests with Benjamini--Yekutieli false discovery rate control \citep{BenjaminiYekutieli2001}, so that selected edges represent Granger-predictive relations after partialling out common factors. An additional advantage of the proposed decomposition is that stability is required only for the idiosyncratic VAR component, not for the observed series as a whole. 
The common factor component may exhibit trend-like nonstationary movements. Thus, the framework can still be applied when the observed multivariate density-valued series exhibits pervasive nonstationary comovement.

The present paper also provides a formal basis for testing directed lagged dependence after removing latent common factors. We develop the asymptotic theory for the factor-adjusted VAR estimator in the transformed Euclidean representation. Under a joint asymptotic regime in which both the basis dimension and the time dimension diverge while the number of regions is fixed, we establish its rate of convergence and, under stronger conditions, derive the asymptotic normality for coefficient-wise inference.

Our contribution is threefold. First, we develop a tractable statistical framework for multivariate time series of density functions that combines B-spline density estimation, a generalized logit transformation of simplex-valued weights, an isometric inner-product structure, latent common factors capturing nationwide movements, and cross-regional VAR dependence. Second, we provide an inferential procedure for identifying directed edges under multiplicity control, thereby separating common-factor comovement from idiosyncratic predictive links. 
Third, our simulations and real-world data analysis clarify how detected edges should be interpreted when the factor dimension, the sample period, and the density prior are varied. In particular, the simulation study shows that genuine idiosyncratic dependence remains detectable after factor adjustment, whereas the empirical application shows that a substantial directed network emerges only after the early pandemic period is excluded and the nationwide prior is retained weak.

\section{Related Work}
\label{sec:lit}

\subsection{Density-valued time series}
\label{sec:lit_density}

A basic difficulty in modeling density-valued data is that the space of density functions is not a linear space in the usual sense. Although densities can be equipped with various metrics, ordinary addition and scalar multiplication do not generally preserve nonnegativity and unit-integral constraints. For this reason, statistical methodology for density-valued data has developed along two main lines: object-oriented approaches and transformation-based approaches.

Object-oriented approaches treat densities directly as structured random objects endowed with an intrinsic geometry or metric, rather than embedding them first into a finite-dimensional Euclidean space. Typical examples are frameworks based on the Wasserstein metric or the Fisher--Rao metric, which view the space of densities as a nonlinear geometric object and formulate statistical analysis through the corresponding tangent spaces or transport maps \citep{PanaretosZemel2020,SrivastavaKlassenJoshiJermyn2011,PetersenZhangKokoszka2022}. In the time-series setting, \citet{ZhangKokoszkaPetersen2022} proposed Wasserstein autoregressive models, providing a rigorous way to describe temporal dependence of density-valued observations directly in Wasserstein space. These approaches are attractive when the density itself, together with its intrinsic geometry, is the primary object of inference.

Transformation-based approaches instead map densities into a linear space, and then apply conventional statistical tools in that representation space. Representative examples include the transformation-based functional data analysis of densities in \citet{PetersenMuller2016}, Bayes-space and simplicial representations for density functions \citep{vanDenBoogaartEgozcuePawlowskyGlahn2010,HronMenafoglioTemplHruzovaFilzmoser2016}, the panel-of-densities analysis in \citet{kutta2025detection}, and the generalized-logit spline-mixture approach in \citet{MatsudaIwafuchi2025}. A common advantage of these methods is that they provide explicit coordinates in which linear operations, regression structures, and standard inferential tools can be developed while retaining essential distributional information.

The present paper belongs to this transformation-based line of research. This is not solely driven by  computational convenience. Our main objective is to estimate cross-regional VAR dependence, remove pervasive common factor movements, and conduct coefficient-wise inference for directed edges in multivariate time series of density functions. While one might in principle formulate dynamic dependence directly in an object-oriented framework as well, it is considerably harder to derive the regression-type asymptotic theory needed for VAR coefficient estimation and edge selection. Conversely, an explicit Euclidean coordinate representation makes the model amenable to both factor adjustment and asymptotic analysis.

At the same time, existing work in either stream of research has paid limited attention to settings where density-valued observations are observed across multiple units and where pervasive common movements must be separated from directed lagged dependence in the idiosyncratic component. This is the setting considered in the present paper.

\subsection{Latent factor components}
\label{sec:lit_factor}

Factor models provide a standard framework for separating common movements from unit-specific variation in high-dimensional data. In macroeconomics and finance, dynamic factor models have been widely used to summarize comovements across many time series by a small number of latent factors \citep{StockWatson2002,ForniHallinLippiReichlin2000,ForniHallinLippiReichlin2005}. The statistical foundations of factor estimation and factor-number selection are developed in, among others, \citet{BaiNg2002} and \citet{Bai2003}. In panel settings, \citet{Bai2009} extended this logic through interactive fixed-effects models with heterogeneous unit-level loadings.

Several types of factor models could in principle be considered for the present problem, including spectral factor formulations as in \citet{ForniHallinLippiReichlin2000,ForniHallinLippiReichlin2005} and dynamic factor formulations as in \citet{StockWatson2002}. In the present paper, however, our main inferential target is not the common component itself, but the coefficient matrix of the factor-adjusted idiosyncratic VAR component. We adopt a Bai-type formulation for two reasons. First, our empirical setting may exhibit pervasive common movements that need not be stationary in the usual time-series sense, and the Bai-type framework allows such common variation to be treated under relatively weak second-moment conditions. Second, our main inferential target is the coefficient matrix of the factor-adjusted idiosyncratic VAR component, and the Bai-type framework leads naturally to regression-style asymptotic arguments for that object.

At the same time, our setting differs fundamentally from the standard interactive fixed-effects regression framework. In \citet{Bai2009}, the primary object is a regression model with regressors that are treated as exogenous or conditionally mean independent, and the latent factor component enters as an unobserved common effect. In our framework, by contrast, the post-adjustment component of interest is itself dynamic and is modeled through a cross-regional VAR. Thus, the factor structure is not introduced merely to absorb nuisance heterogeneity. Rather, it is introduced in order to separate pervasive common movements from the idiosyncratic dynamic component whose lagged dependence is the main object of inference.

The present paper builds on this literature by applying a Bai-type factor framework to transformed density-valued time series. Once density-valued observations are mapped into an isometric Euclidean coordinate system, the resulting model can be written in a factor-augmented representation in which the common component captures pervasive nationwide movements, while the remaining component is modeled by a cross-regional VAR. This makes the Bai-type algebraic framework useful, but the interpretation is importantly different from that in conventional scalar panel applications.

The difference is not only in the type of observed data but also in the role played by dynamics. First, the observed objects here are not scalar outcomes but transformed representations of densities. Second, the regressors in the idiosyncratic equation are generated by lagged dependent variables through a VAR structure. This means that, unlike in standard interactive fixed-effects regressions, the regressor structure is intrinsically dynamic and does not satisfy the same exogeneity interpretation. For this reason, the present model is not a direct application of \citet{Bai2009}. Rather, it is a factor-augmented dynamic regression setting in which the Bai-type decomposition remains useful, but the subsequent asymptotic arguments must be adapted to account for the VAR dependence.

This distinction is important for our substantive objective. In much of the factor-model literature, the common component is itself the main object of interpretation, and the idiosyncratic term is treated largely as residual variation. In our framework, the logic is reversed. The common component is removed because it reflects pervasive nationwide movements that may obscure regional transmission patterns, while the idiosyncratic component is retained and modeled explicitly. The main object of interest is therefore the directed lagged dependence embedded in the idiosyncratic VAR coefficient matrix.

A further advantage of this formulation is that it aligns with our asymptotic regime. In our problem, the number of regions is fixed, while the spline basis dimension and the time length diverge. After the isometric transformation, this produces a high-dimensional representation that differs from the conventional large-panel setting of the factor literature. The resulting asymptotic theory is therefore tailored to transformed density-valued time series with fixed regional dimension and growing basis dimension, rather than to standard scalar panels with many cross-sectional units. In this sense, our approach is Bai-type in algebraic structure, but it is designed for a substantively and theoretically different dynamic environment.

\subsection{Network dependence and Granger-causality analysis}
\label{sec:lit_network}

A large literature studies directed dependence in multivariate time series through vector autoregressive models and Granger-causality analysis. In this framework, directed edges are typically defined by whether past values of one unit may predict current values of another unit, conditional on the remaining variables in the system \citep{Granger1969,Lutkepohl2005}. Related developments include graphical approaches to multivariate time-series dependence \citep{Eichler2007} and network measures of connectedness and spillovers in large systems \citep{BillioGetmanskyLoPelizzon2012,DieboldYilmaz2014}.

Our paper is related to this literature, but differs from standard applications in two essential respects. First, the objects entering the VAR are not scalar observations such as levels, growth rates, or means, but transformed representations of density functions. The resulting directed edges therefore describe predictive relations in the evolution of regional density functions, rather than in a conventional scalar panel. Second, the network is constructed only after removing latent common factors. This distinction is crucial in our setting, because strong nationwide epidemic waves may induce pervasive comovement across regions, and such comovement should not be confused with genuine directed lagged dependence.
These two features jointly determine the interpretation of the network in our setting.

Accordingly, our procedure measures Granger-causality through the coefficient matrix of the idiosyncratic VAR component in a factor-adjusted density-valued VAR model. Existing Granger-causality analyses based on ordinary VAR models are not designed to assess lead--lag relations between density functions themselves, nor do they directly separate such relations from common factor-driven distributional comovement. In this sense, the proposed network is not simply a distributional analogue of a standard scalar VAR network, but a device for identifying predictive dependence in density dynamics after common movements have been partialled out.

From an inferential perspective, our procedure is based on coefficient-wise \(t\)-tests for VAR edges under multiple testing. This requires asymptotic justification for the VAR coefficient estimators after factor adjustment. Building on the Bai-type framework discussed above, we develop such inference in the transformed density-valued setting, where the generalized logit transformation, together with the induced inner-product structure, provides an isometric Euclidean representation of density-valued observations. To control multiplicity in edge selection, we use the false discovery rate procedure of \citet{BenjaminiYekutieli2001}.

Taken together, these three strands of literature motivate the present framework. Transformation-based methods provide a workable Euclidean representation for density-valued observations, factor models provide a way to remove pervasive common movements, and Granger-causality analysis provides a language for directed lagged dependence. The contribution of this paper is to combine them into a single framework suited to multivariate time series of density functions.
\section{Method}
\label{sec:met}

The methodology introduced in this section applies more generally to multivariate time series of density functions observed across units, despite the fact that the empirical application in this paper focuses on regional Ct-value distributions. Let \(C\) denote the number of units, and suppose that at each time \(t\) and unit \(c=1,\ldots,C\), we observe a density supported on a compact interval \([a,b]\). This section first explains how such density-valued time series are constructed in practice, and then proposes a new density-valued VAR model with latent factors by extending the interactive fixed-effects framework of \citet{Bai2009} to density-valued time series. Moreover, we develop estimation and asymptotic theory for the resulting estimators.

\subsection{Density Estimation by B-spline Mixtures}
\label{sec:bspline}

In practice, density-valued time series are observed through samples drawn from an underlying density. We begin by describing a general procedure for estimating a density at each time point and for each unit. Suppose that we observe independent samples \(x_1,\dots,x_n\) from an unknown density \(f\) supported on a compact interval \([a,b]\). We approximate \(f\) by a finite mixture of normalized B-spline basis functions. Let \(\{B_j(\cdot)\}_{j=1}^{J+1}\) be B-splines of degree \(k\) defined on a knot sequence covering \([a,b]\). Normalizing each basis yields the density components
\[
  \phi_j(x)
  =
  \frac{B_j(x)}{\displaystyle\int_a^b B_j(u)\,du},
  \qquad j=1,\dots,J+1.
\]

The unknown density is approximated by
\begin{equation}
  f(x;w)
  =
  \sum_{j=1}^{J+1} w_j\,\phi_j(x),
  \qquad
  w\in\mathcal{S}_{J+1}
  :=
  \left\{ w_j>0,\ \sum_{j=1}^{J+1} w_j = 1 \right\},
  \label{eq:mixture}
\end{equation}
where \(w=(w_1,\dots,w_{J+1})'\) is the mixture-weight vector. The log-likelihood is
\[
  \ell(w)
  =
  \sum_{i=1}^n \log\!\left( \sum_{j=1}^{J+1} w_j\,\phi_j(x_i) \right),
\]
which we maximize using the EM algorithm \citep{DempsterLairdRubin1977,McLachlanPeel2000}. 
In the E-step, we compute the posterior responsibilities
\[
  \tau_{i,j}
  =
  \frac{ w_j\,\phi_j(x_i) }
       { \sum_{\ell=1}^{J+1} w_\ell\,\phi_\ell(x_i) },
  \qquad i=1,\dots,n,\ \ j=1,\dots,J+1,
\]
which represent the conditional probabilities that observation \(x_i\) is associated with the \(j\)th spline component. In the M-step, we update the mixture weights by
\[
  w_j
  =
  \frac{1}{n} \sum_{i=1}^n \tau_{i,j},
  \qquad j=1,\ldots,J+1.
\]
Starting from an initial value for \(w\), we iterate the E-step and M-step until convergence.

To stabilize estimation in small samples, we also consider a Dirichlet-regularized version of the procedure. Specifically, suppose that
\[
  w \sim \mathrm{Dir}(\alpha_1^{(0)},\dots,\alpha_{J+1}^{(0)}).
\]
Then the posterior parameters are
\begin{align}
  \alpha_j^{\mathrm{post}}
  =
  \gamma \alpha_j^{(0)} + \sum_{i=1}^n \tau_{i,j},
  \qquad j=1,\ldots,J+1,
  \label{eq:posterior}
\end{align}
where \(\gamma\) controls the prior strength. The corresponding posterior mean is
\begin{align}
  w_j
  =
  \frac{\alpha_j^{\mathrm{post}}}
       {\sum_{\ell=1}^{J+1} \alpha_\ell^{\mathrm{post}}},
  \qquad j=1,\ldots,J+1.
  \label{Mstep}
\end{align}
The E-step remains unchanged, while the M-step is replaced by the posterior-mean update in \eqref{eq:posterior} and \eqref{Mstep}. Starting from an initial value for \(w\), we iterate the modified E-step and M-step until convergence.

Applying this procedure separately to each time \(t\) and unit \(c\), we obtain a sequence of weight vectors \(w_{t,c}\in\mathcal{S}_{J+1}\). These vectors serve as the observed density-valued time series used in the subsequent analysis. In the Ct-value application, the units correspond to regions.

\subsection{Generalized logit and softmax transforms}
\label{sec:logit}

Let \(w_{t,c}\in\mathcal{S}_{J+1}\) denote the estimated B-spline mixture weights for unit \(c=1,\ldots,C\) at time \(t\). A basic difficulty is that these weight vectors lie in the simplex \(\mathcal{S}_{J+1}\), so they satisfy positivity and unit-sum constraints and therefore do not form a linear space under the usual addition and scalar multiplication. For this reason, the weight vectors cannot be analyzed directly by standard linear time-series methods.

To obtain a linear representation, we map the simplex-valued weights into a finite-dimensional Euclidean space by the generalized logit transform of \citet{MatsudaIwafuchi2025}. This transformation provides an explicit coordinate system for the B-spline mixture weights while retaining a one-to-one correspondence with the underlying spline-mixture representation. We then equip the transformed space with an inner product chosen so that it is isometric to the associated spline-mixture function space under the \(L^2[a,b]\) inner product. This construction allows us to analyze density-valued time series by linear methods without losing the geometric meaning of distances in the original spline-mixture representation.

For \(p=(p_1,\ldots,p_{J+1})'\in\mathcal{S}_{J+1}\) with base component \(p_{J+1}\), define, for \(\delta>0\),
\[
\operatorname{logit}_\delta(p)
=
\left(
\log\frac{\delta+p_1}{\delta+p_{J+1}},
\ldots,
\log\frac{\delta+p_J}{\delta+p_{J+1}}
\right)'
\in \mathbb{R}^{J}.
\]
When \(\delta=0\), \(\operatorname{logit}_\delta\) reduces to the usual multinomial logit transform from \(\mathcal{S}_{J+1}\) to \(\mathbb{R}^J\). For \(\delta>0\), the shift by \(\delta\) improves numerical stability when some components of \(p\) are close to zero. It is convenient to extend the domain to
\[
\mathcal{S}^\delta_{J+1}
=
\left\{
p\in\mathbb{R}^{J+1}:\;
p_j>-\delta,\;
\sum_{j=1}^{J+1}p_j=1
\right\},
\]
which strictly contains the probability simplex \(\mathcal{S}_{J+1}\).

The inverse map of \(\operatorname{logit}_\delta\) is given by, 
for $b\in\mathbb R^J$,
\[
\operatorname{softmax}_\delta(b)
=
\left(
\frac{((J+1)\delta+1)e^{b_1}}{1+\sum_{j=1}^J e^{b_j}}-\delta,
\ldots,
\frac{((J+1)\delta+1)e^{b_J}}{1+\sum_{j=1}^J e^{b_j}}-\delta,
\frac{(J+1)\delta+1}{1+\sum_{j=1}^J e^{b_j}}-\delta
\right)'
\in\mathcal S_{J+1}^\delta.
\]
When \(\delta=0\), \(\operatorname{softmax}_\delta\) reduces to the standard softmax map to \(\mathcal{S}_{J+1}\). It is straightforward to verify that \(\operatorname{softmax}_\delta(b)\in\mathcal S_{J+1}^\delta\), since each component is strictly larger than \(-\delta\) and
\[
\sum_{m=1}^{J+1}\bigl(\operatorname{softmax}_\delta(b)\bigr)_m
=
\frac{((J+1)\delta+1)\left(1+\sum_{j=1}^J e^{b_j}\right)}{1+\sum_{j=1}^J e^{b_j}}-(J+1)\delta
=
1.
\]
Moreover,
\[
\operatorname{softmax}_\delta(\operatorname{logit}_\delta(p))=p
\quad\text{for } p\in\mathcal S_{J+1}^\delta,
\qquad
\operatorname{logit}_\delta(\operatorname{softmax}_\delta(b))=b
\quad\text{for } b\in\mathbb R^J.
\]
Hence, \(\operatorname{logit}_\delta\) and \(\operatorname{softmax}_\delta\) define a bijection between \(\mathcal S_{J+1}^\delta\) and \(\mathbb R^J\).

Using this bijection, we transport the usual addition and scalar multiplication in \(\mathbb R^J\) to \(\mathcal S_{J+1}^\delta\). Specifically, for \(w_1,w_2\in\mathcal S_{J+1}^\delta\) and \(\alpha\in \mathbb R\), define
\begin{align*}
  w_1\oplus w_2
  &:=
  \operatorname{softmax}_\delta\!\bigl(\operatorname{logit}_\delta(w_1)+\operatorname{logit}_\delta(w_2)\bigr),\\
  \alpha\otimes w_1
  &:=
  \operatorname{softmax}_\delta\!\bigl(\alpha\,\operatorname{logit}_\delta(w_1)\bigr).
\end{align*}
Thus, via the generalized logit transform, we may identify \(\mathcal S_{J+1}^\delta\) with \(\mathbb R^J\) and equip it with a linear structure.

Next we introduce an inner product on the transformed space so that the coordinate representation is isometric to its image in the associated spline-mixture function space under the \(L^2[a,b]\) inner product.
Let \(e_i\in\mathbb{R}^J\), \(i=1,\ldots,J\), be the canonical unit vector, and define
\[
  \hat e_i=\operatorname{softmax}_\delta(e_i)\in\mathcal S_{J+1}^\delta.
\]
We identify any weight vector \(p=(p_1,\ldots,p_{J+1})'\in\mathcal S_{J+1}^\delta\) with the associated spline-mixture function
\[
  \Psi(p)(s)
  :=
  \sum_{m=1}^{J+1} p_m\,\phi_m(s),
  \qquad s\in[a,b].
\]
For brevity, write \(\hat e_i(s):=\Psi(\hat e_i)(s)\). Define the Gram matrix
\[
  (H_J)_{ij}
  =
  \int_a^b \hat e_i(s)\,\hat e_j(s)\,ds,
  \qquad i,j=1,\ldots,J.
\]
This induces the inner product
\[
  \langle u,v\rangle_H
  =
  u'H_Jv,
  \qquad u,v\in\mathbb{R}^J.
\]
By construction, this is the pullback of the \(L^2[a,b]\) inner product of the associated spline-mixture functions through the map \(b\mapsto \Psi(\operatorname{softmax}_\delta(b))\). Hence the transformed coordinate space is isometric to its image in the spline-mixture function space.

The introduction of \(\delta>0\) improves numerical stability when some mixture weights are close to, or equal to, zero. At the same time, because \(\mathcal S_{J+1}^\delta\) strictly contains the probability simplex, the associated spline-mixture functions may take negative values for some choices of transformed coordinates. This is a potentially controversial feature. In our framework, however, \(\mathcal S_{J+1}\) remains the empirical starting point, and the enlargement to \(\mathcal S_{J+1}^\delta\) is used as a local device to obtain a stable bijection and linear structure for statistical modeling. 
The gain is computational and inferential tractability, while interpretation remains anchored in the original density representation through the inverse transformation.

We represent the density-valued time series for unit \(c\) by the transformed series
\[
  Y_{t,c}
  =
  \operatorname{logit}_\delta(w_{t,c})
  \in \mathbb{R}^{J}.
\]
Stacking across units, we define
\begin{align}
  Y_t
  =
  (Y_{t,1}',\ldots,Y_{t,C}')'
  \in \mathbb{R}^{CJ}.
  \label{eq:data}
\end{align}
For the stacked vectors \(Y_t\in\mathbb{R}^{CJ}\), we introduce the block-diagonal metric
\begin{align}
  \mathcal{H}=I_C\otimes H_J.
  \label{eq:inner_product}
\end{align}
In the Ct-value application, the units correspond to regions.

\subsection{Density-valued VAR models with latent factors}
\label{sec:dfm}

\subsubsection{Model description}

Suppose we observe the transformed density-valued time series \(Y_t\) defined in \eqref{eq:data}. Recall that each \(Y_{t,c}\in\mathbb{R}^J\) is obtained by applying the generalized logit map to the spline-mixture weights of the density at time \(t\) and unit \(c\), and that this coordinate space is equipped with an inner product making it isometric to the corresponding spline-mixture image under the \(L^2\) geometry. A natural starting point is the density-valued VAR model without latent factors,
\begin{align}
Y_{t,c}
=
\sum_{k=1}^p \sum_{d=1}^C V_{k,cd} Y_{t-k,d}
+u_{t,c},
\qquad c=1,\ldots,C,
\label{eq:plainvar}
\end{align}
where \(V_k\in\mathbb{R}^{C\times C}\) captures cross-regional autoregressive dependence and \(u_{t,c}\in\mathbb{R}^J\) is a disturbance term. This specification directly represents lagged dependence across regions in the transformed density-valued series.

However, in the present setting one expects strong contemporaneous comovement across regions, driven for example by nationwide epidemic waves or other aggregate shocks. If such common movements remain in the disturbance term \(u_{t,c}\), then a direct VAR fit may attribute part of this pervasive comovement to cross-regional lagged dependence, making the resulting network difficult to interpret. 
One possible response is to impose a structural VAR decomposition on the contemporaneous dependence \citep{Lutkepohl2005,KilianLutkepohl2017}. In practice, however, such an approach requires additional identifying restrictions, and the resulting analysis may depend on assumptions about contemporaneous ordering or other structural constraints.

To separate common movements from directed idiosyncratic dynamics without imposing such structural restrictions, we augment \eqref{eq:plainvar} with latent factors and consider the following density-valued VAR model with latent factors:
\begin{align}
Y_{t,c}
=
\sum_{k=1}^p \sum_{d=1}^C V_{k,cd} Y_{t-k,d}
+\Lambda_c f_t
+\varepsilon_{t,c},
\qquad c=1,\ldots,C.
\label{eq:factorvar}
\end{align}
Here, \(V_k\in\mathbb{R}^{C\times C}\) captures cross-regional autoregressive dependence, \(\Lambda_c\in\mathbb{R}^{J\times r}\) is the loading matrix for unit \(c\), and \(f_t\in\mathbb{R}^r\) denotes the vector of latent common factors. Thus, the model decomposes the transformed density-valued series into a common factor component and an idiosyncratic component with cross-regional VAR dependence.

Defining
\[
\Lambda
=
(\Lambda_1',\ldots,\Lambda_C')'
\in\mathbb{R}^{CJ\times r},
\]
the model can be written in matrix form as
\begin{align}
Y_t
=
\sum_{k=1}^p X_{t-k}\,\mathrm{Vec}(V_k')
+\Lambda f_t
+\varepsilon_t,
\label{eq:factorvar_mtrx}
\end{align}
where
\[
X_t = I_C \otimes (Y_{t,1},\ldots,Y_{t,C}).
\]

From this matrix representation, \eqref{eq:factorvar} can be viewed as a density-valued dynamic extension of the interactive fixed-effects framework of \citet{Bai2009}. Relative to standard factor models, it allows density-valued observations together with cross-regional VAR dependence in the idiosyncratic component.

\subsubsection{Estimation}

We estimate \(V_k\), \(\Lambda\), and \(f_t\) in \eqref{eq:factorvar} under the normalization
\begin{align*}
    (CJ)^{-1}\Lambda'\mathcal{H}\Lambda&=I_r,\\
    \frac{1}{T-p}\sum_{t=p+1}^T f_tf_t'&=\mathrm{diag}(v_1,\ldots,v_r).
\end{align*}
For notational simplicity, we assume that \(Y_t\) has zero mean. For nonzero-mean processes, we replace \(Y_t\) with
\[
Y_t-\frac{1}{T}\sum_{s=1}^T Y_s.
\]

Our estimation procedure alternates between PCA estimation of the latent factor component and projected least-squares estimation of the VAR component. Given current VAR coefficients, we extract the factor structure from the residual covariance matrix; given the resulting loading space, we partial out the factor component and update the VAR coefficients by least squares. This iterative scheme parallels \citet{Bai2009}, adapted here to the whitened transformed representation of density-valued observations.

Let \(\mathcal{K}=I_C\otimes K_J\), where \(K_JK_J'=H_J\). Define
\[
  \tilde Y_t = \mathcal{K}'Y_t,\qquad
  \tilde\Lambda = \mathcal{K}'\Lambda,\qquad
  \tilde X_t = \mathcal{K}'X_t,\qquad
  \tilde\varepsilon_t = \mathcal{K}'\varepsilon_t.
\]
Then the model becomes
\[
\tilde Y_t
=
\sum_{k=1}^p
\tilde X_{t-k}\,\mathrm{Vec}(V_k')
+
\tilde\Lambda f_t
+
\tilde\varepsilon_t,
\]
which rewrites the model in a coordinate space equipped with the standard inner product.

\paragraph{Step 1: Residualizing the VAR component.}
Given current values of \(V_k\), compute
\[
  \tilde R_t
  =
  \tilde Y_t
  -
  \sum_{k=1}^p
  \tilde X_{t-k}\,\mathrm{Vec}(V_k').
\]

\paragraph{Step 2: PCA for the factor component.}
Using these residuals, update the loading space by PCA. Form the covariance matrix
\[
  \tilde S
  =
  \frac{1}{T-p}
  \sum_{t=p+1}^T
  \tilde R_t\tilde R_t'.
\]
Let \(\hat{\tilde\Lambda}\in\mathbb{R}^{CJ\times r}\) collect the top \(r\) eigenvectors of \(\tilde S\), normalized such that
\[
(CJ)^{-1}\hat{\tilde\Lambda}'\hat{\tilde\Lambda}=I_r.
\]

\paragraph{Step 3: FWL estimation of \(V_k\).}
Define the projection matrix
\[
  M_{\hat{\tilde\Lambda}}
  =
  I_{CJ}-(CJ)^{-1}\hat{\tilde\Lambda}\hat{\tilde\Lambda}'.
\]
The projected model is
\[
M_{\hat{\tilde\Lambda}}\tilde Y_t
=
\sum_{k=1}^p
M_{\hat{\tilde\Lambda}}\tilde X_{t-k}\,\mathrm{Vec}(V_k')
+
M_{\hat{\tilde\Lambda}}\tilde\varepsilon_t,
\qquad t=p+1,\ldots,T.
\]
Let \(\tilde W_t=(\tilde X_{t-1},\ldots,\tilde X_{t-p})\). Then
\[
\begin{pmatrix}
\mathrm{vec}(\hat V_1')\\
\vdots\\
\mathrm{vec}(\hat V_p')
\end{pmatrix}
=
\left(
\sum_{t=p+1}^T
\tilde W_t' M_{\hat{\tilde\Lambda}} \tilde W_t
\right)^{-1}
\left(
\sum_{t=p+1}^T
\tilde W_t' M_{\hat{\tilde\Lambda}} \tilde Y_t
\right).
\]

\paragraph{Step 4: Iteration and factor recovery.}
Repeat Steps~1--3 until convergence of \(V_k\), \(k=1,\ldots,p\). After convergence, recover the factors by projecting the residual onto the estimated loading space, and map the loadings back to the original transformed coordinates:
\begin{align*}
  \hat f_t
  &=
  (CJ)^{-1}\hat{\tilde\Lambda}'
  \left(
  \tilde Y_t
  -
  \sum_{k=1}^p
  \tilde X_{t-k}\,\mathrm{Vec}(\hat V_k')
  \right),\\
  \hat\Lambda &= (\mathcal{K}')^{-1}\hat{\tilde\Lambda}.
\end{align*}

\subsection{Asymptotic Theory}
\label{sec:asympt}

We now study the large-sample properties of the estimator for the factor-adjusted VAR component.
Starting from the model in \eqref{eq:factorvar_mtrx}, we work with its whitened representation,
because the transformed coordinate space is equipped with the block-diagonal inner product
matrix \(\mathcal H=I_C\otimes H_J\), rather than the standard Euclidean one. Let \(\mathcal K=I_C\otimes K_J\),
where \(K_JK_J'=H_J\), and define
\[
\tilde Y_t = \mathcal K'Y_t,\qquad
\tilde\Lambda = \mathcal K'\Lambda,\qquad
\tilde X_t = \mathcal K'X_t,\qquad
\tilde\varepsilon_t = \mathcal K'\varepsilon_t.
\]
Then the model can be rewritten as
\begin{align}
\tilde Y_t
=
\tilde W_t\beta
+\tilde\Lambda f_t
+\tilde \varepsilon_t,
\label{eq:whitened_factorvar}
\end{align}
where \(\tilde W_t =(\tilde X_{t-1},\ldots,\tilde X_{t-p})\), and
\[
\beta=(\mathrm{Vec}'(V_1'),\ldots,\mathrm{Vec}'(V_p'))'.
\]
We derive a rate of convergence for the estimator of \(\beta\) and, under stronger conditions,
its asymptotic normality under the regime
\[
J,T \to \infty
\quad\text{with}\quad
C \text{ fixed}.
\]

Unlike \citet{Bai2009}, our regressors are generated by lagged dependent variables through the VAR structure, so the model does not satisfy the same exogeneity interpretation as a standard interactive fixed-effects regression. In particular, although the algebraic structure of the proof follows the Bai-type template, the dynamic nature of the regressors implies that certain weighted score terms involve internally generated regressors rather than externally given covariates. For this reason, the asymptotic theory requires additional conditions that are specific to the present factor-adjusted VAR setting.

A further difference from \citet{Bai2009} is that the asymptotic roles of \(N\) and \(T_0\) are reversed in our setting, because the PCA step is carried out along the time direction rather than the cross-sectional direction. As a result, the ratio conditions in the asymptotic theory are stated in terms of \(N/T_0\), rather than \(T_0/N\). Throughout this subsection, a superscript \(0\) denotes the true value of a parameter.

An important feature of the proposed framework is that stationarity is required only for the idiosyncratic VAR component, not for the observed series as a whole. The common factor component may exhibit nonstationary behavior, including trend-like movements, provided that the empirical second-moment condition in Assumption~B(i) holds. Thus, the model can accommodate multivariate density-valued time series with pervasive nonstationary common movements, while still allowing inference on directed dependence in the stable idiosyncratic component.

Because PCA is carried out along the time direction in our setting, the resulting projection and expansion arguments generate weighted score terms indexed by time averages of lagged regressors. In the Bai setting, the corresponding terms can be handled by exogeneity of the regressors, but that argument is no longer available here because the regressors are internally generated by the VAR dynamics. For this reason, Assumption~C(vi) is introduced to control such weighted score terms, while Assumption~E provides the additional limit and central limit conditions needed for asymptotic normality. These conditions are stronger and more model-specific than those used in the standard Bai setting, reflecting the fact that the present framework combines latent common factors with an idiosyncratic VAR structure rather than an exogenous regression design.

\paragraph{Regularity conditions.}

\paragraph{Assumption A (Regression design and identification).}
\begin{enumerate}[(i)]
\item For the true \(C\times C\) VAR coefficient matrices \(V_i^0\), \(i=1,\ldots,p\), there exists a constant \(\kappa\in(0,1)\), independent of \(J\) and \(T\), such that
\[
\sum_{i=1}^p \|V_i^0\| \le \kappa,
\]
where \(\|\cdot\|\) denotes the operator norm.

\item Let
\[
\mathcal F = \{\Lambda\in\mathbb{R}^{CJ\times r}:(CJ)^{-1}\Lambda'\Lambda=I_r\}.
\]
Then
\[
\inf_{\tilde \Lambda\in\mathcal F}\lambda_{\min}(D(\tilde \Lambda))\ge c>0,
\]
where \(\lambda_{\min}(D)\) denotes the minimum eigenvalue of \(D\), and
\begin{align*}
D(\tilde \Lambda)
&=
\frac{1}{CJ(T-p)}\sum_{t=p+1}^T Z_t(\tilde \Lambda)'Z_t(\tilde \Lambda),\\
Z_t(\tilde \Lambda)
&=
M_{\tilde \Lambda}\tilde W_t
-
\frac{1}{T-p}\sum_{s=p+1}^T
a_{ts}\,M_{\tilde \Lambda}\tilde W_s,\\
a_{ts}
&=
f_t^{0\prime}\Big(\frac{1}{T-p}F_0'F_0\Big)^{-1}f_s^0,
\end{align*}
with
\[
M_{\tilde \Lambda}=I_{CJ}-(CJ)^{-1}\tilde \Lambda\tilde \Lambda',
\qquad
F_0=(f^0_{p+1},\ldots,f^0_{T})'.
\]
\end{enumerate}

\paragraph{Assumption B (Factors and loadings).}
Let \(r\) be fixed. Treat the common factors \(\{f_t^0\}_{t=p+1}^T\) and the loading matrix
\[
\tilde\Lambda_0=(\tilde\lambda_{01},\ldots,\tilde\lambda_{0N})' \in \mathbb R^{N\times r},
\qquad N=CJ,
\]
as non-stochastic arrays.

\begin{enumerate}[(i)]
\item \textbf{Factor second-moment limit.}
\[
\frac{1}{T_0}\sum_{t=p+1}^T f_t^0 f_t^{0\prime} \to \Sigma_f,
\qquad T_0:=T-p,
\]
where \(\Sigma_f\) is positive definite. Moreover,
\[
\sup_{p+1\le t\le T}\|f_t^0\| \le K .
\]

\item \textbf{Loading second-moment limit.}
\[
\frac{1}{N}\tilde\Lambda_0'\tilde\Lambda_0 \to \Sigma_\Lambda,
\]
where \(\Sigma_\Lambda\) is positive definite. Moreover,
\[
\sup_{1\le i\le N}\|\tilde\lambda_{0i}\| \le K .
\]
\end{enumerate}

\noindent
\textit{Remark.}
Assumption B does not require the factor path to be stationary in the usual
time-series sense. It only requires the empirical second-moment limit
\[
\frac{1}{T_0}\sum_{t=p+1}^T f_t^0 f_t^{0\prime} \to \Sigma_f
\]
to exist. In particular, under an infill interpretation, trend-like common components
are allowed as long as this empirical second-moment limit exists.

\paragraph{Assumption C (Moments and dependence).}
Let \(\{\tilde\varepsilon_t\}_{t=p+1}^T\) be the idiosyncratic error sequence in
\[
\tilde Y_t=\tilde W_t\beta_0+\tilde\Lambda_0 f_t^0+\tilde\varepsilon_t.
\]
Assume:

\begin{enumerate}[(i)]
\item \textbf{Moments and scale.}
There exists \(\delta>0\) such that
\[
\sup_t\max_{1\le i\le N}E|\tilde\varepsilon_{it}|^{4+\delta}\le K<\infty.
\]

\item \textbf{No serial correlation.}
The sequence \(\{\tilde\varepsilon_t\}\) is independent across \(t\), and
\[
E(\tilde\varepsilon_t)=0,
\qquad
E(\tilde\varepsilon_t\tilde\varepsilon_t')=\Sigma_{\varepsilon,t}.
\]

\item \textbf{Uniformly bounded cross-sectional covariance.}
\[
\sup_t \|\Sigma_{\varepsilon,t}\|\le K<\infty.
\]

\item \textbf{Weak cross-sectional dependence.}
For each \(t\),
\[
\max_{1\le i\le N}\sum_{j=1}^N \big|\Sigma_{\varepsilon,t}(i,j)\big|\le K<\infty.
\]

\item \textbf{Linear-form moment bound.}
For any nonstochastic unit vectors \(u,v\in\mathbb R^N\),
\[
\sup_t E|u'\tilde\varepsilon_t|^{4+\delta}\le K,
\qquad
\sup_t E\big|(u'\tilde\varepsilon_t)(v'\tilde\varepsilon_t)\big|^2\le K .
\]

\item \textbf{Uniform covariance summability for weighted score terms.}
There exists \(K<\infty\) such that, for any deterministic uniformly bounded array
\(\{c_{ts}\}_{t,s=p+1}^T\),
\[
\frac{1}{NT_0}
\sum_{t,u=p+1}^T
\left|
E\!\left[
\Big(\sum_{s=p+1}^T c_{ts}\tilde W_s'\tilde\varepsilon_t\Big)
\Big(\sum_{v=p+1}^T c_{uv}\tilde W_v'\tilde\varepsilon_u\Big)
\right]
\right|
\le K .
\]
\end{enumerate}

\noindent
\textit{Remark.}
Assumption~C(vi) is imposed only to control the weighted score term
\[
\frac{1}{\sqrt{NT_0}}
\sum_{t=p+1}^T
\Big\{
\frac{1}{T_0}\sum_{s=p+1}^T c_{ts}\tilde W_s'
\Big\}\tilde\varepsilon_t,
\]
which appears in the proof of Theorem~\ref{thm:rate_beta}. The difficulty is that the weighted regressor
\[
b_t:=\frac{1}{T_0}\sum_{s=p+1}^T c_{ts}\tilde W_s
\]
contains both past and future values of \(\tilde W_s\), so it is not measurable with respect to the information set at time \(t-1\). Therefore, the mean-independence condition in Assumption~D is not sufficient by itself to control the second moment of \(b_t'\tilde\varepsilon_t\). Assumption~C(vi) is introduced as a convenient sufficient condition ensuring uniform summability of the covariance of such weighted score terms. It is stronger than what may be minimally necessary, but it avoids the need to impose more elaborate dependence conditions.
This is precisely the point at which the present time-direction PCA argument differs from the standard Bai setting with exogenous regressors, where the corresponding weighted score term can be handled by exogeneity.

\paragraph{Assumption D (Sequential exogeneity).}
For each \(t=p+1,\ldots,T\),
\[
E(\tilde\varepsilon_t \mid \{\tilde Y_s:s<t\})=0 .
\]

\paragraph{Assumption E (Probability limits and a score CLT).}
Assumption~E collects the additional probability-limit and central-limit conditions needed for asymptotic normality once the weighted score terms have been controlled.
For the asymptotic normality result in Theorem~\ref{thm:asy_normal_beta}, assume additionally that
\(N/T_0\to 0\). Let \(N:=CJ\) and \(T_0:=T-p\), and let \(Z_t(\tilde\Lambda)\) be defined in
Assumption~A(ii). Define
\[
H_0:=\frac{1}{N}\tilde\Lambda_0'\tilde\Lambda_0,
\qquad
\bar\Lambda_0:=\tilde\Lambda_0 H_0^{-1/2},
\]
so that \(\bar\Lambda_0\in\mathcal F\) and
\[
P_{\bar\Lambda_0}
=
\frac{1}{N}\bar\Lambda_0\bar\Lambda_0'
=
\tilde\Lambda_0(\tilde\Lambda_0'\tilde\Lambda_0)^{-1}\tilde\Lambda_0'
=
\Pi_0.
\]
Write
\[
Z_t^0 := Z_t(\bar\Lambda_0),
\qquad
\Sigma_{\varepsilon,t}:=E(\tilde\varepsilon_t\tilde\varepsilon_t').
\]
Assume:

\begin{enumerate}[(i)]
\item \textbf{Probability limits of \(D(\bar\Lambda_0)\) and \(\Omega(\bar\Lambda_0)\).}
The following probability limits exist:
\[
D(\bar\Lambda_0)
:=
\frac{1}{NT_0}\sum_{t=p+1}^T Z_t^{0\prime} Z_t^0
\xrightarrow{p} D_0,
\qquad
\Omega(\bar\Lambda_0)
:=
\frac{1}{NT_0}\sum_{t=p+1}^T Z_t^{0\prime}\Sigma_{\varepsilon,t} Z_t^0
\xrightarrow{p} \Omega_0,
\]
where \(D_0\) is positive definite.

\item \textbf{Central limit theorem for the score.}
\[
\frac{1}{\sqrt{NT_0}}\sum_{t=p+1}^T Z_t^{0\prime}\tilde\varepsilon_t
\xrightarrow{d}
\mathcal N(0,\Omega_0).
\]
\end{enumerate}

We first derive a convergence rate under 
$N/T_0\to\rho \geq 0$,
and then establish asymptotic normality under the stronger condition 
$N/T_0 \to 0$.
\begin{theorem}[Rate of convergence of \(\hat\beta\)]
\label{thm:rate_beta}
Consider the model in \eqref{eq:whitened_factorvar},
\[
\tilde Y_t
=
\tilde W_t\beta_0+\tilde\Lambda_0 f_t^0+\tilde\varepsilon_t,
\qquad t=p+1,\ldots,T,
\]
and let \((\hat\beta,\hat\Lambda)\) minimize
\[
Q(\beta,\tilde\Lambda)
=
\sum_{t=p+1}^T
(\tilde Y_t-\tilde W_t\beta)'M_{\tilde\Lambda}
(\tilde Y_t-\tilde W_t\beta)
\]
over \((\beta,\tilde\Lambda)\in\mathbb R^{pC^2}\times\mathcal F\), where
\[
\mathcal F=\{\Lambda\in\mathbb R^{N\times r}:N^{-1}\Lambda'\Lambda=I_r\},\quad
M_{\tilde\Lambda}=I_N-N^{-1}\tilde\Lambda\tilde\Lambda',\quad
N:=CJ,\ T_0:=T-p.
\]
Suppose Assumptions~A--D hold, with \(p\), \(r\), and \(C\) fixed. If \(N\to\infty\), \(T_0\to\infty\), and
\[
\frac{N}{T_0}\to \rho\geq 0,
\]
then
\[
\sqrt{NT_0}\,(\hat\beta-\beta_0)=O_p(1).
\]
\end{theorem}

\begin{theorem}[Asymptotic normality of \(\hat\beta\)]
\label{thm:asy_normal_beta}
Let \(N:=CJ\) and \(T_0:=T-p\), and suppose Assumptions~A--E hold, with \(r\) and \(C\) fixed. Assume additionally that
\[
\frac{N}{T_0}\to 0.
\]
Then
\[
\sqrt{NT_0}\,(\hat\beta-\beta_0)
\xrightarrow{d}
\mathcal N\!\left(0,\ D_0^{-1}\Omega_0 D_0^{-1}\right).
\]
\end{theorem}

Theorem~\ref{thm:rate_beta} establishes a convergence rate, and hence consistency of the estimator, while Theorem~\ref{thm:asy_normal_beta} provides its asymptotic normality under the stronger condition \(N/T_0\to0\). The main difficulty in the asymptotic analysis is Proposition~\ref{prop:beta_linear_expansion}, which derives a linear expansion for \(\hat\beta\) in the presence of latent factors and internally generated VAR regressors. Although the overall proof follows the Bai-type projection-and-expansion strategy, the present setting is not a routine extension of \citet{Bai2009}. Because the regressors are lagged dependent variables and PCA is carried out along the time direction, the expansion produces a weighted score term involving time averages of lagged regressors. This term cannot be handled by the standard exogeneity argument used in the Bai setting. Assumption~C(vi) is introduced precisely to control this weighted score contribution. The resulting expansion contains a leading score term and a bias term. The bias term is of order \(O_p(\sqrt{N/T_0})\), so it affects the rate result in Theorem~\ref{thm:rate_beta} when \(N/T_0\to\rho\ge0\), but becomes asymptotically negligible under the stronger condition \(N/T_0\to0\), which yields the asymptotic normality result in Theorem~\ref{thm:asy_normal_beta}. Full proofs and auxiliary technical lemmas are provided in the online Appendix.

\paragraph{Remark on feasible inference.}
Theorem~\ref{thm:asy_normal_beta} implies that
\[
\widehat{\operatorname{Var}}(\hat\beta)
=
\frac{1}{NT_0}\hat D^{-1}\hat\Omega \hat D^{-1},
\]
so coefficient-wise inference may be based on the corresponding standard errors, where
\[
\hat D
:=
\frac{1}{NT_0}\sum_{t=p+1}^T \hat Z_t'\hat Z_t,
\qquad
\hat Z_t
:=
Z_t(\hat\Lambda),
\]
and
\[
\hat\Omega
:=
\frac{1}{NT_0}\sum_{t=p+1}^T \hat Z_t'\hat\Sigma_{\varepsilon,t}\hat Z_t.
\]
Here \(\hat Z_t\) denotes the feasible analogue of \(Z_t(\bar\Lambda_0)\), obtained by replacing \(\bar\Lambda_0\) with \(\hat\Lambda\) and \(a_{ts}\) with
\[
\hat a_{ts}
:=
\hat f_t'
\Big(\frac{1}{T_0}\hat F'\hat F\Big)^{-1}
\hat f_s,
\qquad
\hat F:=(\hat f_{p+1},\ldots,\hat f_T)'.
\]
\(\hat\Sigma_{\varepsilon,t}\) denotes a feasible estimator of the contemporaneous covariance matrix of the idiosyncratic error.
In finite samples, a fully unrestricted estimator of \(\hat\Omega\) may be singular or numerically unstable when \(T_0\) is not sufficiently large relative to the dimension of the projected score vectors. In the simulation and empirical analyses, we therefore estimate \(\hat\Sigma_{\varepsilon,t}\) by a block-diagonal matrix with \(J\times J\) blocks, one for each region. This block-diagonal approximation yields a stable and nonsingular feasible covariance estimator for coefficient-wise \(t\)-tests.

\section{Simulation study}\label{sec:simulation}

This section examines whether increasing the number of estimated factors mechanically eliminates the idiosyncratic VAR network in our procedure. This issue is important for interpreting the Ct-value application, because the detected network may weaken or disappear as the factor dimension is varied. The simulation study shows that, when genuine idiosyncratic VAR dependence is present, increasing the number of estimated factors does not mechanically erase the directed network. Rather, it mainly reduces spuriously detected edges induced by common factor movements and sharpens the remaining edge set.

\subsection{Data-generating process}\label{subsec:simulation_dgp}

We generate data from a factor-augmented density-valued VAR model of the form
\begin{equation}
Y_t = (V_{\mathrm{true}} \otimes I_J) Y_{t-1} + L f_t + \varepsilon_t,
\label{eq:sim_dgp_main}
\end{equation}
where $Y_t \in \mathbb{R}^{CJ}$, $V_{\mathrm{true}}$ is a $C \times C$ coefficient matrix, $I_J$ is the $J \times J$ identity matrix, $L$ is the loading matrix, and $f_t$ is an $r^\ast$-dimensional factor process.

In the simulation, we set $T=114$, $C=20$, and $J=15$, so that $Y_t$ has dimension $CJ=300$. The true number of common factors is fixed at $r^\ast=5$. The factor process follows
\begin{equation}
f_t = A f_{t-1} + u_t,
\label{eq:sim_factor_var}
\end{equation}
where
\[
A=\mathrm{diag}(0,0,0.9,0.9,0.9),
\]
and $u_t$ is Gaussian with independent components and standard deviations
\[
(1.0,\,1.0,\,0.3,\,0.3,\,0.3).
\]
Thus, the first two factors are white noise, whereas the remaining factors are persistent.

To construct the loading matrix, let $Q=(q_1,\ldots,q_{r^\ast})$ be a $C\times r^\ast$ orthonormal matrix obtained by applying a QR decomposition to a Gaussian random matrix, and let $U=(u_1,\ldots,u_{r^\ast})$ be a $J\times r^\ast$ matrix whose columns are independently generated Gaussian vectors normalized to unit length. The loading matrix $L=(\ell_1,\ldots,\ell_{r^\ast}) \in \mathbb{R}^{CJ\times r^\ast}$ is defined by
\[
\ell_1 = \left(\mathbf{1}_C/\sqrt{C}\right)\otimes u_1,
\qquad
\ell_k = q_k \otimes u_k,\quad k=2,\ldots,r^\ast,
\]
where $\mathbf{1}_C$ denotes the $C$-vector of ones. Hence, the first factor has a global loading shared by all units, whereas the remaining factors generate structured cross-sectional heterogeneity.

The true VAR matrix $V_{\mathrm{true}}$ is constructed from
\[
M = Q \,\mathrm{diag}(1.00,0.90,0.80,0.70,0.60)\, Q^\prime.
\]
We restrict attention to off-diagonal entries in rows indexed by the source set $\{1,2,5\}$, rank these candidates by $|M_{ij}|$, and select the top 30 entries. For each selected pair $(i,j)$, we set
\[
(V_{\mathrm{true}})_{ij}=\alpha_V |M_{ij}|,
\]
and set all remaining entries to zero. If the spectral radius of $V_{\mathrm{true}}$ exceeds $0.95$, we rescale the entire matrix so that its spectral radius becomes $0.95$. We consider three levels of VAR strength, $\alpha_V \in \{0, 0.5, 1.0\}$. When $\alpha_V=0$, there is no idiosyncratic VAR dependence. When $\alpha_V>0$, the data contain genuine directed edges in the idiosyncratic VAR component.

The idiosyncratic disturbance is Gaussian,
\[
\varepsilon_t \sim N(0,\sigma_\varepsilon^2 I_{CJ}),
\qquad \sigma_\varepsilon=0.1.
\]
For each design, we generate 100 Monte Carlo replications.

\subsection{Estimation, edge selection, and results}\label{subsec:simulation_estimation}

For each simulated sample, we estimate the model using exactly the same procedure as in Section~\ref{sec:dfm}. We vary the number of estimated factors over
\[
r=0,1,\ldots,8,
\]
so that the estimated factor dimension may be smaller or larger than the true value $r^\ast=5$.

When $r=0$, we estimate a VAR without factors and compute block-diagonal robust $t$-statistics. More precisely, the covariance matrix of the error term is assumed to be block-diagonal across units, with each block given by a $J\times J$ covariance matrix.

When $r \geq 1$, we iteratively estimate the VAR coefficient matrix and the factor loading matrix by alternating between factor extraction from the residuals and projected least-squares estimation of the VAR component. Given the final estimator, we compute the robust $t$-statistics for the VAR coefficients
based on the asymptotic covariance formula in Section~\ref{sec:asympt}.
The asymptotic covariance matrix is estimated by
\[
(NT_0)^{-1}\hat D^{-1}\hat\Omega \hat D^{-1},
\]
where $\hat D$ and $\hat\Omega$ are sample analogues of the corresponding population quantities in Assumption~E. In estimating $\hat\Omega$, we assume that $\Sigma_{\varepsilon,t}$ has a block-diagonal structure,
\[
\mathrm{diag}
\bigl(
\varepsilon_{t,1}\varepsilon_{t,1}^\prime,\ldots,
\varepsilon_{t,C}\varepsilon_{t,C}^\prime
\bigr),
\]
where $\varepsilon_{t,c}$ denotes the $J$-dimensional idiosyncratic error vector for unit $c$ at time $t$.

Directed edges are selected from positive off-diagonal entries of the estimated VAR matrix using one-sided $t$-tests with the Benjamini--Yekutieli false discovery rate control at level 10\% \citep{BenjaminiYekutieli2001}. For each replication and each value of $r$, we record the number of detected edges, the recall rate relative to the true edge set, the empirical false discovery proportion (FDP), and the precision.
These four summaries allow us to assess whether increasing the estimated factor dimension mainly removes spurious edges induced by common factor movements, or instead erases genuine idiosyncratic VAR dependence.

\begin{figure}[tbp]
\centering
\includegraphics[width=0.9\textwidth]{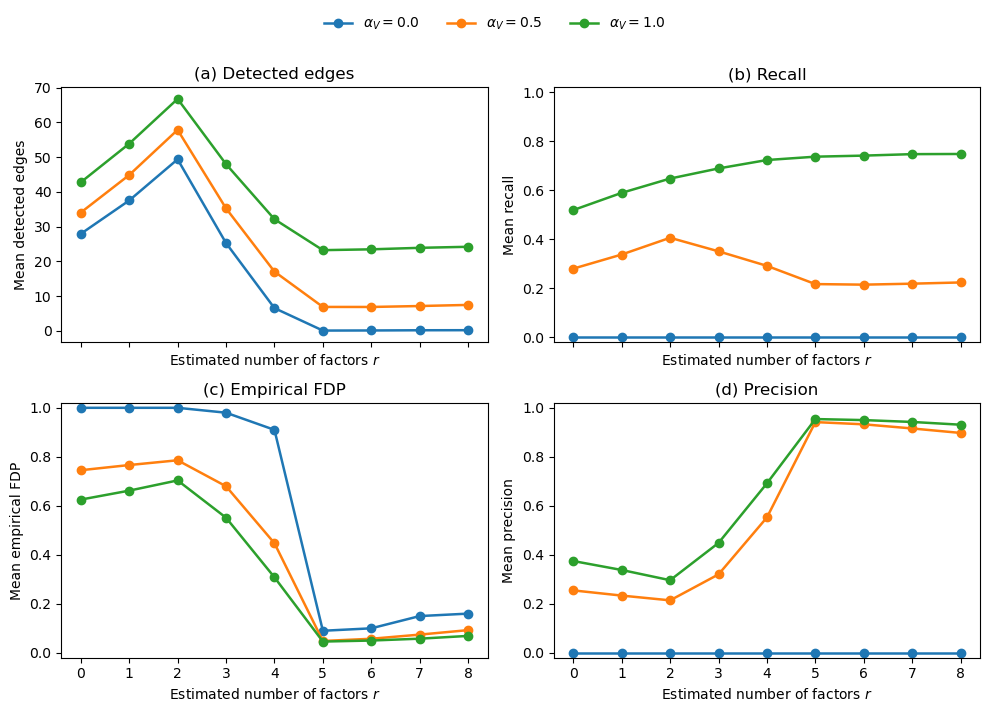}
\caption{Simulation results as functions of the estimated number of factors $r$. Panel (a) reports the mean number of detected edges, panel (b) the mean recall, panel (c) the mean empirical false discovery proportion (FDP), and panel (d) the mean precision. The three curves correspond to $\alpha_V=0$, $0.5$, and $1.0$.}
\label{fig:sim_edges_recall}
\end{figure}

Figure~\ref{fig:sim_edges_recall} reports four summaries of edge-selection performance as functions of the estimated number of factors $r$: the mean number of detected edges, the mean recall, the mean empirical FDP, and the mean precision. The results reveal a clear transition from many but noisy detected edges at small $r$ to fewer but more credible edges once sufficiently many factors are included.

When $\alpha_V=0$, recall is identically zero by construction, so any detected edge is spurious. In this case, the mean number of detected edges is large for small values of $r$, peaking around $r=2$, while the empirical FDP is essentially one. Once $r$ reaches about the true factor dimension, however, the mean number of detected edges falls to nearly zero. Thus, when the idiosyncratic VAR component is absent, apparent directed edges generated by common factor movements largely disappear after sufficient factor adjustment.

When $\alpha_V=0.5$, the same reduction in the number of detected edges is accompanied not by the disappearance of the true signal, but by fewer false discoveries and higher precision. As $r$ increases beyond 4, the mean number of detected edges drops sharply, but recall remains positive even for $r>r^\ast$. At the same time, the empirical FDP declines rapidly and precision rises to a high level. This shows that, under moderate idiosyncratic VAR dependence, increasing the number of estimated factors mainly removes false positives driven by common factor movements, while preserving a nontrivial fraction of genuine directed edges.

When $\alpha_V=1.0$, the pattern is even clearer. Although the mean number of detected edges again declines as $r$ increases, recall remains high and continues to improve slightly for large values of $r$. At the same time, the empirical FDP becomes small and precision becomes very high once $r$ reaches about 5. Thus, when the true idiosyncratic VAR signal is strong, estimating more factors does not destroy the directed network; instead, it sharpens it by removing common-factor-driven spurious edges while retaining most of the genuine ones.

Overall, the simulation study shows that increasing the number of estimated factors does not mechanically eliminate genuine idiosyncratic VAR edges. What large values of $r$ mainly do is reduce false discoveries and sharpen the selected network. Therefore, if detected edges disappear in the Ct-value application, this should not be viewed as a mechanical consequence of factor augmentation itself. Rather, it suggests that the apparent network at small $r$ may be driven mainly by common factor movements, or that the idiosyncratic VAR component is weak or absent.
\section{Empirical data analysis}
\label{sec:dt}

The real-world dataset analyzed in this study consists of raw polymerase chain reaction (PCR) test results, including SARS-CoV-2 cycle threshold (Ct) values, obtained from a major clinical diagnostics laboratory in Brazil. The Ct value is the number of amplification cycles required for the viral signal to cross a detection threshold, so lower Ct values correspond to higher viral load. Temporal and regional changes in the Ct distribution can therefore be interpreted as changes in the distribution of viral load among tested individuals.

In total, the dataset comprises approximately 343{,}000 individual PCR test observations collected between 20 March 2020 and 22 May 2022. Although specimens were collected from more than 2{,}000 sites across the country, all laboratory analyses were performed at a centralized facility in S\~ao Paulo. A key feature of the dataset is therefore that cross-regional differences in the observed Ct distributions are less likely to reflect systematic differences in laboratory processing across locations. In addition, two PCR assays were used by the central laboratory (Seegene and Taq), and the laboratory confirmed that the resulting Ct values are comparable across assays.

We define the week index by
\[
t = 1+\lfloor (\text{date}-\text{start date})/7 \rfloor,
\]
where the start date is 16 March 2020. We take Monday, 16 March 2020, as the reference date to define week boundaries, even though the first PCR observations enter from 20 March 2020. This yields \(T=114\) weekly observations in total. Figure~\ref{fig:weekly_summary} reports basic weekly summaries of the raw Ct observations: the weekly mean Ct value and the 5th and 95th percentiles computed from all observations nationwide, together with the weekly sample size. The 5th--95th percentile band of weekly Ct values is broadly stable over time, although it narrows around early 2022, driven mainly by an upward shift in the lower tail, even as the weekly sample size varies substantially.

\begin{figure}[t]
  \centering
  \includegraphics[width=0.95\linewidth]{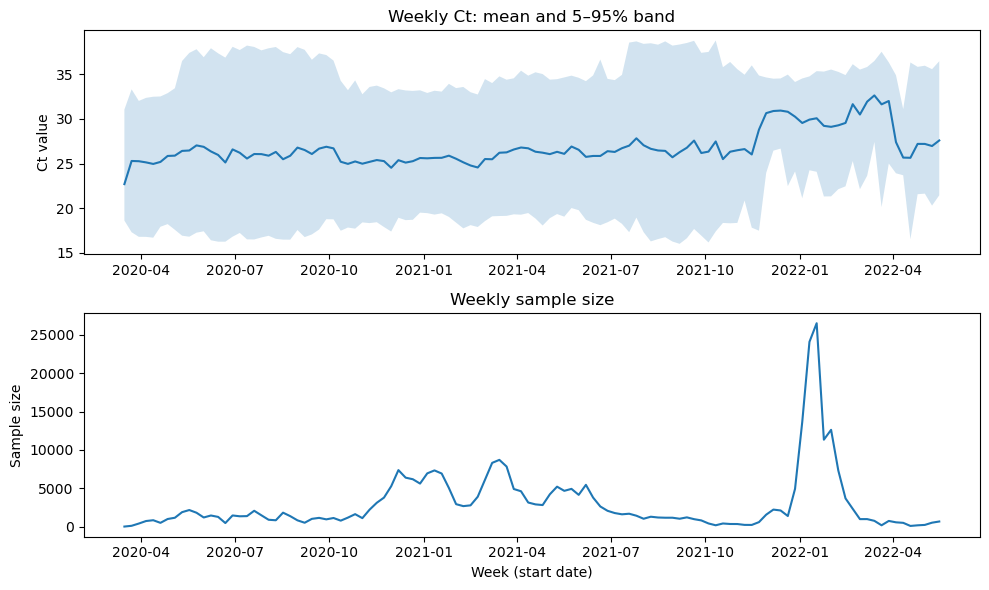}
  \caption{Weekly summaries of the raw Ct observations. The top panel reports the weekly mean Ct value together with the 5th and 95th percentiles, computed from nationwide observations. The bottom panel shows the weekly sample size. Both panels are indexed by the week start date.}
  \label{fig:weekly_summary}
\end{figure}

Our empirical analysis follows the construction and modeling framework developed in Section~\ref{sec:met}. We first aggregate raw Ct observations into a weekly sequence of region-level density estimates, represented by B-spline mixture weights. We then apply the generalized logit transform to obtain the vector-valued series \(Y_{t,c}\) and fit the factor-adjusted VAR model developed in Section~\ref{sec:met}. The resulting VAR coefficients are summarized as a directed network that captures lagged dependence in the evolution of Ct-value distributions across regions.

The main empirical issue is not merely whether one directed network can be drawn from the data, but under what conditions an idiosyncratic network becomes visible after strong common movements have been removed. In particular, both the sample period and the prior strength used in weekly density estimation affect how clearly the idiosyncratic VAR component can be separated from pervasive common factor movements. The empirical analysis below is therefore organized around this visibility question.

\subsection{Construction of weekly Ct-value distributions}
\label{sec:data}

Each PCR observation is indexed by the date of sample collection and the reported municipality (city) and state. 
To connect city-level observations to the region-level model in Section~\ref{sec:met}, we construct \(C=20\) geographical regions by spatial clustering of Brazilian cities.
We first apply spatial \(k\)-means clustering to city coordinates (latitude and longitude) with a deliberately large initial number of clusters (\(K=60\)). We then iteratively merge clusters whose average weekly sample size falls below 30 tests per week, always merging the smallest cluster into its nearest neighboring cluster in terms of centroid distance. This procedure yields \(C=20\) geographically coherent regions with more balanced weekly observation counts, which helps avoid extremely sparse cells in the weekly density estimation. 
The resulting spatial partition
is shown in Figure~\ref{fig:cluster_map}. Table~\ref{tab:summary_by_rep} reports summary statistics of the weekly sample size \(n_{t,c}\) for each region, labeled by its representative city.

\begin{figure}[t]
  \centering
  \includegraphics[width=0.9\linewidth]{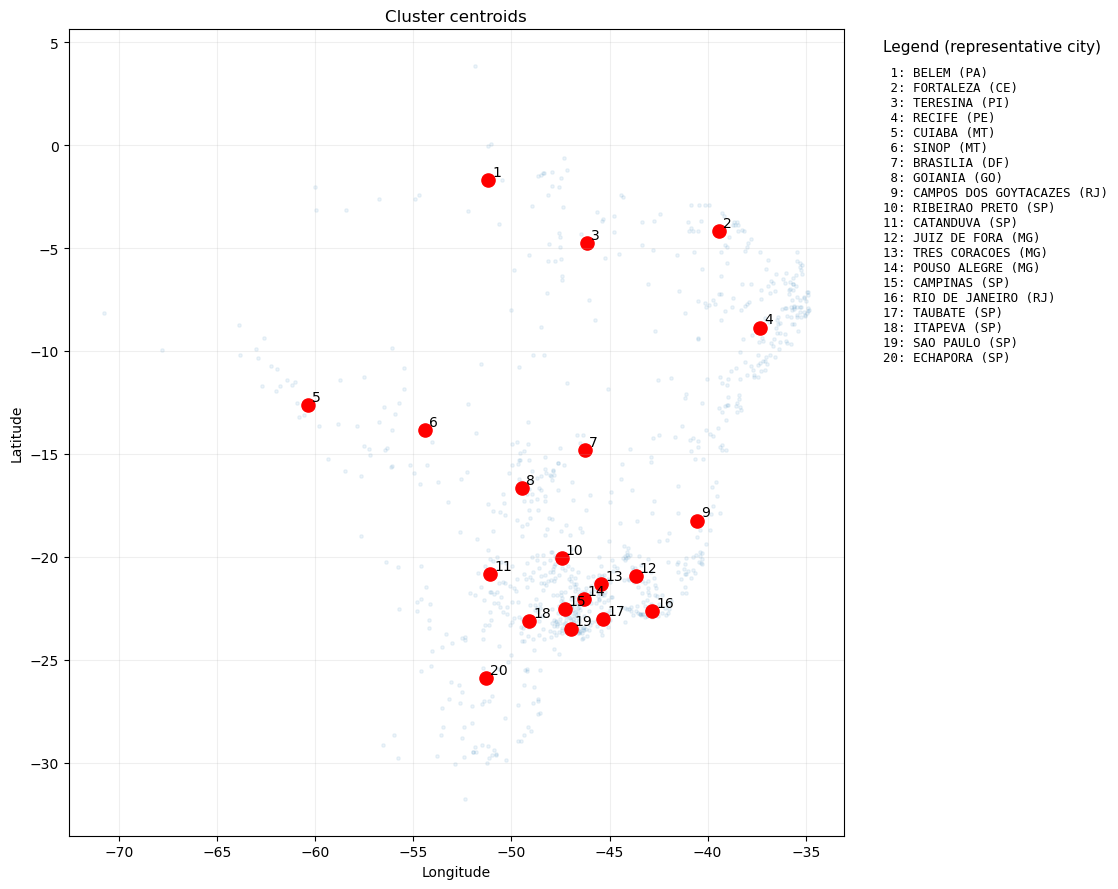}
  \caption{Map of the 20 regions in Brazil. Red dots indicate region centroids, and the numeric labels (ID 1--20) follow the north-to-south ordering. The legend lists each region ID and its representative city.}
  \label{fig:cluster_map}
\end{figure}

\begin{table}[htbp]
\centering
\caption{Summary statistics of weekly sample size by representative area}
\label{tab:summary_by_rep}
\begin{tabular}{r l r r r r r r r r}
\hline
ID & Rep. City & count & mean & std & min & 25\% & 50\% & 75\% & max \\
\hline
1  & Belem                 & 114 &  73.8 & 149.9 & 0 &   3.00 &  10.0 &   34.75 &  623 \\
2  & Fortaleza             & 114 &  43.9 &  67.6 & 0 &   0.25 &  17.0 &   58.00 &  396 \\
3  & Teresina              & 114 &  44.7 &  80.1 & 0 &   3.00 &  11.0 &   49.25 &  414 \\
4  & Recife                & 114 & 255.5 & 599.4 & 0 &  19.25 & 100.0 &  235.50 & 4640 \\
5  & Cuiaba                & 114 & 142.1 & 291.8 & 0 &  31.25 &  72.0 &  132.25 & 2372 \\
6  & Sinop                 & 114 &  45.2 &  58.1 & 0 &   8.25 &  23.5 &   60.00 &  384 \\
7  & Brasilia              & 114 & 157.6 & 225.5 & 0 &  37.75 &  98.0 &  165.00 & 1462 \\
8  & Goiania               & 114 &  57.1 &  89.8 & 0 &  11.25 &  35.5 &   58.25 &  716 \\
9  & Campos dos Goytacazes & 114 &  80.1 & 111.4 & 0 &  15.50 &  38.5 &   86.00 &  626 \\
10 & Ribeirao Preto        & 114 & 137.1 & 223.0 & 0 &  31.00 &  73.0 &  170.50 & 1528 \\
11 & Catanduva             & 114 & 110.9 & 236.0 & 0 &  12.00 &  33.0 &   96.75 & 1573 \\
12 & Juiz de Fora          & 114 &  94.6 & 131.6 & 0 &  15.00 &  34.5 &  119.50 &  675 \\
13 & Tres Coracoes         & 114 &  30.7 &  47.3 & 0 &   3.25 &  15.0 &   31.50 &  255 \\
14 & Pouso Alegre          & 114 & 143.3 & 364.1 & 0 &   8.25 &  28.0 &  133.75 & 2713 \\
15 & Campinas              & 114 & 271.0 & 336.4 & 1 &  62.25 & 141.0 &  358.75 & 1971 \\
16 & Rio de Janeiro        & 114 & 271.5 & 407.8 & 0 &  36.00 & 124.0 &  313.25 & 2988 \\
17 & Taubate               & 114 &  76.9 & 137.1 & 0 &   6.25 &  21.0 &   72.00 &  661 \\
18 & Itapeva               & 114 &  32.8 &  56.3 & 0 &   4.00 &  12.5 &   39.50 &  372 \\
19 & Sao Paulo             & 114 & 785.8 & 960.4 & 1 & 163.25 & 362.5 & 1127.75 & 5173 \\
20 & Echapora              & 114 & 150.6 & 281.9 & 0 &  12.25 &  30.0 &  127.00 & 1701 \\
\hline
\end{tabular}
\end{table}

Let \(x_{t,c,i}\) denote the \(i\)th individual Ct value observed in region \(c\) during week \(t\). For each region \(c\) and week \(t\), we collect all Ct observations observed in that cell and estimate the weekly region-level density on the fixed support \([10,40]\) by the B-spline mixture model in \eqref{eq:mixture}. In the empirical implementation, we use \(J+1\) normalized B-spline components and obtain the mixture weights \(w_{t,c}\in\mathcal S_{J+1}\) by the EM algorithm with a Dirichlet prior, as described in Section~\ref{sec:bspline}. This regularization is important in practice because weekly sample sizes vary substantially across regions and are close to zero in the lower quartiles for several regions.

More specifically, for each week we use the nationwide weekly density as a common prior and estimate region-specific weekly densities as posterior-mean B-spline weights. The prior strength is controlled by the scalar parameter \(\gamma\) in \eqref{eq:posterior}. A larger \(\gamma\) pulls region-level densities more strongly toward the nationwide weekly pattern, whereas a smaller \(\gamma\) leaves more room for region-specific variation. Since this choice turns out to matter for the visibility of the idiosyncratic VAR network, we compare several values of \(\gamma\) in the empirical analysis below.

We transform the estimated weight vectors into \(\mathbb R^{J}\) via the generalized logit map in Section~\ref{sec:logit}, \(Y_{t,c} = \mathrm{logit}_{\delta}(w_{t,c})\), and stack them as \(Y_t=(Y_{t,1}',\ldots,Y_{t,20}')'\). All subsequent estimation uses the block-diagonal metric \(\mathcal H = I_C\otimes H_J\) introduced in \eqref{eq:inner_product}. 
Equivalently, we work with the whitened series
\[
\tilde Y_t = \mathcal K' Y_t ,
\]
where \(\mathcal K=I_C\otimes K_J\) and \(\mathcal H=\mathcal K\mathcal K'\).
Unless stated otherwise, we fix the tuning parameter \(\delta=1\) and the spline specification \((J=15,k=3)\) throughout the empirical analysis for numerical stability and comparability across regions and weeks.

\subsection{Full-sample versus post-September 2020 results}
\label{sec:emp_period_prior}

We now examine how the detected cross-regional VAR network depends on the sample period and the prior strength used in the weekly density estimation. The main empirical question is whether the directed edges identified by our procedure reflect genuine idiosyncratic distributional dynamics, or whether they are masked by strong common factor movements and overshrinkage toward the nationwide weekly density.

To address this issue, we compare two sample periods. The first uses the full sample from 16 March 2020 onward. The second starts on 28 September 2020 and excludes the first six months of the pandemic. We focus on this later subsample because the early phase appears to be dominated by strong nationwide common movements, whereas the later phase is expected to contain richer regional heterogeneity. In addition, weekly sample sizes become more stable after the early phase, which improves the reliability of the region-level density estimates.

\begin{figure}[t]
  \centering
  \includegraphics[width=0.90\linewidth]{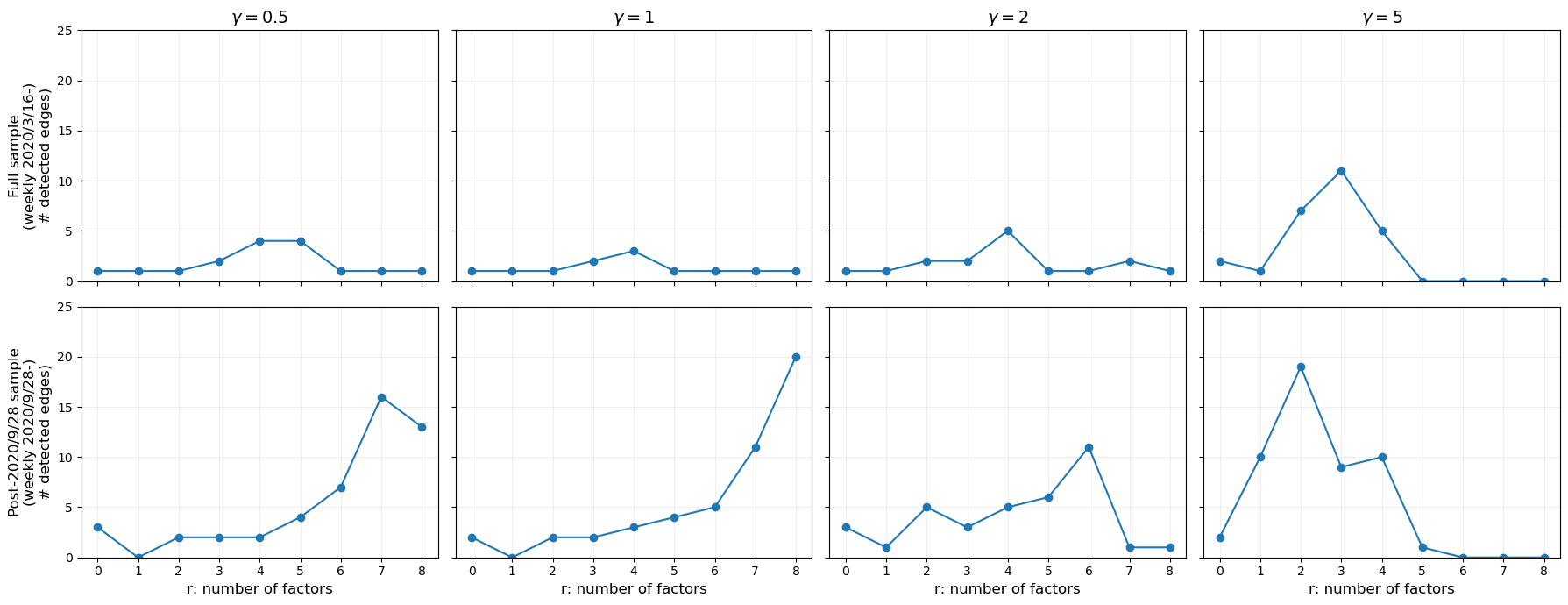}
  \caption{Number of detected directed edges as a function of the estimated number of factors \(r\). The top row uses the full sample from 2020/3/16 onward, and the bottom row uses the subsample from 2020/9/28 onward. The four columns correspond to \(\gamma=0.5,1,2,5\). Edges are selected from positive off-diagonal VAR coefficients using one-sided \(t\)-tests with Benjamini--Yekutieli false discovery rate control at the 5\% level.}
  \label{fig:edge_count_compare}
\end{figure}

Figure~\ref{fig:edge_count_compare} reports the number of detected directed edges as a function of the estimated number of factors \(r\), for four values of the prior-strength parameter, \(\gamma=0.5,1,2,5\). In each case, edges are selected from positive off-diagonal VAR coefficients using one-sided \(t\)-tests with Benjamini--Yekutieli false discovery rate control at the 5\% level \citep{BenjaminiYekutieli2001}.

The contrast between the two sample periods is sharp. In the full sample, the detected network is generally weak for all values of \(\gamma\). Under the strong prior \(\gamma=5\), the number of detected edges increases temporarily around \(r=2\) and \(r=3\), but then falls back to nearly zero for \(r\geq5\). This pattern is consistent with the simulation evidence in Section~\ref{sec:simulation}: apparent directed edges induced by strong common factor movements may remain visible when too few factors are included, but they disappear once enough factors are added. By contrast, in the post-2020/9/28 subsample, a substantial number of edges emerge under weak priors, especially for \(\gamma=0.5\) and \(\gamma=1\), when \(r\) is sufficiently large. This suggests that the early phase of the pandemic is more strongly dominated by common movements, whereas the later sample contains a more visible idiosyncratic VAR structure.

The prior-strength parameter also has a clear effect. Under weak priors (\(\gamma=0.5\) and \(\gamma=1\)), the later subsample yields a substantial directed network for larger values of \(r\). By contrast, stronger priors (\(\gamma=2\) and \(\gamma=5\)) largely suppress the detected network even in the later subsample. This is consistent with the idea that a large \(\gamma\) overshrinks the region-level weekly densities toward the nationwide weekly pattern and thereby removes region-specific variation needed to identify the idiosyncratic VAR component.

Taken together, Figure~\ref{fig:edge_count_compare} shows that two empirical choices are crucial for recovering an interpretable directed network from the Ct-value data: excluding the first six months of the sample and using a weak prior in the weekly density estimation. It also suggests that networks detected only at small values of \(r\) should be interpreted with caution, because they may still reflect insufficient removal of common movements. 
In the remainder of the empirical analysis, we therefore focus on the post-2020/9/28 subsample and use \(\gamma=1\) as our baseline specification. This choice is motivated by the preceding evidence: weak priors preserve region-specific variation, while sufficiently large values of \(r\) help remove strong common movements and make the idiosyncratic VAR network visible.

\subsection{Directed network in the post-2020/9/28 weak-prior baseline}
\label{sec:emp_network}

We next examine the directed network under the post-2020/9/28 subsample with the weak prior \(\gamma=1\), focusing on the large-\(r\) results for \(r=7\) and \(r=8\). As shown in Section~\ref{sec:emp_period_prior}, weak-prior specifications combined with sufficiently large values of \(r\) are the settings in which the idiosyncratic VAR network becomes visible after strong common movements are removed. Rather than singling out one specification as uniquely preferred, we compare two nearby large-\(r\) specifications and focus on the directed edges that remain selected in both of them.

\begin{figure}[ht]
  \centering
  \includegraphics[width=0.7\linewidth]{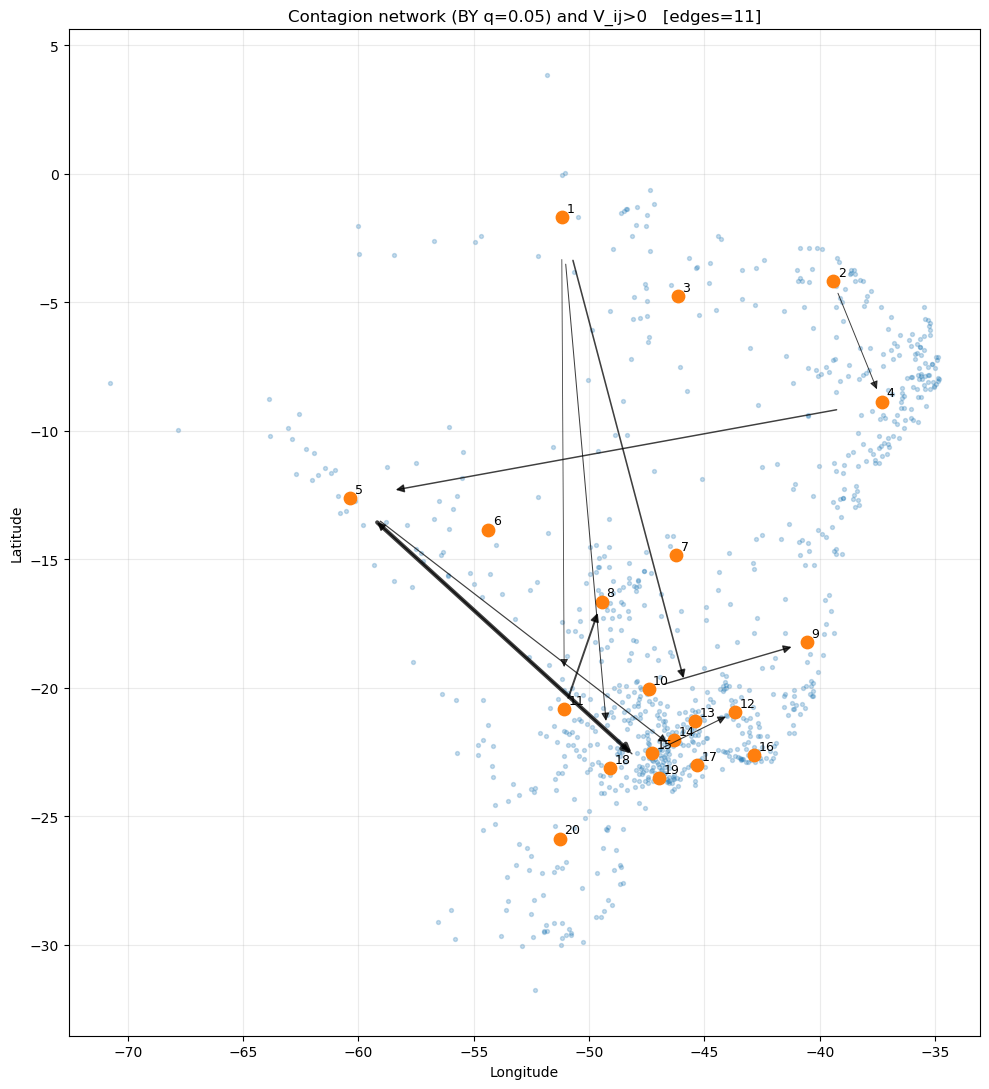}
  \caption{Directed network for the post-2020/9/28 subsample with \(\gamma=1\) and \(r=7\). All edges shown here remain selected in the \(r=8\) specification. Edges are selected from positive off-diagonal VAR coefficients using one-sided \(t\)-tests with Benjamini--Yekutieli false discovery rate control at the 5\% level.}
  \label{fig:baseline_network}
\end{figure}

Figure~\ref{fig:baseline_network} visualizes the network for \(r=7\), all of whose edges remain selected in the \(r=8\) specification, while Table~\ref{tab:common_edges_baseline}
lists the edges common to both specifications.
Each directed edge corresponds to a positive off-diagonal VAR coefficient selected by one-sided \(t\)-tests with Benjamini--Yekutieli false discovery rate control at the 5\% level. An arrow \(c' \to c\) indicates a Granger-predictive relation in the idiosyncratic VAR component \citep{Granger1969,Lutkepohl2005}: past distributional innovations in region \(c'\) help predict current innovations in region \(c\) after partialling out the common factor movements.

\begin{table}[t]
\centering
\caption{Directed edges selected in both \(r=7\) and \(r=8\) under the post-2020/9/28 weak-prior baseline \((\gamma=1)\). Region labels follow the representative-city convention in Table~\ref{tab:summary_by_rep}.}
\label{tab:common_edges_baseline}
\begin{tabular}{ccllrr}
\hline
Source ID & Target ID & Source rep. city & Target rep. city & Coef. (\(r=8\)) & \(t\)-stat (\(r=8\)) \\
\hline
5  & 19 & Cuiaba         & Sao Paulo             & 0.191 & 5.806 \\
11 & 8  & Catanduva      & Goiania               & 0.203 & 4.069 \\
1  & 13 & Belem          & Tres Coracoes         & 0.125 & 3.999 \\
4  & 5  & Recife         & Cuiaba                & 0.115 & 3.534 \\
10 & 9  & Ribeirao Preto & Campos dos Goytacazes & 0.171 & 3.805 \\
19 & 5  & Sao Paulo      & Cuiaba                & 0.237 & 3.783 \\
5  & 17 & Cuiaba         & Taubate               & 0.183 & 3.330 \\
15 & 12 & Campinas       & Juiz de Fora          & 0.226 & 3.734 \\
1  & 18 & Belem          & Itapeva               & 0.124 & 4.189 \\
2  & 4  & Fortaleza      & Recife                & 0.154 & 3.703 \\
1  & 11 & Belem          & Catanduva             & 0.080 & 3.691 \\
\hline
\end{tabular}
\end{table}

\begin{table}[t]
\centering
\caption{Outgoing-edge counts by source region under the post-2020/9/28 weak-prior baseline \((\gamma=1)\).}
\label{tab:sender_summary_baseline}
\begin{tabular}{cccc}
\hline
Source ID & Representative city & Outgoing at \(r=7\) & Outgoing at \(r=8\) \\
\hline
1  & Belem          & 3 & 3 \\
5  & Cuiaba         & 2 & 2 \\
2  & Fortaleza      & 1 & 2 \\
4  & Recife         & 1 & 2 \\
10 & Ribeirao Preto & 1 & 2 \\
11 & Catanduva      & 1 & 5 \\
15 & Campinas       & 1 & 1 \\
19 & Sao Paulo      & 1 & 2 \\
\hline
\end{tabular}
\end{table}

Two features stand out. 
First, the \(r=7\) network remains fully preserved in the broader \(r=8\) specification: all 11 edges selected at \(r=7\) remain selected at \(r=8\), with 9 additional edges. 
By contrast, the corresponding full-sample specifications yield only one selected edge in each case. This reinforces the conclusion of Section~\ref{sec:emp_period_prior}: once the first six months are excluded and the density prior is kept weak, the data contain a nontrivial idiosyncratic VAR network rather than a pattern driven solely by pervasive common movements.

Second, sender importance should be judged not only by the number of outgoing edges in a single specification, but also by whether those edges remain selected across nearby large-\(r\) specifications. The region represented by Belem has three outgoing edges at both \(r=7\) and \(r=8\), directed toward the regions represented by Tres Coracoes, Itapeva, and Catanduva. Although the representative city of region 1 is Belem, this northern region also contains Manaus. Thus, the persistent outgoing role of region 1 is broadly consistent with epidemiological accounts emphasizing the importance of the Manaus region during the spread of the Gamma (P.1) variant \citep{Faria2021Science,Sabino2021Lancet,Banho2022CommunMed}.

Cuiaba also appears as a stable sender. Its outgoing edge count is two at both \(r=7\) and \(r=8\), and its links to Sao Paulo and Taubate remain selected in both specifications. In addition, one of these links is among the strongest selected edges in the table. By contrast, the region represented by Catanduva has the largest outgoing count at \(r=8\), but not at \(r=7\). This shows that the source region with the largest number of outgoing edges in a single specification need not be the most robust sender across nearby large-\(r\) specifications.

Taken together, Figure~\ref{fig:baseline_network} and Tables~\ref{tab:common_edges_baseline}--\ref{tab:sender_summary_baseline} indicate that the post-2020/9/28 weak-prior baseline yields a network with multiple sender regions. Within this pattern, the northern region represented by Belem (including Manaus) and the inland region represented by Cuiaba stand out because their outgoing roles remain visible across both \(r=7\) and \(r=8\).

At the same time, the estimated network should be interpreted with caution. Our analysis does not model infections or case counts directly; instead, it studies lagged dependence in weekly Ct-value distributions. Hence, the detected edges should not be read as literal flows of infected individuals across regions. Rather, they capture predictive dependence in how the shapes of the observed viral-load distributions evolve across regions. Since Ct distributions depend on who is tested and when samples are collected, the network may reflect variation in testing composition and timing as well as transmission dynamics. Even so, the estimated network provides evidence of geographically structured idiosyncratic dependence in the evolution of Ct-value distributions across Brazil.

\subsection{Summary of empirical findings}
\label{sec:emp_summary}

The empirical analysis yields three main findings. First, when the full sample from March 2020 onward is used, the directed network implied by the idiosyncratic VAR component remains weak for all values of the prior-strength parameter considered here. Second, once the first six months are excluded and the sample starts on 28 September 2020, a substantial number of directed edges emerge under weak priors, especially for \(\gamma=0.5\) and \(\gamma=1\), when the number of factors is sufficiently large. Third, under the post-2020/9/28 weak-prior baseline with \(\gamma=1\), the large-\(r\) results reveal a network with multiple sender regions, among which the northern region represented by Belem (including Manaus) and the inland region represented by Cuiaba remain visible across both \(r=7\) and \(r=8\).

Taken together, these results indicate that the Ct-value data contain a visible idiosyncratic VAR structure only after the early pandemic phase is excluded and the density prior is kept weak enough to preserve region-specific variation. Thus, the disappearance of detected edges in some specifications should not be viewed as a mechanical failure of the factor-adjusted VAR framework, but rather as a consequence of strong common movements and excessive shrinkage in the density estimation step.
\section{Conclusion}
\label{sec:conc}
This paper proposed a density-valued VAR model with latent factors for multivariate time series of density functions observed across units.
The new method combines B-spline mixture density estimation, a generalized logit transformation, an isometric inner-product structure for the transformed coordinates, latent factor decomposition, and cross-unit VAR dependence in a unified framework. This construction allows the separation between strong common movements and directed idiosyncratic distributional dynamics, while summarizing the latter as a Granger-predictive network.

On the theoretical side, we established large-sample properties of the factor-adjusted VAR estimator in the transformed Euclidean representation. Under a joint asymptotic regime in which the basis dimension and the time dimension diverge with the number of regions fixed, we derived its rate of convergence and, under stronger conditions, asymptotic normality for coefficient-wise inference. These results provide a formal basis for testing directed lagged dependence after removing latent common factors.

The simulation study showed that increasing the number of estimated factors does not mechanically eliminate a genuine idiosyncratic VAR network. When the true VAR component is absent, spuriously detected edges largely disappear once enough factors are included. When genuine idiosyncratic dependence is present, however, recall remains positive even for large values of the estimated factor dimension, while false discoveries decline sharply. This distinction is important for interpreting empirical results in settings with strong common movements.

In the empirical application to weekly regional distributions of SARS-CoV-2 Ct values in Brazil, the detected network depended crucially on both the sample period and the strength of the density prior. In the full sample, only a weak directed network was detected. By contrast, once the first six months of the pandemic were excluded and the nationwide weekly prior was kept weak, a substantial idiosyncratic network became visible. Under the post-2020/9/28 weak-prior baseline, the large-\(r\) results revealed stable outgoing links from the northern region represented by Belem, which includes Manaus, and from Cuiab\'a toward southeastern metropolitan areas.

These findings show that the proposed framework can recover directed idiosyncratic dependence in multivariate time series of density functions after strong common movements are removed. 
More broadly, the paper provides a tractable approach to analyzing time series of distributions when common factors and directed idiosyncratic dependence are both present. 
Although the empirical application treats regions as the observational units, the same framework can also be applied when the units represent different variables, demographic groups, or other collections of density-valued observations over time.
A natural next step is to extend the present VAR structure beyond scalar coefficients multiplying the transformed density coordinates as a whole, for example by allowing component-specific effects or more general cross-component interactions in the transformed space. Such extensions would broaden the class of directed distributional dynamics that can be represented within the same density-valued time-series framework.

\section{Appendix}
This Appendix collects auxiliary technical results and 
the proofs of Theorems~\ref{thm:rate_beta} and \ref{thm:asy_normal_beta} in the main text. The proofs follow the algebraic template of \cite{Bai2009}, 
adapted to the present density-valued setting with PCA carried out along the time direction.

Throughout the Appendix, for a matrix \(A\), \(\|A\|\) denotes its operator norm and \(\|A\|_F\) its Frobenius norm. For a vector \(x\), \(\|x\|\) denotes its Euclidean norm.

\subsection{Proofs}
\begin{proof}[Proof of Theorem~\ref{thm:rate_beta}]
Let \(N:=CJ\) and \(T_0:=T-p\), and assume that \(N\to\infty\), \(T_0\to\infty\), and
\(N/T_0\to \rho\ge 0\).

By the projection-replacement argument established before the condition \(N/T_0\to0\) is invoked in Corollary~\ref{cor:linear_representation_routeA}, we have
\begin{align*}
\sqrt{NT_0}\,(\hat\beta-\beta_0)
&=
D(\hat\Lambda)^{-1}
\frac{1}{\sqrt{NT_0}}
\sum_{t=p+1}^T
\Bigg[
\tilde W_t' M_{\Pi_0}
-
\Big\{
\frac{1}{T_0}\sum_{s=p+1}^T a_{ts}\tilde W_s'
\Big\}M_{\Pi_0}
\Bigg]\tilde\varepsilon_t \\
&\qquad
+\sqrt{\frac{N}{T_0}}\,\zeta_{NT_0}
+o_p(1),
\end{align*}
where
\[
a_{ts}
=
f_t^{0\prime}
\Big(
F_0'F_0/T_0
\Big)^{-1}
f_s^0,
\qquad
F_0=(f_{p+1}^0,\ldots,f_T^0)',
\]
and
\[
\Pi_0
=
\tilde\Lambda_0(\tilde\Lambda_0'\tilde\Lambda_0)^{-1}\tilde\Lambda_0'.
\]

By Assumption~A(ii),
\[
\inf_{\tilde\Lambda\in\mathcal F}\lambda_{\min}\{D(\tilde\Lambda)\}\ge c>0
\]
with probability approaching one. Since \(\hat\Lambda\in\mathcal F\), it follows that
\[
D(\hat\Lambda)^{-1}=O_p(1).
\]

Next consider the score term
\[
S_{NT_0}
:=
\frac{1}{\sqrt{NT_0}}
\sum_{t=p+1}^T
\Bigg[
\tilde W_t' M_{\Pi_0}
-
\Big\{
\frac{1}{T_0}\sum_{s=p+1}^T a_{ts}\tilde W_s'
\Big\}M_{\Pi_0}
\Bigg]\tilde\varepsilon_t .
\]
Since \(\|M_{\Pi_0}\|\le 1\), we have
\begin{align*}
\|S_{NT_0}\|
&\le
\Bigg\|
\frac{1}{\sqrt{NT_0}}
\sum_{t=p+1}^T
\tilde W_t'\tilde\varepsilon_t
\Bigg\|
+
\Bigg\|
\frac{1}{\sqrt{NT_0}}
\sum_{t=p+1}^T
\Big\{
\frac{1}{T_0}\sum_{s=p+1}^T a_{ts}\tilde W_s'
\Big\}\tilde\varepsilon_t
\Bigg\|.
\end{align*}
The first term is \(O_p(1)\) by Assumptions~A--D and the same second-moment argument as in Step~1 of the proof of Lemma~\ref{lem:score_replacement_routeA}. For the second term, Assumption~B(i) implies that the coefficients \(a_{ts}\) are uniformly bounded. Therefore, by Assumption~C(vi),
\[
\Bigg\|
\frac{1}{\sqrt{NT_0}}
\sum_{t=p+1}^T
\Big\{
\frac{1}{T_0}\sum_{s=p+1}^T a_{ts}\tilde W_s'
\Big\}\tilde\varepsilon_t
\Bigg\|
=O_p(1).
\]
Therefore
\[
S_{NT_0}=O_p(1).
\]

Moreover, \(\zeta_{NT_0}=O_p(1)\) by Lemma~\ref{lem:bias_bound}. Since \(N/T_0\to \rho\ge 0\), we have
\[
\sqrt{\frac{N}{T_0}}\,\zeta_{NT_0}=O_p(1).
\]

Combining these bounds yields
\[
\sqrt{NT_0}\,(\hat\beta-\beta_0)=O_p(1),
\]
which completes the proof.
\end{proof}

\begin{proof}[Proof of Theorem~\ref{thm:asy_normal_beta}]
Let \(N:=CJ\) and \(T_0:=T-p\). Define
\[
H_0:=\frac{1}{N}\tilde\Lambda_0'\tilde\Lambda_0,
\qquad
\bar\Lambda_0:=\tilde\Lambda_0 H_0^{-1/2},
\qquad
\Pi_0:=\frac{1}{N}\bar\Lambda_0\bar\Lambda_0'
=
\tilde\Lambda_0(\tilde\Lambda_0'\tilde\Lambda_0)^{-1}\tilde\Lambda_0'.
\]
By Corollary~\ref{cor:linear_representation_routeA}, we have the linear expansion
\begin{equation}
\sqrt{NT_0}\,(\hat\beta-\beta_0)
=
D(\hat\Lambda)^{-1}
\left(
\frac{1}{\sqrt{NT_0}}\sum_{t=p+1}^T Z_t(\bar\Lambda_0)'\tilde\varepsilon_t
\right)
+o_p(1),
\label{eq:thm2_lin}
\end{equation}
where \(Z_t(\bar\Lambda_0)\) is defined as in Assumption~A(ii).

By Proposition~\ref{prop:consistency_ours}(ii),
\[
\|P_{\hat\Lambda}-\Pi_0\|\xrightarrow{p}0,
\]
and hence
\[
\|M_{\hat\Lambda}-M_{\Pi_0}\|\xrightarrow{p}0.
\]
Hence, by the definition of \(Z_t(\Lambda)\), we obtain
\[
\|D(\hat\Lambda)-D(\bar\Lambda_0)\|=o_p(1).
\]
Combined with Assumption~E(i), this yields
\[
D(\hat\Lambda)\xrightarrow{p}D_0.
\]
Since \(D_0\succ0\), we further obtain
\[
D(\hat\Lambda)^{-1}\xrightarrow{p}D_0^{-1}.
\]
By Assumption~E(ii),
\[
\frac{1}{\sqrt{NT_0}}\sum_{t=p+1}^T Z_t(\bar\Lambda_0)'\tilde\varepsilon_t
\xrightarrow{d}\mathcal N(0,\Omega_0).
\]
Therefore, Slutsky's theorem (e.g., \citealp[Theorem 2.7]{vanderVaart1998})
applied to \eqref{eq:thm2_lin} yields
\[
\sqrt{NT_0}\,(\hat\beta-\beta_0)
\ \xrightarrow{d}\
\mathcal N\!\left(0,\ D_0^{-1}\Omega_0 D_0^{-1}\right).
\]
This completes the proof.
\end{proof}

\subsection{Auxiliary results}
\begin{lemma}\label{lem:projection_bounds}
Let
\[
\tilde Y_t=\tilde W_t\beta_0+\tilde\Lambda_0 f_t^0+\tilde\varepsilon_t,
\qquad t=p+1,\ldots,T,
\]
where \(\tilde W_t=(\tilde X_{t-1},\ldots,\tilde X_{t-p})\) and
\(\tilde X_t=I_C\otimes(\tilde Y_{t,1},\ldots,\tilde Y_{t,C})\).
Let \(N=CJ\) and, for any \(\tilde\Lambda\in\mathcal F:=\{\Lambda\in\mathbb R^{N\times r}:N^{-1}\Lambda'\Lambda=I_r\}\), define
\[
P_{\tilde\Lambda}=N^{-1}\tilde\Lambda\tilde\Lambda',
\qquad
M_{\tilde\Lambda}=I_N-P_{\tilde\Lambda}.
\]
Assume Assumptions A(i), B, C, and D in the main text. Then, as \(J,T\to\infty\) with \(C\) fixed,
\begin{enumerate}[(i)]
\item \label{it:A1_1}
\[
\sup_{\tilde\Lambda\in\mathcal F}
\left\|
\frac{1}{N(T-p)}
\sum_{t=p+1}^T
\tilde W_t' M_{\tilde\Lambda}\tilde\varepsilon_t
\right\|
=o_p(1).
\]

\item \label{it:A1_2}
\[
\sup_{\tilde\Lambda\in\mathcal F}
\left|
\frac{1}{N(T-p)}
\sum_{t=p+1}^T
f_t^{0\prime}\tilde\Lambda_0' M_{\tilde\Lambda}\tilde\varepsilon_t
\right|
=o_p(1).
\]

\item \label{it:A1_3}
\[
\sup_{\tilde\Lambda\in\mathcal F}
\left|
\frac{1}{N(T-p)}
\sum_{t=p+1}^T
\tilde\varepsilon_t'P_{\tilde\Lambda}\tilde\varepsilon_t
\right|
=o_p(1).
\]
\end{enumerate}
\end{lemma}

\begin{proof}[Proof of Lemma~\ref{lem:projection_bounds}]
Write \(T_0:=T-p\). Fix \(r\) and let
\[
\mathcal F=\{\tilde\Lambda\in\mathbb R^{N\times r}:N^{-1}\tilde\Lambda'\tilde\Lambda=I_r\}.
\]
For \(\tilde\Lambda\in\mathcal F\), define
\[
P_{\tilde\Lambda}=N^{-1}\tilde\Lambda\tilde\Lambda',
\qquad
M_{\tilde\Lambda}=I_N-P_{\tilde\Lambda}.
\]

We first prove part (iii), which is then used in the proofs of parts (i) and (ii).

\textit{(iii).}
For \(\tilde\Lambda\in\mathcal F\), write
\[
P_{\tilde\Lambda}=UU',
\qquad
U:=N^{-1/2}\tilde\Lambda\in\mathbb R^{N\times r},
\qquad
U'U=I_r.
\]
Then
\[
S_3(\tilde\Lambda)
=
\frac{1}{NT_0}
\sum_{t=p+1}^T
\tilde\varepsilon_t'P_{\tilde\Lambda}\tilde\varepsilon_t
=
\frac{1}{N}\operatorname{tr}\!\left[
U'
\left(
\frac{1}{T_0}\sum_{t=p+1}^T \tilde\varepsilon_t\tilde\varepsilon_t'
\right)
U
\right].
\]
Hence, for
\[
A:=\frac{1}{T_0}\sum_{t=p+1}^T \tilde\varepsilon_t\tilde\varepsilon_t',
\]
we have
\[
|S_3(\tilde\Lambda)|
\le
\frac{1}{N}\,\|U'AU\|_*,
\]
where \(\|\cdot\|_*\) denotes the nuclear norm. Since \(U'AU\) is an \(r\times r\) matrix,
\[
\|U'AU\|_*\le r\|U'AU\|\le r\|A\|,
\]
and therefore
\[
\sup_{\tilde\Lambda\in\mathcal F}|S_3(\tilde\Lambda)|
\le
\frac{r}{N}
\left\|
\frac{1}{T_0}\sum_{t=p+1}^T \tilde\varepsilon_t\tilde\varepsilon_t'
\right\|.
\]

Now decompose
\[
\frac{1}{T_0}\sum_{t=p+1}^T \tilde\varepsilon_t\tilde\varepsilon_t'
=
\bar\Sigma_\varepsilon
+
\frac{1}{T_0}\sum_{t=p+1}^T
\big(\tilde\varepsilon_t\tilde\varepsilon_t'-\Sigma_{\varepsilon,t}\big),
\qquad
\bar\Sigma_\varepsilon:=\frac{1}{T_0}\sum_{t=p+1}^T \Sigma_{\varepsilon,t}.
\]
By Assumption~C(iii),
\[
\|\bar\Sigma_\varepsilon\|
\le
\frac{1}{T_0}\sum_{t=p+1}^T \|\Sigma_{\varepsilon,t}\|
\le K.
\]
Hence
\[
\frac{r}{N}\|\bar\Sigma_\varepsilon\|=O(N^{-1})=o(1).
\]

For the centered part, define
\[
C_t:=\tilde\varepsilon_t\tilde\varepsilon_t'-\Sigma_{\varepsilon,t},
\qquad t=p+1,\ldots,T.
\]
By Assumption~C(ii), \(\{C_t\}\) are independent across \(t\) and \(E(C_t)=0\). Using \(\|A\|\le \|A\|_F\),
\[
\left\|
\frac{1}{T_0}\sum_{t=p+1}^T C_t
\right\|
\le
\left\|
\frac{1}{T_0}\sum_{t=p+1}^T C_t
\right\|_F.
\]
Moreover,
\[
E\left\|
\frac{1}{T_0}\sum_{t=p+1}^T C_t
\right\|_F^2
=
\frac{1}{T_0^2}\sum_{t=p+1}^T E\|C_t\|_F^2.
\]
Now
\[
\|C_t\|_F^2
\le
2\|\tilde\varepsilon_t\tilde\varepsilon_t'\|_F^2
+
2\|\Sigma_{\varepsilon,t}\|_F^2.
\]
Since
\[
\|\tilde\varepsilon_t\tilde\varepsilon_t'\|_F^2
=
(\tilde\varepsilon_t'\tilde\varepsilon_t)^2
=
\Big(\sum_{i=1}^N \tilde\varepsilon_{it}^2\Big)^2
=
\sum_{i=1}^N\sum_{j=1}^N \tilde\varepsilon_{it}^2\tilde\varepsilon_{jt}^2,
\]
Assumption~C(i) and Hölder's inequality imply
\[
E\|\tilde\varepsilon_t\tilde\varepsilon_t'\|_F^2
\le
\sum_{i=1}^N\sum_{j=1}^N
\big(E|\tilde\varepsilon_{it}|^4\big)^{1/2}
\big(E|\tilde\varepsilon_{jt}|^4\big)^{1/2}
\le
K N^2.
\]
Also, since \(\|\Sigma_{\varepsilon,t}\|_F\le \sqrt N\,\|\Sigma_{\varepsilon,t}\|\),
Assumption~C(iii) gives
\[
\|\Sigma_{\varepsilon,t}\|_F^2\le N\|\Sigma_{\varepsilon,t}\|^2\le K N.
\]
Therefore
\[
E\|C_t\|_F^2\le K N^2
\]
uniformly in \(t\), and hence
\[
E\left\|
\frac{1}{T_0}\sum_{t=p+1}^T C_t
\right\|_F^2
\le
\frac{K N^2}{T_0}.
\]
It follows that
\[
\left\|
\frac{1}{T_0}\sum_{t=p+1}^T C_t
\right\|
=
O_p\!\left(\frac{N}{\sqrt{T_0}}\right).
\]
Consequently,
\[
\frac{r}{N}
\left\|
\frac{1}{T_0}\sum_{t=p+1}^T C_t
\right\|
=
O_p(T_0^{-1/2})
=
o_p(1).
\]

Combining the deterministic and centered parts, we obtain
\[
\sup_{\tilde\Lambda\in\mathcal F}|S_3(\tilde\Lambda)|
\le
\frac{r}{N}
\left\|
\frac{1}{T_0}\sum_{t=p+1}^T \tilde\varepsilon_t\tilde\varepsilon_t'
\right\|
=
o_p(1).
\]

\textit{(i).}
Decompose \(S_1(\tilde\Lambda)\) using \(M_{\tilde\Lambda}=I_N-P_{\tilde\Lambda}\):
\[
S_1(\tilde\Lambda)
=
\frac{1}{NT_0}\sum_{t=p+1}^T \tilde W_t'\tilde\varepsilon_t
-
\frac{1}{NT_0}\sum_{t=p+1}^T \tilde W_t'P_{\tilde\Lambda}\tilde\varepsilon_t
=:A_1-B_1(\tilde\Lambda).
\]

\textit{Step 1: the non-projection term.}
By Assumptions C(ii) and D, the error \(\tilde\varepsilon_t\) is mean independent of past outcomes, and hence of \(\tilde W_t\), which is measurable with respect to \(\{\tilde Y_s:s<t\}\). Therefore \(E(A_1)=0\), and
\[
E\|A_1\|^2
=
\frac{1}{N^2T_0^2}
\sum_{t=p+1}^T
E\!\left[
\tilde\varepsilon_t'\tilde W_t\tilde W_t'\tilde\varepsilon_t
\right].
\]
Conditioning on \(\tilde W_t\) and using \(E(\tilde\varepsilon_t\tilde\varepsilon_t')=\Sigma_{\varepsilon,t}\),
\[
E\!\left[
\tilde\varepsilon_t'\tilde W_t\tilde W_t'\tilde\varepsilon_t
\,\middle|\, \tilde W_t
\right]
=
\operatorname{tr}(\Sigma_{\varepsilon,t}\tilde W_t\tilde W_t')
\le
\|\Sigma_{\varepsilon,t}\|\,\|\tilde W_t\|_F^2
\le
K\|\tilde W_t\|_F^2.
\]
Since \(\tilde W_t\) is \(N\times (pC^2)\) and \(C,p\) are fixed, Assumption A(i) together with Assumptions B--D implies
\[
E\|\tilde W_t\|_F^2 = O(N)
\]
uniformly in \(t\). Therefore
\[
E\|A_1\|^2
\le
\frac{1}{N^2T_0^2}\sum_{t=p+1}^T K\cdot O(N)
=
O\!\left(\frac{1}{NT_0}\right)\to 0,
\]
so \(A_1=o_p(1)\).

\textit{Step 2: the projection term, uniformly in \(\tilde\Lambda\).}
By Cauchy--Schwarz,
\[
\|B_1(\tilde\Lambda)\|
\le
\left(
\frac{1}{NT_0}\sum_{t=p+1}^T \|\tilde W_t\|_F^2
\right)^{1/2}
\left(
\frac{1}{NT_0}\sum_{t=p+1}^T \tilde\varepsilon_t'P_{\tilde\Lambda}\tilde\varepsilon_t
\right)^{1/2}.
\]
The first factor is \(O_p(1)\) because
\[
\frac{1}{T_0}\sum_{t=p+1}^T \|\tilde W_t\|_F^2 = O_p(N),
\]
and the second factor is \(o_p(1)\) uniformly in \(\tilde\Lambda\in\mathcal F\) by part (iii). Hence
\[
\sup_{\tilde\Lambda\in\mathcal F}\|B_1(\tilde\Lambda)\|=o_p(1).
\]
Combining Step 1 and Step 2,
\[
\sup_{\tilde\Lambda\in\mathcal F}\|S_1(\tilde\Lambda)\|
\le
\|A_1\|+\sup_{\tilde\Lambda\in\mathcal F}\|B_1(\tilde\Lambda)\|
=o_p(1).
\]

\textit{(ii).}
Similarly, write
\[
S_2(\tilde\Lambda)
=
\frac{1}{NT_0}\sum_{t=p+1}^T f_t^{0\prime}\tilde\Lambda_0'\tilde\varepsilon_t
-
\frac{1}{NT_0}\sum_{t=p+1}^T f_t^{0\prime}\tilde\Lambda_0'P_{\tilde\Lambda}\tilde\varepsilon_t
=:A_2-B_2(\tilde\Lambda).
\]

\textit{Step 1: the non-projection term.}
Since \(\tilde\Lambda_0\) and \(\{f_t^0\}\) are non-stochastic arrays,
\[
E(A_2)=0
\]
by Assumption C(ii), and by independence across \(t\),
\[
E(A_2^2)
=
\frac{1}{N^2T_0^2}
\sum_{t=p+1}^T
E\!\left[
\big(
f_t^{0\prime}\tilde\Lambda_0'\tilde\varepsilon_t
\big)^2
\right].
\]
Moreover,
\[
E\!\left[
\big(
f_t^{0\prime}\tilde\Lambda_0'\tilde\varepsilon_t
\big)^2
\right]
=
f_t^{0\prime}\tilde\Lambda_0'\Sigma_{\varepsilon,t}\tilde\Lambda_0 f_t^0
\le
K\,\|\tilde\Lambda_0\|^2\,\|f_t^0\|^2.
\]
By Assumption B(ii),
\[
\|\tilde\Lambda_0\|^2 = O(N),
\]
and by Assumption B(i),
\[
\frac{1}{T_0}\sum_{t=p+1}^T \|f_t^0\|^2 = O(1).
\]
Hence
\[
E(A_2^2)
\le
\frac{K}{N^2T_0^2}\,\|\tilde\Lambda_0\|^2
\sum_{t=p+1}^T \|f_t^0\|^2
=
O\!\left(\frac{1}{NT_0}\right)
=o(1),
\]
so \(A_2=o_p(1)\).

\textit{Step 2: the projection term, uniformly in \(\tilde\Lambda\).}
By Cauchy--Schwarz,
\[
|B_2(\tilde\Lambda)|
\le
\left(
\frac{1}{NT_0}\sum_{t=p+1}^T \|\tilde\Lambda_0 f_t^0\|^2
\right)^{1/2}
\left(
\frac{1}{NT_0}\sum_{t=p+1}^T \tilde\varepsilon_t'P_{\tilde\Lambda}\tilde\varepsilon_t
\right)^{1/2}.
\]
The first factor is \(O_p(1)\), since
\[
\|\tilde\Lambda_0 f_t^0\|^2
\le
\|\tilde\Lambda_0\|^2\,\|f_t^0\|^2,
\qquad
\|\tilde\Lambda_0\|^2=O(N),
\qquad
\frac{1}{T_0}\sum_{t=p+1}^T \|f_t^0\|^2=O(1),
\]
and the second factor is \(o_p(1)\) uniformly in \(\tilde\Lambda\) by part (iii). Hence
\[
\sup_{\tilde\Lambda\in\mathcal F}|B_2(\tilde\Lambda)|=o_p(1),
\]
and therefore
\[
\sup_{\tilde\Lambda\in\mathcal F}|S_2(\tilde\Lambda)|
\le
|A_2|+\sup_{\tilde\Lambda\in\mathcal F}|B_2(\tilde\Lambda)|
=o_p(1).
\]
This completes the proof.
\end{proof}

\begin{proposition}[Consistency]
\label{prop:consistency_ours}
Consider the transformed model
\[
\tilde Y_t=\tilde W_t\beta_0+\tilde\Lambda_0 f_t^0+\tilde\varepsilon_t,
\qquad t=p+1,\ldots,T,
\]
with \(N:=CJ\) and
\[
\mathcal F:=\{\Lambda\in\mathbb R^{N\times r}:N^{-1}\Lambda'\Lambda=I_r\}.
\]
Let \((\hat\beta,\hat\Lambda)\) minimize
\[
Q(\beta,\tilde\Lambda)
=
\sum_{t=p+1}^T(\tilde Y_t-\tilde W_t\beta)'M_{\tilde\Lambda}(\tilde Y_t-\tilde W_t\beta),
\qquad
M_{\tilde\Lambda}=I_N-N^{-1}\tilde\Lambda\tilde\Lambda'.
\]
Under Assumptions~A--D, as \(J,T\to\infty\) with \(C\) fixed, the following hold:
\begin{enumerate}[(i)]
\item \(\hat\beta\xrightarrow{p}\beta_0\).

\item \(N^{-1}\tilde\Lambda_0'\hat\Lambda\) is invertible with probability approaching one, and
\[
\|P_{\hat\Lambda}-\Pi_0\|\xrightarrow{p}0,
\]
where
\[
P_{\hat\Lambda}=N^{-1}\hat\Lambda\hat\Lambda',
\qquad
\Pi_0:=\tilde\Lambda_0(\tilde\Lambda_0'\tilde\Lambda_0)^{-1}\tilde\Lambda_0'.
\]
\end{enumerate}
\end{proposition}

\begin{proof}[Proof of Proposition~\ref{prop:consistency_ours}] 
We divide the proof into four steps.
\paragraph{Step 1 (Decomposition).}
Let \(N:=CJ\) and \(T_0:=T-p\). Define the normalized criterion
\[
S_{NT_0}(\beta,\tilde\Lambda):=\frac{1}{NT_0}\sum_{t=p+1}^T
(\tilde Y_t-\tilde W_t\beta)'M_{\tilde\Lambda}(\tilde Y_t-\tilde W_t\beta),
\qquad
M_{\tilde\Lambda}:=I_N-\frac{1}{N}\tilde\Lambda\tilde\Lambda',
\quad \tilde\Lambda\in\mathcal F.
\]
Without loss of generality, set \(\beta_0=0\) by replacing \(\tilde Y_t\) with \(\tilde Y_t-\tilde W_t\beta_0\). Then
\[
\tilde Y_t=\tilde\Lambda_0 f_t^0+\tilde\varepsilon_t,
\]
and hence
\begin{align*}
S_{NT_0}(\beta,\tilde\Lambda)
&=
\frac{1}{NT_0}\sum_{t=p+1}^T
(\tilde\Lambda_0 f_t^0-\tilde W_t\beta)'M_{\tilde\Lambda}(\tilde\Lambda_0 f_t^0-\tilde W_t\beta) \\
&\quad
+\frac{2}{NT_0}\sum_{t=p+1}^T
(\tilde\Lambda_0 f_t^0-\tilde W_t\beta)'M_{\tilde\Lambda}\tilde\varepsilon_t
+\frac{1}{NT_0}\sum_{t=p+1}^T\tilde\varepsilon_t'M_{\tilde\Lambda}\tilde\varepsilon_t .
\end{align*}
Write
\[
\tilde\varepsilon_t'M_{\tilde\Lambda}\tilde\varepsilon_t
=\tilde\varepsilon_t'\tilde\varepsilon_t-\tilde\varepsilon_t'P_{\tilde\Lambda}\tilde\varepsilon_t,
\qquad
P_{\tilde\Lambda}:=\frac{1}{N}\tilde\Lambda\tilde\Lambda'.
\]
Therefore,
\[
S_{NT_0}(\beta,\tilde\Lambda)
=
\widetilde S_{NT_0}(\beta,\tilde\Lambda)
+
\underbrace{\frac{1}{NT_0}\sum_{t=p+1}^T\tilde\varepsilon_t'\tilde\varepsilon_t}_{=:C_{NT_0}}
+
R_{NT_0}(\beta,\tilde\Lambda),
\]
where
\begin{align*}
\widetilde S_{NT_0}(\beta,\tilde\Lambda)
&:=
\frac{1}{NT_0}\sum_{t=p+1}^T
(\tilde\Lambda_0 f_t^0-\tilde W_t\beta)'M_{\tilde\Lambda}(\tilde\Lambda_0 f_t^0-\tilde W_t\beta),\\[2mm]
R_{NT_0}(\beta,\tilde\Lambda)
&:=
\frac{2}{NT_0}\sum_{t=p+1}^T
(\tilde\Lambda_0 f_t^0-\tilde W_t\beta)'M_{\tilde\Lambda}\tilde\varepsilon_t
-\frac{1}{NT_0}\sum_{t=p+1}^T\tilde\varepsilon_t'P_{\tilde\Lambda}\tilde\varepsilon_t .
\end{align*}
Since \(C_{NT_0}\) does not depend on \((\beta,\tilde\Lambda)\), it plays no role in minimization.

\medskip
\noindent\textbf{Uniform bound for the remainder.}
Expanding the cross term gives
\[
(\tilde\Lambda_0 f_t^0-\tilde W_t\beta)'M_{\tilde\Lambda}\tilde\varepsilon_t
=
f_t^{0\prime}\tilde\Lambda_0' M_{\tilde\Lambda}\tilde\varepsilon_t
-\beta'\tilde W_t' M_{\tilde\Lambda}\tilde\varepsilon_t .
\]
Hence, for any \(\beta\in\mathbb R^K\),
\begin{align*}
\sup_{\tilde\Lambda\in\mathcal F}|R_{NT_0}(\beta,\tilde\Lambda)|
&\le
2\sup_{\tilde\Lambda\in\mathcal F}\left|
\frac{1}{NT_0}\sum_{t=p+1}^T f_t^{0\prime}\tilde\Lambda_0' M_{\tilde\Lambda}\tilde\varepsilon_t
\right| \\
&\quad
+2\|\beta\|\sup_{\tilde\Lambda\in\mathcal F}\left\|
\frac{1}{NT_0}\sum_{t=p+1}^T \tilde W_t' M_{\tilde\Lambda}\tilde\varepsilon_t
\right\| \\
&\quad
+\sup_{\tilde\Lambda\in\mathcal F}\left|
\frac{1}{NT_0}\sum_{t=p+1}^T \tilde\varepsilon_t'P_{\tilde\Lambda}\tilde\varepsilon_t
\right|.
\end{align*}
By Lemma~\ref{lem:projection_bounds}\ref{it:A1_2}, Lemma~\ref{lem:projection_bounds}\ref{it:A1_1}, and Lemma~\ref{lem:projection_bounds}\ref{it:A1_3}, respectively, each supremum term on the right-hand side is \(o_p(1)\). Therefore,
\[
\sup_{\tilde\Lambda\in\mathcal F}|R_{NT_0}(\beta,\tilde\Lambda)|
=
o_p(1)\,(1+\|\beta\|).
\]

\paragraph{Step 2 (Square completion).}
Fix \(\tilde\Lambda\in\mathcal F\). Let \(F_0:=(f_{p+1}^0,\ldots,f_T^0)'\in\mathbb R^{T_0\times r}\) and define
\[
B:=\Big(\frac{1}{T_0}F_{0}^\prime F_0\Big)\otimes\Big(\frac{1}{N}I_N\Big),
\qquad
\eta(\tilde\Lambda):=\mathrm{vec}(M_{\tilde\Lambda}\tilde\Lambda_0)\in\mathbb R^{Nr},
\]
and
\[
A(\tilde\Lambda)
:=\frac{1}{NT_0}\sum_{t=p+1}^T \tilde W_t' M_{\tilde\Lambda}\tilde W_t \in\mathbb R^{K\times K},
\qquad
C(\tilde\Lambda)
:=-\frac{1}{NT_0}\sum_{t=p+1}^T \big(f_t^0\otimes M_{\tilde\Lambda}\tilde W_t\big)\in\mathbb R^{Nr\times K},
\]
where \(K=\dim(\beta)=pC^2\).

Expanding \(\widetilde S_{NT_0}(\beta,\tilde\Lambda)\) and using vec/Kronecker identities yields
\[
\widetilde S_{NT_0}(\beta,\tilde\Lambda)
=
\beta'A(\tilde\Lambda)\beta
+\eta(\tilde\Lambda)'B\,\eta(\tilde\Lambda)
+2\,\beta' C(\tilde\Lambda)'\eta(\tilde\Lambda).
\]
Completing the square gives
\begin{align*}
\widetilde S_{NT_0}(\beta,\tilde\Lambda)
&=
\beta'\big(A(\tilde\Lambda)-C(\tilde\Lambda)'B^{-1}C(\tilde\Lambda)\big)\beta\\
&\quad
+\Big(\eta(\tilde\Lambda)+B^{-1}C(\tilde\Lambda)\beta\Big)'B
\Big(\eta(\tilde\Lambda)+B^{-1}C(\tilde\Lambda)\beta\Big) \\
&=:\ \beta'D(\tilde\Lambda)\beta+\theta(\tilde\Lambda,\beta)'B\,\theta(\tilde\Lambda,\beta),
\end{align*}
where
\[
D(\tilde\Lambda):=A(\tilde\Lambda)-C(\tilde\Lambda)'B^{-1}C(\tilde\Lambda),
\qquad
\theta(\tilde\Lambda,\beta):=\eta(\tilde\Lambda)+B^{-1}C(\tilde\Lambda)\beta .
\]
Since \(B\succ0\) with probability approaching one by Assumption~B(i), we have
\[
\widetilde S_{NT_0}(\beta,\tilde\Lambda)\ge \beta'D(\tilde\Lambda)\beta.
\]

Let
\[
a_{ts}:= f_t^{0\prime}\Big(\frac{1}{T_0}F_0^{\prime}F_0\Big)^{-1}f_s^0
\]
and define
\[
Z_t(\tilde\Lambda)
:=
M_{\tilde\Lambda}\tilde W_t-\frac{1}{T_0}\sum_{s=p+1}^T a_{ts}\,M_{\tilde\Lambda}\tilde W_s.
\]
Then the Frisch--Waugh--Lovell identity implies
\[
D(\tilde\Lambda)
=
\frac{1}{NT_0}\sum_{t=p+1}^T Z_t(\tilde\Lambda)'Z_t(\tilde\Lambda).
\]
In particular, Assumption~A(ii) implies that there exists \(c_0>0\) such that
\[
\inf_{\tilde\Lambda\in\mathcal F}\lambda_{\min}\{D(\tilde\Lambda)\}\ge c_0
\quad \text{with probability approaching one.}
\]

\paragraph{Step 3 (Consistency of \(\hat\beta\)).}
Let
\[
H_0:=\frac{1}{N}\tilde\Lambda_0'\tilde\Lambda_0,
\qquad
\bar\Lambda_0:=\tilde\Lambda_0 H_0^{-1/2}.
\]
Then \(\bar\Lambda_0\in\mathcal F\), and
\[
P_{\bar\Lambda_0}
=
\frac{1}{N}\bar\Lambda_0\bar\Lambda_0'
=
\tilde\Lambda_0(\tilde\Lambda_0'\tilde\Lambda_0)^{-1}\tilde\Lambda_0'
=
\Pi_0,
\]
so \(M_{\bar\Lambda_0}\tilde\Lambda_0=0\).

By Step~2,
\[
\widetilde S_{NT_0}(\beta,\tilde\Lambda)\ge c_0\|\beta\|^2
\]
uniformly in \(\tilde\Lambda\in\mathcal F\) with probability approaching one. Combining this with Step~1 gives
\[
S_{NT_0}(\beta,\tilde\Lambda)
\ge
C_{NT_0}
+
c_0\|\beta\|^2
-
o_p(1)(1+\|\beta\|)
\]
uniformly in \(\tilde\Lambda\in\mathcal F\).
Since the quadratic term dominates the linear remainder for \(\|\beta\|\) large enough, there exists \(B<\infty\) such that
\[
\inf_{\|\beta\|>B,\ \tilde\Lambda\in\mathcal F}
S_{NT_0}(\beta,\tilde\Lambda)
>
C_{NT_0}+1
\]
with probability approaching one.

On the other hand, evaluating the criterion at \((0,\bar\Lambda_0)\), we have
\[
S_{NT_0}(0,\bar\Lambda_0)
=
C_{NT_0}+R_{NT_0}(0,\bar\Lambda_0),
\]
because \(M_{\bar\Lambda_0}\tilde\Lambda_0=0\). By Step~1,
\[
R_{NT_0}(0,\bar\Lambda_0)=o_p(1),
\]
and hence
\[
S_{NT_0}(0,\bar\Lambda_0)=C_{NT_0}+o_p(1).
\]
Therefore, with probability approaching one, any minimizer \((\hat\beta,\hat\Lambda)\) of \(S_{NT_0}(\beta,\tilde\Lambda)\) must satisfy
\[
\|\hat\beta\|\le B.
\]

Restricting attention to the compact set \(\{\beta:\|\beta\|\le B\}\), Step~1 yields
\[
\sup_{\|\beta\|\le B,\ \tilde\Lambda\in\mathcal F}|R_{NT_0}(\beta,\tilde\Lambda)|=o_p(1).
\]
Hence
\[
S_{NT_0}(\beta,\tilde\Lambda)
=
\widetilde S_{NT_0}(\beta,\tilde\Lambda)
+
C_{NT_0}
+
o_p(1)
\]
uniformly over \(\|\beta\|\le B\) and \(\tilde\Lambda\in\mathcal F\). Since \((\hat\beta,\hat\Lambda)\) is a minimizer,
\[
\widetilde S_{NT_0}(\hat\beta,\hat\Lambda)
\le
\widetilde S_{NT_0}(0,\bar\Lambda_0)+o_p(1).
\]
Because \(M_{\bar\Lambda_0}\tilde\Lambda_0=0\),
\[
\widetilde S_{NT_0}(0,\bar\Lambda_0)=0,
\]
and therefore
\begin{equation}\label{eq:tildeS_hat_op1_bai}
\widetilde S_{NT_0}(\hat\beta,\hat\Lambda)=o_p(1).
\end{equation}

Again by Step~2,
\[
\widetilde S_{NT_0}(\hat\beta,\hat\Lambda)\ge \hat\beta'D(\hat\Lambda)\hat\beta\ge c_0\|\hat\beta\|^2
\quad \text{with probability approaching one.}
\]
Combining this with \eqref{eq:tildeS_hat_op1_bai} yields
\[
\|\hat\beta\|=o_p(1),
\]
that is,
\[
\hat\beta\xrightarrow{p}0=\beta_0.
\]

\paragraph{Step 4 (Consistency of the loading space).}
Expanding \(\widetilde S_{NT_0}(\beta,\tilde\Lambda)\) yields
\begin{align*}
\widetilde S_{NT_0}(\beta,\tilde\Lambda)
&=
\underbrace{\frac{1}{NT_0}\sum_{t=p+1}^T
f_t^{0\prime}\tilde\Lambda_0'M_{\tilde\Lambda}\tilde\Lambda_0 f_t^0}_{=:c(\tilde\Lambda)}
-2\beta'\underbrace{\frac{1}{NT_0}\sum_{t=p+1}^T
\tilde W_t'M_{\tilde\Lambda}\tilde\Lambda_0 f_t^0}_{=:b(\tilde\Lambda)}\\
&\quad
+\beta'\underbrace{\frac{1}{NT_0}\sum_{t=p+1}^T
\tilde W_t'M_{\tilde\Lambda}\tilde W_t}_{=:A(\tilde\Lambda)}\beta .
\end{align*}
Evaluating at \((\hat\beta,\hat\Lambda)\) and using \eqref{eq:tildeS_hat_op1_bai}, we obtain
\[
c(\hat\Lambda)
=
\widetilde S_{NT_0}(\hat\beta,\hat\Lambda)
+2\hat\beta'b(\hat\Lambda)-\hat\beta'A(\hat\Lambda)\hat\beta
=o_p(1),
\]
where we used \(\|\hat\beta\|=o_p(1)\) from Step~3 together with \(b(\hat\Lambda)=O_p(1)\) and \(A(\hat\Lambda)=O_p(1)\).

Write \(F_0=(f_{p+1}^0,\ldots,f_T^0)'\in\mathbb R^{T_0\times r}\). Then
\[
c(\hat\Lambda)
=
\frac{1}{N}\mathrm{tr}\!\left(
\tilde\Lambda_0' M_{\hat\Lambda}\tilde\Lambda_0\cdot
\frac{1}{T_0}F_0^{\prime}F_0
\right)
=o_p(1).
\]
By Assumption~B(i), \(T_0^{-1}F_0^{\prime}F_0 \to \Sigma_f\succ0\), and therefore
\begin{equation}\label{eq:Lambda_space_key_bai}
\frac{1}{N}\tilde\Lambda_0' M_{\hat\Lambda}\tilde\Lambda_0=o_p(1).
\end{equation}

Define
\[
H_0:=\frac{1}{N}\tilde\Lambda_0'\tilde\Lambda_0,
\qquad
\bar\Lambda_0:=\tilde\Lambda_0 H_0^{-1/2}.
\]
By Assumption~B(ii), \(H_0\to \Sigma_\Lambda\succ0\), so \(H_0\) is invertible with probability approaching one and
\[
\frac{1}{N}\bar\Lambda_0'\bar\Lambda_0=I_r.
\]
Hence
\[
P_{\bar\Lambda_0}:=\frac{1}{N}\bar\Lambda_0\bar\Lambda_0'
\]
is an orthogonal projector of rank \(r\), and
\[
P_{\bar\Lambda_0}
=
\tilde\Lambda_0(\tilde\Lambda_0'\tilde\Lambda_0)^{-1}\tilde\Lambda_0'
=
\Pi_0.
\]

From \eqref{eq:Lambda_space_key_bai},
\[
\frac{1}{N}\bar\Lambda_0' M_{\hat\Lambda}\bar\Lambda_0
=
H_0^{-1/2}\Big(\frac{1}{N}\tilde\Lambda_0' M_{\hat\Lambda}\tilde\Lambda_0\Big)H_0^{-1/2}
=o_p(1).
\]
Since both \(P_{\hat\Lambda}\) and \(P_{\bar\Lambda_0}\) are orthogonal projectors of rank \(r\),
\[
\|P_{\hat\Lambda}-P_{\bar\Lambda_0}\|_F^2
=
2\,\mathrm{tr}(M_{\hat\Lambda}P_{\bar\Lambda_0}).
\]
Using \(P_{\bar\Lambda_0}=(1/N)\bar\Lambda_0\bar\Lambda_0'\), we obtain
\[
\mathrm{tr}(M_{\hat\Lambda}P_{\bar\Lambda_0})
=
\frac{1}{N}\mathrm{tr}(\bar\Lambda_0'M_{\hat\Lambda}\bar\Lambda_0)
=o_p(1),
\]
and therefore
\[
\|P_{\hat\Lambda}-P_{\bar\Lambda_0}\|_F^2=o_p(1).
\]
Since \(\|A\|\le \|A\|_F\) for any matrix \(A\), we conclude that
\[
\|P_{\hat\Lambda}-\Pi_0\|
=
\|P_{\hat\Lambda}-P_{\bar\Lambda_0}\|
=o_p(1).
\]

Finally, since
\[
\frac{1}{N}\tilde\Lambda_0'P_{\hat\Lambda}\tilde\Lambda_0
=
\frac{1}{N}\tilde\Lambda_0'\tilde\Lambda_0
-\frac{1}{N}\tilde\Lambda_0'M_{\hat\Lambda}\tilde\Lambda_0
=
H_0+o_p(1),
\]
and \(H_0\) is invertible with probability approaching one, it follows that
\[
\frac{1}{N}\tilde\Lambda_0'\hat\Lambda
\]
is invertible with probability approaching one. This completes the proof.
\end{proof}

\begin{proposition}[Loading PCA]
\label{prop:loading_pca}
Let
\[
\tilde Y_t=\tilde W_t\beta_0+\tilde\Lambda_0 f_t^0+\tilde\varepsilon_t,
\qquad t=p+1,\ldots,T,
\]
with \(N:=CJ\) and \(T_0:=T-p\). Let \(\hat\beta\) be the estimator of \(\beta_0\) defined in Proposition~\ref{prop:consistency_ours}, and define
\[
u_t:=\tilde Y_t-\tilde W_t\hat\beta,\qquad
\hat S:=\frac{1}{NT_0}\sum_{t=p+1}^T u_tu_t'.
\]
Let \(\hat\Lambda\in\mathbb R^{N\times r}\) collect the eigenvectors associated with the \(r\) largest eigenvalues of \(\hat S\), normalized such that
\[
\frac{1}{N}\hat\Lambda'\hat\Lambda=I_r.
\]
Let \(\hat V_{NT_0}\) denote the corresponding diagonal matrix of eigenvalues, defined by
\begin{equation}\label{eq:eig_loading}
\hat S\,\hat\Lambda=\hat\Lambda\,\hat V_{NT_0}.
\end{equation}

Under Assumptions~A--D, as \(J,T\to\infty\) with \(C\) fixed, the following hold:

\begin{enumerate}[(i)]
\item \textbf{Eigenvalues.}
\(\hat V_{NT_0}\) is invertible with probability approaching one, and
\[
\hat V_{NT_0}\xrightarrow{p}V,
\]
where \(V\) is the \(r\times r\) diagonal matrix consisting of the \(r\) eigenvalues of
\[
\Sigma_f^{1/2}\Sigma_\Lambda\Sigma_f^{1/2}
\quad\text{(equivalently, of }\Sigma_\Lambda^{1/2}\Sigma_f\Sigma_\Lambda^{1/2}\text{)}.
\]

\item \textbf{Loading space and explicit rotation.}
Define
\[
H:=\Big(\frac{1}{T_0}F_0^{\prime}F_0\Big)\Big(\frac{1}{N}\tilde\Lambda_0'\hat\Lambda\Big)\hat V_{NT_0}^{-1},
\qquad
F_0:=(f_{p+1}^0,\ldots,f_T^0)'.
\]
Then \(H\) is invertible with probability approaching one, and
\[
\frac{1}{N}\|\hat\Lambda-\tilde\Lambda_0 H\|_F^2
=
O_p\!\big(\|\hat\beta-\beta_0\|^2\big)
+
O_p\!\Big(\frac{1}{\min\{N,T_0\}}\Big).
\]
In particular, if \(\hat\beta\xrightarrow{p}\beta_0\), then
\[
\frac{1}{N}\|\hat\Lambda-\tilde\Lambda_0 H\|_F^2
=
o_p(1),
\qquad\text{and hence}\qquad
\|P_{\hat\Lambda}-\Pi_0\|\xrightarrow{p}0,
\]
where
\[
P_{\hat\Lambda}:=N^{-1}\hat\Lambda\hat\Lambda',
\qquad
\Pi_0:=\tilde\Lambda_0(\tilde\Lambda_0'\tilde\Lambda_0)^{-1}\tilde\Lambda_0'.
\]
\end{enumerate}
\end{proposition}

\begin{proof}[Proof of Proposition~\ref{prop:loading_pca}]
Let \(T_0:=T-p\) and \(N:=CJ\). Define
\[
u_t:=\tilde Y_t-\tilde W_t\hat\beta
=\tilde\Lambda_0 f_t^0+\tilde\varepsilon_t-\tilde W_t(\hat\beta-\beta_0),
\qquad t=p+1,\ldots,T.
\]
Write
\[
s_t:=\tilde\Lambda_0 f_t^0,\qquad
\varepsilon_t:=\tilde\varepsilon_t,\qquad
r_t:=-\tilde W_t(\hat\beta-\beta_0),
\]
so that \(u_t=s_t+\varepsilon_t+r_t\). Recall
\[
\hat S=\frac{1}{NT_0}\sum_{t=p+1}^T u_tu_t',
\qquad
\hat S\,\hat\Lambda=\hat\Lambda\,\hat V_{NT_0},
\qquad
\frac1N\hat\Lambda'\hat\Lambda=I_r.
\]

\paragraph{Step 1 (Signal decomposition and a spectral bound).}
Define the signal matrix
\begin{equation}\label{eq:S_signal}
S_f
:=\frac{1}{NT_0}\sum_{t=p+1}^T s_ts_t'
=\frac1N\tilde\Lambda_0\Big(\frac1{T_0}F_0{\prime}F_0\Big)\tilde\Lambda_0',
\qquad
F_0:=(f_{p+1}^0,\ldots,f_T^0)'.
\end{equation}
Then
\[
\hat S = S_f + R,
\]
where
\[
R
=\frac{1}{NT_0}\sum_{t=p+1}^T
\Big(
s_t\varepsilon_t'+\varepsilon_ts_t'
+s_tr_t'+r_ts_t'
+\varepsilon_tr_t'+r_t\varepsilon_t'
+\varepsilon_t\varepsilon_t'
+r_tr_t'
\Big).
\]

\medskip
\noindent\emph{Claim.} Under Assumptions~A--D,
\begin{equation}\label{eq:R_rate}
\|R\|
=
O_p\big(\|\hat\beta-\beta_0\|\big)
+
O_p\!\Big(\frac1{\sqrt{T_0}}\Big)
+
O_p\!\Big(\frac1{N}\Big).
\end{equation}

\noindent\emph{Proof of the claim.}
We bound each block in operator norm.

\smallskip
\noindent (a) Terms involving \(r_t\).
By Cauchy--Schwarz and \(\|ab'\|\le \|a\|\,\|b\|\),
\[
\Big\|\frac{1}{NT_0}\sum s_tr_t'\Big\|
\le
\Big(\frac{1}{NT_0}\sum \|s_t\|^2\Big)^{1/2}
\Big(\frac{1}{NT_0}\sum \|r_t\|^2\Big)^{1/2}.
\]
By \eqref{eq:S_signal},
\[
\frac{1}{NT_0}\sum \|s_t\|^2
=
\frac1N\tr\!\left(\tilde\Lambda_0' \tilde\Lambda_0\cdot \frac1{T_0}F_0{\prime}F_0\right)
=O(1)
\]
using Assumption~B(i)--(ii). Moreover,
\[
\frac{1}{NT_0}\sum \|r_t\|^2
\le
\|\hat\beta-\beta_0\|^2\cdot \frac{1}{NT_0}\sum \|\tilde W_t\|_F^2
=O_p(1)\,\|\hat\beta-\beta_0\|^2,
\]
because \(K=\dim(\beta)=pC^2\) is fixed and, by Assumption~A(i) together with Assumptions~B and C,
\(E\|\tilde W_t\|_F^2=O(N)\) uniformly in \(t\). Hence
\[
\Big\|\frac{1}{NT_0}\sum s_tr_t'\Big\| = O_p(\|\hat\beta-\beta_0\|).
\]
The same bound holds for \(\frac{1}{NT_0}\sum r_ts_t'\), \(\frac{1}{NT_0}\sum \varepsilon_tr_t'\), and \(\frac{1}{NT_0}\sum r_t\varepsilon_t'\). Finally,
\[
\Big\|\frac{1}{NT_0}\sum r_tr_t'\Big\|
\le
\frac{1}{NT_0}\sum \|r_t\|^2
=
O_p(\|\hat\beta-\beta_0\|^2).
\]

\smallskip
\noindent (b) Cross terms \(s_t\varepsilon_t'\) and \(\varepsilon_ts_t'\).
Since \(s_t=\tilde\Lambda_0 f_t^0\) is deterministic and \(E(\varepsilon_t)=0\) by Assumption~C(ii),
we have \(E(s_t\varepsilon_t')=0\). Moreover, the summands are independent over \(t\) by Assumption~C(ii).
Using \(\|A\|\le \|A\|_F\),
\[
E\Big\|\frac{1}{NT_0}\sum s_t\varepsilon_t'\Big\|_F^2
=
\frac{1}{N^2T_0^2}\sum E(\|s_t\|^2\|\varepsilon_t\|^2).
\]
Moreover,
\[
\|s_t\|^2=\|\tilde\Lambda_0 f_t^0\|^2
\le \|\tilde\Lambda_0\|^2\|f_t^0\|^2=O(N)
\]
uniformly in \(t\), by Assumption~B(i)--(ii),
and \(E\|\varepsilon_t\|^2=O(N)\) by Assumption~C(i). Therefore
\[
E\Big\|\frac{1}{NT_0}\sum s_t\varepsilon_t'\Big\|_F^2
=O\Big(\frac{1}{T_0}\Big),
\]
and hence
\[
\Big\|\frac{1}{NT_0}\sum s_t\varepsilon_t'\Big\|
=O_p\Big(\frac{1}{\sqrt{T_0}}\Big).
\]
The same bound holds for \(\frac{1}{NT_0}\sum \varepsilon_ts_t'\).

\smallskip
\noindent (c) Pure error term.
By Assumption~C(ii),
\[
E\Big(\frac{1}{NT_0}\sum \varepsilon_t\varepsilon_t'\Big)
=\frac{1}{N}\cdot \frac1{T_0}\sum_{t=p+1}^T \Sigma_{\varepsilon,t},
\]
so by Assumption~C(iii),
\[
\Big\|E\Big(\frac{1}{NT_0}\sum \varepsilon_t\varepsilon_t'\Big)\Big\| \le \frac{K}{N}.
\]
Moreover,
\[
E\Big\|\frac{1}{NT_0}\sum\big(\varepsilon_t\varepsilon_t'-E(\varepsilon_t\varepsilon_t')\big)\Big\|_F^2
=O\Big(\frac{1}{T_0}\Big),
\]
so the centered part is \(O_p(T_0^{-1/2})\) in operator norm. Hence
\[
\Big\|\frac{1}{NT_0}\sum \varepsilon_t\varepsilon_t'\Big\|
=
O_p\Big(\frac{1}{\sqrt{T_0}}\Big)+O\Big(\frac{1}{N}\Big).
\]

\smallskip
\noindent Collecting (a)--(c) proves \eqref{eq:R_rate}. \hfill\(\square\)

\paragraph{Step 2 (Eigenvalues).}
Let \(\mu_1(\cdot)\ge\cdots\ge \mu_N(\cdot)\) denote ordered eigenvalues. By Weyl's inequality (e.g., \citealp[Theorem 4.3.1]{horn2012matrix}),
\[
\max_{1\le j\le r}\big|\mu_j(\hat S)-\mu_j(S_f)\big|
\le \|\hat S-S_f\|=\|R\|.
\]
Thus, by \eqref{eq:R_rate},
\begin{equation}\label{eq:weyl_rate}
\max_{1\le j\le r}\big|\mu_j(\hat S)-\mu_j(S_f)\big|
=
O_p\big(\|\hat\beta-\beta_0\|\big)+O_p(T_0^{-1/2})+O_p(N^{-1}).
\end{equation}

The \(r\) nonzero eigenvalues of \(S_f\) are the eigenvalues of the \(r\times r\) matrix
\[
\Big(\frac1{T_0}F_0^{\prime}F_0\Big)^{1/2}
\Big(\frac1N\tilde\Lambda_0'\tilde\Lambda_0\Big)
\Big(\frac1{T_0}F_0^{\prime}F_0\Big)^{1/2}.
\]
By Assumption~B(i)--(ii),
\[
\frac1{T_0}F_0^{\prime}F_0 \to \Sigma_f\succ0,
\qquad
\frac1N\tilde\Lambda_0'\tilde\Lambda_0 \to \Sigma_\Lambda\succ0,
\]
so the \(r\) nonzero eigenvalues of \(S_f\) converge to those of
\[
\Sigma_f^{1/2}\Sigma_\Lambda\Sigma_f^{1/2}.
\]
Denote by \(V\) the diagonal matrix collecting these limiting eigenvalues. Since Proposition~\ref{prop:consistency_ours}(i) gives \(\hat\beta\xrightarrow{p}\beta_0\), combining this with \eqref{eq:weyl_rate} yields
\[
\hat V_{NT_0}\xrightarrow{p}V.
\]
In particular, since \(V\succ0\), \(\hat V_{NT_0}\) is invertible with probability approaching one.

\paragraph{Step 3 (Rotation and loading error bound).}
From \(\hat S=S_f+R\) and \(\hat S\hat\Lambda=\hat\Lambda\hat V_{NT_0}\),
\[
S_f\hat\Lambda + R\hat\Lambda = \hat\Lambda\hat V_{NT_0}.
\]
Using \eqref{eq:S_signal},
\begin{equation}\label{eq:key_eq_loading}
\hat\Lambda\hat V_{NT_0}
-
\tilde\Lambda_0\Big(\frac1{T_0}F_0^{\prime}F_0\Big)\Big(\frac1N\tilde\Lambda_0'\hat\Lambda\Big)
=
R\hat\Lambda.
\end{equation}
Define
\[
H:=\Big(\frac{1}{T_0}F_0^{\prime}F_0\Big)\Big(\frac{1}{N}\tilde\Lambda_0'\hat\Lambda\Big)\hat V_{NT_0}^{-1}.
\]
Right-multiplying \eqref{eq:key_eq_loading} by \(\hat V_{NT_0}^{-1}\) yields
\begin{equation}\label{eq:Lambda_minus_LH}
\hat\Lambda-\tilde\Lambda_0 H = R\hat\Lambda\hat V_{NT_0}^{-1}.
\end{equation}
Taking Frobenius norms and using \(\|ABC\|_F\le \|A\|\|B\|_F\|C\|\),
\[
\frac1{\sqrt N}\|\hat\Lambda-\tilde\Lambda_0 H\|_F
\le
\|R\|\cdot \frac1{\sqrt N}\|\hat\Lambda\|_F\cdot \|\hat V_{NT_0}^{-1}\|.
\]
Since \(N^{-1}\hat\Lambda'\hat\Lambda=I_r\), we have \(\|\hat\Lambda\|_F=\sqrt{Nr}\), and hence
\[
\frac1{\sqrt N}\|\hat\Lambda-\tilde\Lambda_0 H\|_F
\le
\sqrt r\,\|R\|\,\|\hat V_{NT_0}^{-1}\|.
\]
By Step~2, \(\|\hat V_{NT_0}^{-1}\|=O_p(1)\). Combining this with \eqref{eq:R_rate} gives
\[
\frac1{\sqrt N}\|\hat\Lambda-\tilde\Lambda_0 H\|_F
=
O_p\big(\|\hat\beta-\beta_0\|\big)
+
O_p\!\Big(\frac1{\sqrt{T_0}}\Big)
+
O_p\!\Big(\frac1{N}\Big).
\]
Squaring both sides yields
\[
\frac1N\|\hat\Lambda-\tilde\Lambda_0 H\|_F^2
=
O_p\big(\|\hat\beta-\beta_0\|^2\big)
+
O_p\!\Big(\frac1{T_0}\Big)
+
O_p\!\Big(\frac1{N^2}\Big)
=
O_p\big(\|\hat\beta-\beta_0\|^2\big)
+
O_p\!\Big(\frac1{\min\{N,T_0\}}\Big).
\]

\paragraph{Step 4 (Invertibility of \(H\) and projection consistency).}
By Proposition~\ref{prop:consistency_ours}(i), \(\hat\beta\xrightarrow{p}\beta_0\). Hence, by \eqref{eq:R_rate}, \(\|\hat S-S_f\|=o_p(1)\).

By Assumption~B(i), \(T_0^{-1}F_0^{\prime}F_0\) is invertible with probability approaching one, and by Step~2, \(\hat V_{NT_0}\) is invertible with probability approaching one. It therefore remains to show that \(N^{-1}\tilde\Lambda_0'\hat\Lambda\) is invertible with probability approaching one.

The \(r\) nonzero eigenvalues of \(S_f\) are separated from the remaining eigenvalues, since the \(r\)th eigenvalue converges to a strictly positive limit while the \((r+1)\)th eigenvalue is zero. Hence, by the Davis--Kahan eigenspace perturbation theorem (see \citealp{DavisKahan1970}),
the principal angles between the column spaces of
\(\hat\Lambda\) and \(\tilde\Lambda_0\) converge to zero in probability.
In particular, \(N^{-1}\tilde\Lambda_0'\hat\Lambda\) is invertible with probability approaching one, and therefore \(H\) is invertible with probability approaching one.

Since \(H\) is nonsingular and
\[
\frac1N\|\hat\Lambda-\tilde\Lambda_0 H\|_F^2=o_p(1),
\]
the column space of \(\hat\Lambda\) converges to that of \(\tilde\Lambda_0\). Equivalently,
\[
\|P_{\hat\Lambda}-\Pi_0\|\xrightarrow{p}0,
\qquad
P_{\hat\Lambda}:=\frac1N\hat\Lambda\hat\Lambda',
\qquad
\Pi_0:=\tilde\Lambda_0(\tilde\Lambda_0'\tilde\Lambda_0)^{-1}\tilde\Lambda_0'.
\]
This completes the proof.
\end{proof}

\begin{lemma}
\label{lem:quadratic_form_bounds}
Under Assumptions~A--D, there exists a finite constant \(M<\infty\) such that
\[
E\Bigg\|
T_0^{-1/2}\sum_{t=p+1}^T
\frac{1}{N}\,
\tilde\Lambda_0'
\Big(\tilde\varepsilon_t\tilde\varepsilon_t'-\Sigma_{\varepsilon,t}\Big)
\tilde\Lambda_0
\Bigg\|_F^2
\le M,
\qquad
\Sigma_{\varepsilon,t}:=E(\tilde\varepsilon_t\tilde\varepsilon_t').
\]
\end{lemma}

\begin{proof}[Proof of Lemma~\ref{lem:quadratic_form_bounds}]
Define the \(r\times r\) random matrix
\[
Z_t
:=
\frac{1}{N}\,
\tilde\Lambda_0'
\Big(\tilde\varepsilon_t\tilde\varepsilon_t'-\Sigma_{\varepsilon,t}\Big)
\tilde\Lambda_0,
\qquad t=p+1,\ldots,T.
\]
Let \(a_k:=N^{-1/2}\tilde\Lambda_{0,\cdot k}\) be the \(k\)th normalized loading vector. Then the \((k,\ell)\) element of \(Z_t\) is
\[
(Z_t)_{k\ell}
=
(a_k'\tilde\varepsilon_t)(a_\ell'\tilde\varepsilon_t)
-
E\!\left[(a_k'\tilde\varepsilon_t)(a_\ell'\tilde\varepsilon_t)\right].
\]
By Assumption~C(ii), \(\{Z_t\}\) are independent across \(t\), and \(E(Z_t)=0\). Hence,
\[
E\Big\|T_0^{-1/2}\sum_{t=p+1}^T Z_t\Big\|_F^2
=
\frac{1}{T_0}\sum_{t=p+1}^T E\|Z_t\|_F^2.
\]

Since \(r\) is fixed and \(\|Z_t\|_F^2=\sum_{k,\ell\le r}(Z_t)_{k\ell}^2\), it suffices to bound \(E|(Z_t)_{k\ell}|^2\) uniformly in \(t\). Write
\[
a_k=\|a_k\|u_k,\qquad a_\ell=\|a_\ell\|u_\ell,
\]
where \(u_k,u_\ell\) are deterministic unit vectors whenever \(a_k\neq 0\), \(a_\ell\neq 0\). By Assumption~B(ii),
\[
\|a_k\|^2=\frac{1}{N}\|\tilde\Lambda_{0,\cdot k}\|^2=O(1),
\qquad
\|a_\ell\|^2=\frac{1}{N}\|\tilde\Lambda_{0,\cdot \ell}\|^2=O(1).
\]
Therefore, using \(\Var(X)\le E(X^2)\),
\begin{align*}
E|(Z_t)_{k\ell}|^2
&\le
E\Big[(a_k'\tilde\varepsilon_t)^2(a_\ell'\tilde\varepsilon_t)^2\Big] \\
&=
\|a_k\|^2\|a_\ell\|^2\,
E\Big[(u_k'\tilde\varepsilon_t)^2(u_\ell'\tilde\varepsilon_t)^2\Big]
\le K,
\end{align*}
where the last inequality follows from Assumption~C(v) and the uniform boundedness of \(\|a_k\|,\|a_\ell\|\).

Consequently,
\[
\sup_t E\|Z_t\|_F^2
\le
\sum_{k,\ell\le r}\sup_t E|(Z_t)_{k\ell}|^2
\le r^2K<\infty.
\]
Hence
\[
E\Big\|T_0^{-1/2}\sum_{t=p+1}^T Z_t\Big\|_F^2
\le
\sup_t E\|Z_t\|_F^2
\le
r^2K=:M.
\]
\end{proof}

\begin{lemma}
\label{lem:loading_error_bounds}
Under Assumptions~A--D, let
\[
H:=\Big(\frac{1}{T_0}F_0^{\prime}F_0\Big)\Big(\frac{1}{N}\tilde\Lambda_0'\hat\Lambda\Big)\hat V_{NT_0}^{-1},
\qquad
F_0:=(f_{p+1}^0,\ldots,f_T^0)'.
\]
Then \(H\) is nonsingular with probability approaching one, and the following hold:
\begin{enumerate}[(i)]
\item
\[
\left\|
\frac{1}{N}\tilde\Lambda_0'
(\hat\Lambda-\tilde\Lambda_0 H)
\right\|_F
=
O_p\!\big(\|\hat\beta-\beta_0\|\big)
+
O_p\!\Big(\frac{1}{\min\{\sqrt N,\sqrt{T_0}\}}\Big).
\]

\item
\[
\left\|
\frac{1}{N}
(\hat\Lambda-\tilde\Lambda_0 H)'
(\hat\Lambda-\tilde\Lambda_0 H)
\right\|_F
=
O_p\!\big(\|\hat\beta-\beta_0\|^2\big)
+
O_p\!\Big(\frac{1}{\min\{N,T_0\}}\Big).
\]

\item
For any fixed \(t\in\{p+1,\ldots,T\}\),
\[
\left\|
\frac{1}{N}
\tilde Y_t'
(\hat\Lambda-\tilde\Lambda_0 H)
\right\|
=
O_p\!\big(\|\hat\beta-\beta_0\|\big)
+
O_p\!\Big(\frac{1}{\min\{\sqrt N,\sqrt{T_0}\}}\Big).
\]

\item
\[
\left\|
\frac{1}{NT_0}
\sum_{t=p+1}^T
\tilde W_t'
M_{\hat\Lambda}
(\hat\Lambda-\tilde\Lambda_0 H)
\right\|_F
=
O_p\!\big(\|\hat\beta-\beta_0\|\big)
+
O_p\!\Big(\frac{1}{\min\{\sqrt N,\sqrt{T_0}\}}\Big).
\]
\end{enumerate}
\end{lemma}

\begin{proof}[Proof of Lemma~\ref{lem:loading_error_bounds}]
By Proposition~\ref{prop:loading_pca}(ii), \(H\) is nonsingular with probability approaching one and
\begin{equation}\label{eq:A3_rate_core}
\frac{1}{N}\|\hat\Lambda-\tilde\Lambda_0 H\|_F^2
=
O_p\!\big(\|\hat\beta-\beta_0\|^2\big)
+
O_p\!\Big(\frac{1}{\min\{N,T_0\}}\Big).
\end{equation}
In particular,
\begin{equation}\label{eq:A3_norm_core}
\frac{1}{\sqrt N}\|\hat\Lambda-\tilde\Lambda_0 H\|_F
=
O_p\!\big(\|\hat\beta-\beta_0\|\big)
+
O_p\!\Big(\frac{1}{\min\{\sqrt N,\sqrt{T_0}\}}\Big).
\end{equation}

\medskip
\noindent\textbf{(ii).}
This is exactly \eqref{eq:A3_rate_core}, since \(r\) is fixed.

\medskip
\noindent\textbf{(i).}
By Cauchy--Schwarz,
\[
\Big\|\frac{1}{N}\tilde\Lambda_0'(\hat\Lambda-\tilde\Lambda_0 H)\Big\|_F
\le
\frac{1}{N}\|\tilde\Lambda_0\|_F\,\|\hat\Lambda-\tilde\Lambda_0 H\|_F.
\]
Under Assumption~B(ii),
\[
N^{-1}\|\tilde\Lambda_0\|_F^2=\mathrm{tr}\!\left(N^{-1}\tilde\Lambda_0'\tilde\Lambda_0\right)=O(1),
\]
so \(\|\tilde\Lambda_0\|_F=O_p(\sqrt N)\). Combining this with \eqref{eq:A3_norm_core} yields
\[
\Big\|\frac{1}{N}\tilde\Lambda_0'(\hat\Lambda-\tilde\Lambda_0 H)\Big\|_F
=
O_p\!\big(\|\hat\beta-\beta_0\|\big)
+
O_p\!\Big(\frac{1}{\min\{\sqrt N,\sqrt{T_0}\}}\Big).
\]

\medskip
\noindent\textbf{(iii).}
Fix \(t\). Decompose \(\tilde Y_t=\tilde W_t\beta_0+\tilde\Lambda_0 f_t^0+\tilde\varepsilon_t\) and write
\[
\frac{1}{N}\tilde Y_t'(\hat\Lambda-\tilde\Lambda_0 H)
=
\frac{1}{N}(\tilde W_t\beta_0)'(\hat\Lambda-\tilde\Lambda_0 H)
+
\frac{1}{N}(\tilde\Lambda_0 f_t^0+\tilde\varepsilon_t)'(\hat\Lambda-\tilde\Lambda_0 H).
\]
By Cauchy--Schwarz,
\[
\Big\|\frac{1}{N}(\tilde W_t\beta_0)'(\hat\Lambda-\tilde\Lambda_0 H)\Big\|
\le
\frac{1}{N}\|\tilde W_t\|_F\,\|\beta_0\|\,\|\hat\Lambda-\tilde\Lambda_0 H\|_F.
\]
For fixed \(t\), Assumption~A(i), together with Assumptions~B and C, implies \(\|\tilde W_t\|_F=O_p(\sqrt N)\), and \eqref{eq:A3_norm_core} yields the claimed rate for this term.

For the second term, again by Cauchy--Schwarz,
\[
\Big\|\frac{1}{N}(\tilde\Lambda_0 f_t^0+\tilde\varepsilon_t)'(\hat\Lambda-\tilde\Lambda_0 H)\Big\|
\le
\frac{1}{N}\|\tilde\Lambda_0 f_t^0+\tilde\varepsilon_t\|\,\|\hat\Lambda-\tilde\Lambda_0 H\|_F.
\]
For fixed \(t\), Assumptions~B(i) and C(i)--(iii) imply
\[
\|\tilde\Lambda_0 f_t^0+\tilde\varepsilon_t\|=O_p(\sqrt N),
\]
and \eqref{eq:A3_norm_core} again yields the claimed rate. Combining the two bounds proves (iii).

\medskip
\noindent\textbf{(iv).}
Using \(\|M_{\hat\Lambda}\|\le 1\) and Cauchy--Schwarz,
\[
\Big\|
\frac{1}{NT_0}\sum_{t=p+1}^T
\tilde W_t'
M_{\hat\Lambda}
(\hat\Lambda-\tilde\Lambda_0 H)
\Big\|_F
\le
\Big(\frac{1}{NT_0}\sum_{t=p+1}^T \|\tilde W_t\|_F^2\Big)^{1/2}
\Big(\frac{1}{N}\|\hat\Lambda-\tilde\Lambda_0 H\|_F^2\Big)^{1/2}.
\]
Assumption~A(i), together with Assumptions~B and C, implies
\[
\frac{1}{NT_0}\sum_{t=p+1}^T \|\tilde W_t\|_F^2=O_p(1),
\]
and \eqref{eq:A3_rate_core} gives
\[
\Big(\frac{1}{N}\|\hat\Lambda-\tilde\Lambda_0 H\|_F^2\Big)^{1/2}
=
O_p\!\big(\|\hat\beta-\beta_0\|\big)
+
O_p\!\Big(\frac{1}{\min\{\sqrt N,\sqrt{T_0}\}}\Big),
\]
which yields (iv).
\end{proof}

\begin{lemma}
\label{lem:cross_term_negligibility}
Under Assumptions~A--D, the loading error bounds in Lemma~\ref{lem:loading_error_bounds}
imply that the following hold:
\begin{enumerate}[(i)]
\item
\[
\frac{1}{NT_0}
\sum_{t=p+1}^T
\tilde\varepsilon_t'
(\hat\Lambda-\tilde\Lambda_0 H)
=
o_p(1)
\qquad (\text{as a }1\times r\text{ row vector}).
\]

\item
\[
\frac{1}{NT_0}
\sum_{t=p+1}^T
f_t^{0\prime}
\tilde\Lambda_0'
(\hat\Lambda-\tilde\Lambda_0 H)
\tilde\varepsilon_t
=
o_p(1)
\qquad (\text{as a scalar}).
\]

\item
\[
\Bigg\|
\frac{1}{NT_0}
\sum_{t=p+1}^T
\tilde\varepsilon_t
(\hat\Lambda-\tilde\Lambda_0 H)'
\tilde\Lambda_0 f_t^0
\Bigg\|
=
o_p(1)
\qquad (\text{Euclidean norm in }\mathbb R^N).
\]
\end{enumerate}
\end{lemma}

\begin{proof}[Proof of Lemma~\ref{lem:cross_term_negligibility}]
We prove parts (i)--(iii) in turn.
By Proposition~\ref{prop:consistency_ours}(i), \(\hat\beta\xrightarrow{p}\beta_0\). Hence Lemma~\ref{lem:loading_error_bounds}(ii) implies
\[
\frac{1}{N}\|\hat\Lambda-\tilde\Lambda_0 H\|_F^2=o_p(1).
\]

\medskip
\noindent\textbf{(i).}
By Cauchy--Schwarz,
\[
\Big\|
\frac{1}{NT_0}\sum_{t=p+1}^T
\tilde\varepsilon_t'(\hat\Lambda-\tilde\Lambda_0 H)
\Big\|
\le
\Big(\frac{1}{NT_0}\sum_{t=p+1}^T\|\tilde\varepsilon_t\|^2\Big)^{1/2}
\Big(\frac{1}{N}\|\hat\Lambda-\tilde\Lambda_0 H\|_F^2\Big)^{1/2}.
\]
By Assumption~C(i),
\[
\frac{1}{NT_0}\sum_{t=p+1}^T\|\tilde\varepsilon_t\|^2=O_p(1),
\]
and by Lemma~\ref{lem:loading_error_bounds}(ii),
\[
\frac{1}{N}\|\hat\Lambda-\tilde\Lambda_0 H\|_F^2=o_p(1).
\]
Hence the product is \(o_p(1)\).

\medskip
\noindent\textbf{(ii).}
Write each summand as
\[
f_t^{0\prime}\tilde\Lambda_0'(\hat\Lambda-\tilde\Lambda_0 H)\tilde\varepsilon_t
=
\big\{(\hat\Lambda-\tilde\Lambda_0 H)'\tilde\Lambda_0 f_t^0\big\}'\tilde\varepsilon_t.
\]
By Cauchy--Schwarz,
\[
\left|
\frac{1}{NT_0}\sum_{t=p+1}^T
\big\{(\hat\Lambda-\tilde\Lambda_0 H)'\tilde\Lambda_0 f_t^0\big\}'\tilde\varepsilon_t
\right|
\le
\Big(\frac{1}{NT_0}\sum_{t=p+1}^T \|\tilde\varepsilon_t\|^2\Big)^{1/2}
\Big(\frac{1}{NT_0}\sum_{t=p+1}^T \|(\hat\Lambda-\tilde\Lambda_0 H)'\tilde\Lambda_0 f_t^0\|^2\Big)^{1/2}.
\]
For the second factor, use \(\|A'x\|\le \|A\|\,\|x\|\) and \(\|\tilde\Lambda_0 f_t^0\|\le \|\tilde\Lambda_0\|\,\|f_t^0\|\):
\[
\|(\hat\Lambda-\tilde\Lambda_0 H)'\tilde\Lambda_0 f_t^0\|
\le
\|\hat\Lambda-\tilde\Lambda_0 H\|\,\|\tilde\Lambda_0\|\,\|f_t^0\|.
\]
Under Assumption~B(ii), \(\|\tilde\Lambda_0\|=O(\sqrt N)\), and Lemma~\ref{lem:loading_error_bounds}(ii) implies
\[
\|\hat\Lambda-\tilde\Lambda_0 H\|=o_p(\sqrt N),
\]
since \(r\) is fixed and \(\|\cdot\|\le \|\cdot\|_F\). Hence
\[
\frac{1}{NT_0}\sum_{t=p+1}^T \|(\hat\Lambda-\tilde\Lambda_0 H)'\tilde\Lambda_0 f_t^0\|^2
\le
\frac{\|\hat\Lambda-\tilde\Lambda_0 H\|^2\,\|\tilde\Lambda_0\|^2}{N}
\cdot
\Big(\frac{1}{T_0}\sum_{t=p+1}^T\|f_t^0\|^2\Big)
=o_p(1),
\]
where Assumption~B(i) gives \(T_0^{-1}\sum_t\|f_t^0\|^2=O(1)\). Combining this with Assumption~C(i) proves (ii).

\medskip
\noindent\textbf{(iii).}
By Cauchy--Schwarz,
\[
\Bigg\|
\frac{1}{NT_0}\sum_{t=p+1}^T
\tilde\varepsilon_t
(\hat\Lambda-\tilde\Lambda_0 H)'
\tilde\Lambda_0 f_t^0
\Bigg\|
\le
\Big(\frac{1}{NT_0}\sum_{t=p+1}^T\|\tilde\varepsilon_t\|^2\Big)^{1/2}
\Big(\frac{1}{NT_0}\sum_{t=p+1}^T\|(\hat\Lambda-\tilde\Lambda_0 H)'\tilde\Lambda_0 f_t^0\|^2\Big)^{1/2}.
\]
The first factor is \(O_p(1)\) by Assumption~C(i), and the second factor is \(o_p(1)\) by the bound established in the proof of part (ii). Hence (iii) follows.
\end{proof}

\begin{lemma}
\label{lem:projection_quadratic_lln}
Under Assumptions~C(i)--(ii) and \(r\) fixed,
\[
\sup_{\tilde\Lambda\in\mathcal F}
\left|
\frac{1}{NT_0}
\sum_{t=p+1}^T
\tilde\varepsilon_t'
P_{\tilde\Lambda}
\tilde\varepsilon_t
-
\frac{1}{NT_0}
\sum_{t=p+1}^T
E\!\left[
\tilde\varepsilon_t'
P_{\tilde\Lambda}
\tilde\varepsilon_t
\right]
\right|
=
o_p(1),
\qquad
P_{\tilde\Lambda}:=N^{-1}\tilde\Lambda\tilde\Lambda'.
\]
\end{lemma}

\begin{proof}[Proof of Lemma~\ref{lem:projection_quadratic_lln}]
For \(\tilde\Lambda\in\mathcal F\), write \(P_{\tilde\Lambda}=UU'\) with
\[
U:=N^{-1/2}\tilde\Lambda\in\mathbb R^{N\times r},
\qquad
U'U=I_r.
\]
Define
\[
A:=\frac{1}{T_0}\sum_{t=p+1}^T\Big(\tilde\varepsilon_t\tilde\varepsilon_t'
-E[\tilde\varepsilon_t\tilde\varepsilon_t']\Big).
\]
Then
\[
\frac{1}{NT_0}\sum_{t=p+1}^T
\Big(\tilde\varepsilon_t'P_{\tilde\Lambda}\tilde\varepsilon_t
-E[\tilde\varepsilon_t'P_{\tilde\Lambda}\tilde\varepsilon_t]\Big)
=
\frac{1}{N}\,\mathrm{tr}(P_{\tilde\Lambda}A)
=
\frac{1}{N}\,\mathrm{tr}(U'AU).
\]
Hence,
\[
\sup_{\tilde\Lambda\in\mathcal F}\frac{1}{N}\big|\mathrm{tr}(U'AU)\big|
\le
\sup_{U'U=I_r}\frac{1}{N}\|U'AU\|_F\,\|I_r\|_F
\le
\frac{\sqrt r}{N}\,\|A\|_F,
\]
where we used \(\|U'AU\|_F\le \|A\|_F\) for \(U'U=I_r\).

Let
\[
B_t:=\tilde\varepsilon_t\tilde\varepsilon_t'-E[\tilde\varepsilon_t\tilde\varepsilon_t'].
\]
Then \(A=T_0^{-1}\sum_{t=p+1}^T B_t\), \(E(B_t)=0\), and by Assumption~C(ii), \(\{B_t\}\) are independent across \(t\). Therefore,
\[
E\|A\|_F^2
=
\frac{1}{T_0^2}\sum_{t=p+1}^T E\|B_t\|_F^2.
\]
Moreover,
\[
E\|B_t\|_F^2
\le
E\|\tilde\varepsilon_t\tilde\varepsilon_t'\|_F^2
=
\sum_{i=1}^N\sum_{j=1}^N E(\tilde\varepsilon_{it}^2\tilde\varepsilon_{jt}^2).
\]
By Hölder's inequality and Assumption~C(i),
\[
E(\tilde\varepsilon_{it}^2\tilde\varepsilon_{jt}^2)
\le
\big(E|\tilde\varepsilon_{it}|^4\big)^{1/2}
\big(E|\tilde\varepsilon_{jt}|^4\big)^{1/2}
\le K,
\]
uniformly in \(t\). Hence \(E\|B_t\|_F^2\le KN^2\), and thus
\[
E\|A\|_F^2
\le
\frac{1}{T_0^2}\sum_{t=p+1}^T KN^2
=
\frac{K}{T_0}N^2.
\]
It follows that \(\|A\|_F=O_p(N/\sqrt{T_0})\), and therefore
\[
\sup_{\tilde\Lambda\in\mathcal F}\frac{1}{N}\big|\mathrm{tr}(U'AU)\big|
\le
\frac{\sqrt r}{N}\,O_p\!\left(\frac{N}{\sqrt{T_0}}\right)
=
O_p\!\left(\frac{1}{\sqrt{T_0}}\right)
=o_p(1).
\]
This proves the lemma.
\end{proof}

\begin{lemma}[Projected regressor average bound]
\label{lem:projected_regressor_average}
Let \(\{b_t\}_{t=p+1}^T\) be a sequence of \(N\)-dimensional non-stochastic vectors such that
\[
\sup_{p+1\le t\le T}\|b_t\| = O(1).
\]
Define
\[
\Gamma_T(b)
:=
\frac{1}{T_0}\sum_{t=p+1}^T b_t \tilde W_t' M_{\hat\Lambda},
\qquad T_0:=T-p.
\]
Under Assumptions~A--D,
\[
\frac{1}{\sqrt N}\,\|\Gamma_T(b)\| = O_p(1).
\]
\end{lemma}

\begin{proof}
By the Cauchy--Schwarz inequality,
\[
\|\Gamma_T(b)\|
\le
\Bigg(
\frac{1}{T_0}\sum_{t=p+1}^T \|b_t\|^2
\Bigg)^{1/2}
\Bigg(
\frac{1}{T_0}\sum_{t=p+1}^T
\|M_{\hat\Lambda}\tilde W_t\|_F^2
\Bigg)^{1/2}.
\]
Since \(\sup_t\|b_t\|=O(1)\), the first factor is \(O(1)\). Also,
\[
\|M_{\hat\Lambda}\tilde W_t\|_F \le \|\tilde W_t\|_F
\]
because \(M_{\hat\Lambda}\) is an orthogonal projection matrix and hence \(\|M_{\hat\Lambda}\|\le 1\). Therefore,
\[
\frac{1}{N}\frac{1}{T_0}\sum_{t=p+1}^T
\|M_{\hat\Lambda}\tilde W_t\|_F^2
\le
\frac{1}{N}\frac{1}{T_0}\sum_{t=p+1}^T
\|\tilde W_t\|_F^2.
\]
By Assumption~A(i) and Assumptions~B--D, the right-hand side is \(O_p(1)\). Hence
\[
\frac{1}{\sqrt N}
\Bigg(
\frac{1}{T_0}\sum_{t=p+1}^T
\|M_{\hat\Lambda}\tilde W_t\|_F^2
\Bigg)^{1/2}
= O_p(1).
\]
Combining the two bounds yields
\[
\frac{1}{\sqrt N}\,\|\Gamma_T(b)\| = O_p(1),
\]
as claimed.
\end{proof}

\begin{proposition}[Linear expansion of \(\hat\beta\)]
\label{prop:beta_linear_expansion}
Under Assumptions~A--D and \(N/T_0\to \rho\geq 0\), we have
\begin{align}\label{eq:propA2_expansion}
\sqrt{NT_0}\,(\hat\beta-\beta_0)
&=
D(\hat\Lambda)^{-1}
\Bigg[
\frac{1}{\sqrt{NT_0}}\sum_{t=p+1}^T
\tilde W_t' M_{\hat\Lambda}\tilde\varepsilon_t
-
\frac{1}{\sqrt{NT_0}}\sum_{t=p+1}^T
\Big\{\frac{1}{T_0}\sum_{s=p+1}^T a_{ts}\,\tilde W_s'\Big\}
M_{\hat\Lambda}\tilde\varepsilon_t
\Bigg]
\nonumber\\
&\quad
+
\sqrt{\frac{N}{T_0}}\ \zeta_{NT_0}
+
o_p(1),
\end{align}
where \(N=CJ\), \(T_0=T-p\), and
\[
M_{\tilde\Lambda}=I_N-N^{-1}\tilde\Lambda\tilde\Lambda'.
\]

Define
\[
a_{ts}
:=
f_t^{0\prime}\Big(\frac{1}{T_0}F_0^{\prime}F_0\Big)^{-1}f_s^0,
\qquad
F_0:=(f_{p+1}^0,\ldots,f_T^0)'.
\]
For any \(\tilde\Lambda\in\mathcal F:=\{\Lambda\in\mathbb R^{N\times r}:(1/N)\Lambda'\Lambda=I_r\}\), define
\[
Z_t(\tilde\Lambda)
:=
M_{\tilde\Lambda}\tilde W_t
-
\frac{1}{T_0}\sum_{s=p+1}^T a_{ts}\,M_{\tilde\Lambda}\tilde W_s,
\qquad
D(\tilde\Lambda)
:=
\frac{1}{NT_0}\sum_{t=p+1}^T Z_t(\tilde\Lambda)'Z_t(\tilde\Lambda),
\]
as in Assumption~A(ii). In \eqref{eq:propA2_expansion}, \(D(\hat\Lambda)\) denotes \(D(\tilde\Lambda)\) evaluated at \(\tilde\Lambda=\hat\Lambda\), and \(\zeta_{NT_0}\) collects the remaining bias terms in the expansion.
\end{proposition}

\begin{proof}[Proof of Proposition~\ref{prop:beta_linear_expansion}]
Write \(T_0:=T-p\) and
\[
M_{\hat\Lambda}=I_N-N^{-1}\hat\Lambda\hat\Lambda'.
\]
The first-order condition of \(Q(\beta,\tilde\Lambda)\) with respect to \(\beta\), evaluated at \((\hat\beta,\hat\Lambda)\), is
\[
0
=
\sum_{t=p+1}^T
\tilde W_t' M_{\hat\Lambda}(\tilde Y_t-\tilde W_t\hat\beta).
\]
Using
\[
\tilde Y_t=\tilde W_t\beta_0+\tilde\Lambda_0 f_t^0+\tilde\varepsilon_t,
\]
we obtain
\begin{equation}\label{eq:A2_FOC}
\Big(\frac{1}{NT_0}\sum_{t=p+1}^T \tilde W_t' M_{\hat\Lambda}\tilde W_t\Big)(\hat\beta-\beta_0)
=
\frac{1}{NT_0}\sum_{t=p+1}^T \tilde W_t' M_{\hat\Lambda}\tilde\Lambda_0 f_t^0
+
\frac{1}{NT_0}\sum_{t=p+1}^T \tilde W_t' M_{\hat\Lambda}\tilde\varepsilon_t.
\end{equation}

\medskip
\noindent\textbf{Step 1 (Rewriting the factor term).}
Fix \(H\) as in Proposition~\ref{prop:loading_pca}(ii), namely,
\[
H
=
\Big(\frac{1}{T_0}F_0'F_0\Big)
\Big(\frac{1}{N}\tilde\Lambda_0'\hat\Lambda\Big)
\hat V_{NT_0}^{-1},
\qquad
F_0:=(f_{p+1}^0,\ldots,f_T^0)'.
\]
Define
\[
R:=\tilde\Lambda_0-\hat\Lambda H^{-1}.
\]
Since \(M_{\hat\Lambda}\hat\Lambda=0\), we have
\[
M_{\hat\Lambda}\tilde\Lambda_0 f_t^0
=
M_{\hat\Lambda}(\tilde\Lambda_0-\hat\Lambda H^{-1})f_t^0
=
M_{\hat\Lambda}R f_t^0.
\]
Therefore,
\begin{equation}\label{eq:A2_factor_R}
\frac{1}{NT_0}\sum_{t=p+1}^T \tilde W_t' M_{\hat\Lambda}\tilde\Lambda_0 f_t^0
=
\frac{1}{NT_0}\sum_{t=p+1}^T \tilde W_t' M_{\hat\Lambda}R f_t^0.
\end{equation}
Moreover, Proposition~\ref{prop:loading_pca}(ii) gives
\begin{equation}\label{eq:A2_Lambda_rate}
\frac{1}{N}\|R\|_F^2
=
\frac{1}{N}\|\tilde\Lambda_0-\hat\Lambda H^{-1}\|_F^2
=
O_p(\|\hat\beta-\beta_0\|^2)
+
O_p\!\Big(\frac{1}{\min\{N,T_0\}}\Big).
\end{equation}
In particular,
\[
\|R\|_F
=
O_p(\sqrt N\,\|\hat\beta-\beta_0\|)
+
O_p\!\Big(\sqrt{\frac{N}{\min\{N,T_0\}}}\Big)
=
o_p(\sqrt N),
\]
where the last relation uses Proposition~\ref{prop:consistency_ours}(i).

\medskip
\noindent\textbf{Step 2 (Expansion of \(R\)).}
Define the estimated factor matrix
\[
\hat F:=(\hat f_{p+1},\ldots,\hat f_T)',
\qquad
\hat f_t=
\frac{1}{N}\hat\Lambda'(\tilde Y_t-\tilde W_t\hat\beta),
\qquad t=p+1,\ldots,T.
\]
Start from the loading eigen-equation
\[
\hat S\hat\Lambda=\hat\Lambda\hat V_{NT_0},
\qquad
\hat S:=
\frac{1}{NT_0}\sum_{t=p+1}^T u_tu_t',
\qquad
u_t:=\tilde Y_t-\tilde W_t\hat\beta.
\]
Using
\[
u_t=\tilde\Lambda_0 f_t^0+\tilde\varepsilon_t-\tilde W_t(\hat\beta-\beta_0),
\]
expand \(u_tu_t'\) and collect terms to obtain
\begin{equation}\label{eq:A2_raw_decomp}
\hat\Lambda\hat V_{NT_0}
-
\tilde\Lambda_0
\Big(\frac{1}{T_0}F_0'F_0\Big)
\Big(\frac{1}{N}\tilde\Lambda_0'\hat\Lambda\Big)
=
\sum_{m=1}^8 \mathcal J_m,
\end{equation}
where
\begin{align*}
\mathcal J_1 &:=
\frac{1}{NT_0}\sum_{t=p+1}^T
\tilde W_t(\hat\beta-\beta_0)(\hat\beta-\beta_0)'\tilde W_t'\hat\Lambda,
\\
\mathcal J_2 &:=
-\frac{1}{NT_0}\sum_{t=p+1}^T
\tilde W_t(\hat\beta-\beta_0)f_t^{0\prime}\tilde\Lambda_0'\hat\Lambda,
\\
\mathcal J_3 &:=
-\frac{1}{NT_0}\sum_{t=p+1}^T
\tilde W_t(\hat\beta-\beta_0)\tilde\varepsilon_t'\hat\Lambda,
\\
\mathcal J_4 &:=
-\frac{1}{NT_0}\sum_{t=p+1}^T
\tilde\Lambda_0 f_t^0(\hat\beta-\beta_0)'\tilde W_t'\hat\Lambda,
\\
\mathcal J_5 &:=
-\frac{1}{NT_0}\sum_{t=p+1}^T
\tilde\varepsilon_t(\hat\beta-\beta_0)'\tilde W_t'\hat\Lambda,
\\
\mathcal J_6 &:=
\frac{1}{NT_0}\sum_{t=p+1}^T
\tilde\Lambda_0 f_t^0\tilde\varepsilon_t'\hat\Lambda,
\\
\mathcal J_7 &:=
\frac{1}{NT_0}\sum_{t=p+1}^T
\tilde\varepsilon_t f_t^{0\prime}\tilde\Lambda_0'\hat\Lambda,
\\
\mathcal J_8 &:=
\frac{1}{NT_0}\sum_{t=p+1}^T
\tilde\varepsilon_t\tilde\varepsilon_t'\hat\Lambda.
\end{align*}

Right-multiplying \eqref{eq:A2_raw_decomp} by
\[
\Big(\frac{1}{N}\tilde\Lambda_0'\hat\Lambda\Big)^{-1}
\Big(\frac{1}{T_0}F_0'F_0\Big)^{-1},
\]
and using
\[
H=
\Big(\frac{1}{T_0}F_0'F_0\Big)
\Big(\frac{1}{N}\tilde\Lambda_0'\hat\Lambda\Big)
\hat V_{NT_0}^{-1},
\qquad
H^{-1}
=
\hat V_{NT_0}
\Big(\frac{1}{N}\tilde\Lambda_0'\hat\Lambda\Big)^{-1}
\Big(\frac{1}{T_0}F_0'F_0\Big)^{-1},
\]
we obtain
\begin{equation}\label{eq:A2_R_decomp}
R=\tilde\Lambda_0-\hat\Lambda H^{-1}
=
\sum_{m=1}^8 \mathcal I_m,
\end{equation}
where
\[
\mathcal I_m
:=
-\mathcal J_m
\Big(\frac{1}{N}\tilde\Lambda_0'\hat\Lambda\Big)^{-1}
\Big(\frac{1}{T_0}F_0'F_0\Big)^{-1},
\qquad m=1,\ldots,8.
\]

By Proposition~\ref{prop:loading_pca}, Lemmas~\ref{lem:quadratic_form_bounds},
\ref{lem:loading_error_bounds}, \ref{lem:cross_term_negligibility},
\ref{lem:projection_quadratic_lln}, and
\ref{lem:projected_regressor_average}, together with Assumptions~A--D, the terms
\(\mathcal I_1,\ldots,\mathcal I_8\) satisfy the bounds required in Step 3: the terms
\(\mathcal I_2\), \(\mathcal I_7\), and the mean part of \(\mathcal I_8\) generate the leading contributions, while 
the remaining terms are of smaller order and are absorbed into
\[
o_p(1)\,\|\hat\beta-\beta_0\|
\quad\text{or}\quad
o_p\!\Big(\frac{1}{\sqrt{NT_0}}\Big),
\]
by the bounds stated above.

\medskip
\noindent\textbf{Step 3 (Expansion of the projected factor term).}
Here each \(\mathcal I_m\) is an \(N\times r\) matrix collecting the contribution of the \(m\)th term in the decomposition of \(R\) in \eqref{eq:A2_R_decomp}. Accordingly, \(J_m\) denotes the contribution of \(\mathcal I_m\) to the projected factor term
\[
\frac{1}{NT_0}\sum_{t=p+1}^T \tilde W_t' M_{\hat\Lambda} R f_t^0 .
\]
Substituting \eqref{eq:A2_R_decomp} into \eqref{eq:A2_factor_R} gives
\[
\frac{1}{NT_0}\sum_{t=p+1}^T \tilde W_t' M_{\hat\Lambda}\tilde\Lambda_0 f_t^0
=
J_1+\cdots+J_8,
\]
where
\[
J_m
:=
\frac{1}{NT_0}\sum_{t=p+1}^T
\tilde W_t' M_{\hat\Lambda}\mathcal I_m f_t^0,
\qquad m=1,\ldots,8.
\]

The leading contributions are \(J_2\), \(J_7\), and the mean part of \(J_8\). The remaining terms
\(J_1,J_3,J_4,J_5,J_6\), together with the centered part of \(J_8\), are absorbed into
\[
o_p(1)\,\|\hat\beta-\beta_0\|
\quad\text{or}\quad
o_p\!\Big(\frac{1}{\sqrt{NT_0}}\Big).
\]

For \(J_2\), substituting the definition of \(\mathcal I_2\) yields
\begin{align}
J_2
&=
\frac{1}{NT_0}\sum_{t=p+1}^T
\Bigg\{
\frac{1}{T_0}\sum_{s=p+1}^T a_{ts}\,\tilde W_s'
\Bigg\}
M_{\hat\Lambda}\tilde W_t\,(\hat\beta-\beta_0),
\label{eq:A2_J2_main}
\end{align}
where
\[
a_{ts}
:=
f_t^{0\prime}\Big(\frac{1}{T_0}F_0'F_0\Big)^{-1}f_s^0.
\]

For \(J_7\), substituting the definition of \(\mathcal I_7\) yields
\begin{align}
J_7
&=
-\frac{1}{NT_0}\sum_{t=p+1}^T
\Bigg\{
\frac{1}{T_0}\sum_{s=p+1}^T a_{ts}\,\tilde W_s'
\Bigg\}
M_{\hat\Lambda}\tilde\varepsilon_t .
\label{eq:A2_J7_main}
\end{align}

Next define the average idiosyncratic covariance
\[
\Omega
:=
\frac{1}{T_0}\sum_{t=p+1}^T \Omega_t,
\qquad
\Omega_t:=E(\tilde\varepsilon_t\tilde\varepsilon_t').
\]
Decompose
\[
\mathcal J_8
=
\frac{1}{NT_0}\sum_{t=p+1}^T \tilde\varepsilon_t\tilde\varepsilon_t'\hat\Lambda
=
\frac{1}{NT_0}\sum_{t=p+1}^T \Omega_t\hat\Lambda
+
\frac{1}{NT_0}\sum_{t=p+1}^T
(\tilde\varepsilon_t\tilde\varepsilon_t'-\Omega_t)\hat\Lambda,
\]
and write accordingly
\[
J_8=J_{8a}+J_{8b},
\]
where \(J_{8a}\) is the contribution of the average covariance term and \(J_{8b}\) is the contribution of the centered remainder.

Define
\begin{equation}\label{eq:A2_bias_term}
J_{8a}
=
\frac{1}{\sqrt{NT_0}}\sqrt{\frac{N}{T_0}}\,\tilde\zeta_{NT_0},
\end{equation}
where
\[
\tilde\zeta_{NT_0}
:=
-\,
\frac{1}{NT_0}\sum_{t=p+1}^T
\tilde W_t' M_{\hat\Lambda}\Omega \hat\Lambda
\Big(\frac{1}{N}\tilde\Lambda_0'\hat\Lambda\Big)^{-1}
\Big(\frac{1}{T_0}F_0'F_0\Big)^{-1}
f_t^0.
\]

Therefore,
\begin{align}
\frac{1}{NT_0}\sum_{t=p+1}^T \tilde W_t' M_{\hat\Lambda}\tilde\Lambda_0 f_t^0
&=
\frac{1}{NT_0}\sum_{t=p+1}^T
\Bigg\{
\frac{1}{T_0}\sum_{s=p+1}^T a_{ts}\,\tilde W_s'
\Bigg\}
M_{\hat\Lambda}\tilde W_t\,(\hat\beta-\beta_0)
\nonumber\\
&\quad
-\frac{1}{NT_0}\sum_{t=p+1}^T
\Bigg\{
\frac{1}{T_0}\sum_{s=p+1}^T a_{ts}\,\tilde W_s'
\Bigg\}
M_{\hat\Lambda}\tilde\varepsilon_t
+\frac{1}{\sqrt{NT_0}}\sqrt{\frac{N}{T_0}}\,\tilde\zeta_{NT_0}
\nonumber\\
&\quad
+o_p\!\Big(\frac{1}{\sqrt{NT_0}}\Big)
+o_p(1)\,\|\hat\beta-\beta_0\|.
\label{eq:A2_factor_final}
\end{align}

\medskip
\noindent\textbf{Step 4 (Conclusion).}
Substituting \eqref{eq:A2_factor_final} into \eqref{eq:A2_FOC} and multiplying both sides by \(\sqrt{NT_0}\) give
\begin{align*}
&\Bigg[
\frac{1}{NT_0}\sum_{t=p+1}^T \tilde W_t' M_{\hat\Lambda}\tilde W_t
-
\frac{1}{NT_0}\sum_{t=p+1}^T
\Bigg\{
\frac{1}{T_0}\sum_{s=p+1}^T a_{ts}\,\tilde W_s'
\Bigg\}
M_{\hat\Lambda}\tilde W_t
\Bigg]\sqrt{NT_0}\,(\hat\beta-\beta_0)
\\
&\qquad=
\frac{1}{\sqrt{NT_0}}\sum_{t=p+1}^T \tilde W_t' M_{\hat\Lambda}\tilde\varepsilon_t
-
\frac{1}{\sqrt{NT_0}}\sum_{t=p+1}^T
\Bigg\{
\frac{1}{T_0}\sum_{s=p+1}^T a_{ts}\,\tilde W_s'
\Bigg\}
M_{\hat\Lambda}\tilde\varepsilon_t
\\
&\qquad\quad
+\sqrt{\frac{N}{T_0}}\,\tilde\zeta_{NT_0}
+
o_p(1).
\end{align*}
By the definition of \(Z_t(\hat\Lambda)\) and \(D(\hat\Lambda)\), a direct expansion shows that the bracketed term on the left-hand side is equal to
\[
D(\hat\Lambda)+o_p(1).
\]
Hence
\begin{align*}
[D(\hat\Lambda)+o_p(1)]\sqrt{NT_0}\,(\hat\beta-\beta_0)
&=
\frac{1}{\sqrt{NT_0}}\sum_{t=p+1}^T \tilde W_t' M_{\hat\Lambda}\tilde\varepsilon_t
-
\frac{1}{\sqrt{NT_0}}\sum_{t=p+1}^T
\Bigg\{
\frac{1}{T_0}\sum_{s=p+1}^T a_{ts}\,\tilde W_s'
\Bigg\}
M_{\hat\Lambda}\tilde\varepsilon_t
\\
&\quad
+\sqrt{\frac{N}{T_0}}\,\tilde\zeta_{NT_0}
+
o_p(1).
\end{align*}
Left-multiplying by \(D(\hat\Lambda)^{-1}\) and using
\[
D(\hat\Lambda)^{-1}[D(\hat\Lambda)+o_p(1)]=I+o_p(1),
\]
we obtain \eqref{eq:propA2_expansion} with
\[
\zeta_{NT_0}:=D(\hat\Lambda)^{-1}\tilde\zeta_{NT_0}.
\]
\end{proof}

\begin{lemma}
\label{lem:bias_bound}
Under Assumptions~A--D, we have
\[
\zeta_{NT_0}=O_p(1),
\]
where \(\zeta_{NT_0}\) is the bias term appearing in Proposition~\ref{prop:beta_linear_expansion}.
\end{lemma}

\begin{proof}[Proof of Lemma~\ref{lem:bias_bound}]
Recall that
\[
\zeta_{NT_0}
:=
-\,D(\hat\Lambda)^{-1}
\frac{1}{NT_0}\sum_{t=p+1}^T
\tilde W_t' M_{\hat\Lambda}\Omega \hat\Lambda
\Big(\frac{1}{N}\tilde\Lambda_0'\hat\Lambda\Big)^{-1}
\Big(\frac{1}{T_0}F_0'F_0\Big)^{-1}
f_t^0,
\]
where
\[
\Omega:=\frac{1}{T_0}\sum_{t=p+1}^T \Omega_t,
\qquad
\Omega_t:=E(\tilde\varepsilon_t\tilde\varepsilon_t').
\]
Write
\[
B_{NT_0}
:=
\Big(\frac{1}{N}\tilde\Lambda_0'\hat\Lambda\Big)^{-1},
\qquad
A_{NT_0}
:=
\Big(\frac{1}{T_0}F_0'F_0\Big)^{-1}.
\]
Then, by the Cauchy--Schwarz inequality,
\begin{align*}
\|\zeta_{NT_0}\|
&\le
\|D(\hat\Lambda)^{-1}\|\,
\frac{1}{N}
\Bigg(
\frac{1}{T_0}\sum_{t=p+1}^T \|\tilde W_t\|_F^2
\Bigg)^{1/2}
\Bigg(
\frac{1}{T_0}\sum_{t=p+1}^T
\|M_{\hat\Lambda}\Omega \hat\Lambda B_{NT_0}A_{NT_0}f_t^0\|^2
\Bigg)^{1/2}.
\end{align*}
Now
\[
\|D(\hat\Lambda)^{-1}\|=O_p(1)
\]
by Assumption~A(ii), and
\[
\|M_{\hat\Lambda}\|\le 1,
\qquad
\frac{1}{N}\frac{1}{T_0}\sum_{t=p+1}^T \|\tilde W_t\|_F^2 = O_p(1)
\]
by Assumption~A(i) together with Assumptions~B--D. Also,
\[
\|\Omega\|
=
\Big\|\frac{1}{T_0}\sum_{t=p+1}^T \Omega_t\Big\|
\le
\frac{1}{T_0}\sum_{t=p+1}^T \|\Omega_t\|
=O(1)
\]
by Assumption~C(iii).

Next, since \((1/N)\hat\Lambda'\hat\Lambda=I_r\), we have
\[
\|\hat\Lambda\|_F=\sqrt{Nr},
\]
so
\[
\frac{1}{\sqrt N}\|\hat\Lambda\|_F=O(1).
\]
By Proposition~\ref{prop:loading_pca}(ii) and Assumption~B(ii),
\[
\frac{1}{N}\tilde\Lambda_0'\hat\Lambda
=
\Big(\frac{1}{N}\tilde\Lambda_0'\tilde\Lambda_0\Big)H^{-1}+o_p(1),
\]
and hence
\[
\Big(\frac{1}{N}\tilde\Lambda_0'\hat\Lambda\Big)^{-1}=O_p(1),
\]
because \((1/N)\tilde\Lambda_0'\tilde\Lambda_0\) converges to a positive definite matrix and \(H^{-1}=O_p(1)\).
Moreover, by Assumption~B(i),
\[
\Big(\frac{1}{T_0}F_0'F_0\Big)^{-1}=O(1),
\qquad
\frac{1}{T_0}\sum_{t=p+1}^T \|f_t^0\|^2 = O(1).
\]

Therefore,
\begin{align*}
\frac{1}{T_0}\sum_{t=p+1}^T
\|M_{\hat\Lambda}\Omega \hat\Lambda B_{NT_0}A_{NT_0}f_t^0\|^2
&\le
\|M_{\hat\Lambda}\|^2
\|\Omega\|^2
\|\hat\Lambda\|_F^2
\|B_{NT_0}\|^2
\|A_{NT_0}\|^2
\frac{1}{T_0}\sum_{t=p+1}^T \|f_t^0\|^2 \\
&=
O_p(N).
\end{align*}
Combining the above bounds yields
\[
\|\zeta_{NT_0}\|
\le
O_p(1)\cdot
\frac{1}{N}\cdot
O_p(\sqrt N)\cdot
O_p(\sqrt N)
=
O_p(1).
\]
Thus
\[
\zeta_{NT_0}=O_p(1).
\]
\end{proof}

\begin{lemma}
\label{lem:rotation_projection_routeA}
Under Assumptions~A--D, let \(H\) be the rotation matrix defined in Proposition~\ref{prop:loading_pca}(ii). Then the following hold:
\begin{enumerate}[(i)]
\item
\[
HH'
=
\Big(\frac{1}{N}\tilde\Lambda_0'\tilde\Lambda_0\Big)^{-1}
+
O_p(\|\hat\beta-\beta_0\|)
+
O_p\left(\frac{1}{\sqrt{\min\{N,T_0\}}}\right).
\]

\item
\[
\|P_{\hat\Lambda}-\Pi_0\|^2
=
O_p(\|\hat\beta-\beta_0\|)
+
O_p\left(\frac{1}{\sqrt{T_0}}\right),
\]
where
\[
P_{\hat\Lambda}:=N^{-1}\hat\Lambda\hat\Lambda',
\qquad
\Pi_0:=\tilde\Lambda_0(\tilde\Lambda_0'\tilde\Lambda_0)^{-1}\tilde\Lambda_0'.
\]
\end{enumerate}
\end{lemma}

\begin{proof}[Proof of Lemma~\ref{lem:rotation_projection_routeA}]
Let
\[
\Delta_\Lambda:=\hat\Lambda-\tilde\Lambda_0 H.
\]
By Proposition~\ref{prop:loading_pca}(ii),
\[
\frac{1}{N}\|\Delta_\Lambda\|_F^2
=
O_p(\|\hat\beta-\beta_0\|^2)
+
O_p\!\Big(\frac{1}{\min\{N,T_0\}}\Big),
\qquad
\frac{1}{N}\hat\Lambda'\hat\Lambda=I_r.
\]
Hence
\[
I_r
=
\frac{1}{N}\hat\Lambda'\hat\Lambda
=
\frac{1}{N}(H'\tilde\Lambda_0'+\Delta_\Lambda')
(\tilde\Lambda_0 H+\Delta_\Lambda).
\]
Expanding gives
\[
I_r
=
H'\Big(\frac{1}{N}\tilde\Lambda_0'\tilde\Lambda_0\Big)H
+
\frac{1}{N}H'\tilde\Lambda_0'\Delta_\Lambda
+
\frac{1}{N}\Delta_\Lambda'\tilde\Lambda_0 H
+
\frac{1}{N}\Delta_\Lambda'\Delta_\Lambda.
\]
By Cauchy--Schwarz and the rate for \(\Delta_\Lambda\), the last three terms are
\[
O_p(\|\hat\beta-\beta_0\|)+O_p\left(\frac{1}{\sqrt{\min\{N,T_0\}}}\right).
\]
Therefore
\[
H'\Big(\frac{1}{N}\tilde\Lambda_0'\tilde\Lambda_0\Big)H
=
I_r
+
O_p(\|\hat\beta-\beta_0\|)
+
O_p\left(\frac{1}{\sqrt{\min\{N,T_0\}}}\right),
\]
which implies
\[
HH'
=
\Big(\frac{1}{N}\tilde\Lambda_0'\tilde\Lambda_0\Big)^{-1}
+
O_p(\|\hat\beta-\beta_0\|)
+
O_p\left(\frac{1}{\sqrt{\min\{N,T_0\}}}\right).
\]
This proves (i).

Next, write
\[
P_{\hat\Lambda}
=
\frac{1}{N}\hat\Lambda\hat\Lambda',
\qquad
\Pi_0
=
\tilde\Lambda_0(\tilde\Lambda_0'\tilde\Lambda_0)^{-1}\tilde\Lambda_0'.
\]
Using \(\hat\Lambda=\tilde\Lambda_0 H+\Delta_\Lambda\), we have
\begin{align*}
P_{\hat\Lambda}
&=
\frac{1}{N}(\tilde\Lambda_0 H+\Delta_\Lambda)(\tilde\Lambda_0 H+\Delta_\Lambda)' \\
&=
\tilde\Lambda_0 HH'\tilde\Lambda_0'/N
+
\frac{1}{N}\tilde\Lambda_0 H\Delta_\Lambda'
+
\frac{1}{N}\Delta_\Lambda H'\tilde\Lambda_0'
+
\frac{1}{N}\Delta_\Lambda\Delta_\Lambda'.
\end{align*}
By part (i), the first term equals $\Pi_0+O_p(\|\hat\beta-\beta_0\|)+O_p\left(\frac{1}{\sqrt{\min\{N,T_0\}}}\right)$, and the remaining three terms are of the same or smaller order by Cauchy--Schwarz. Hence
\[
\|P_{\hat\Lambda}-\Pi_0\|^2
=
O_p(\|\hat\beta-\beta_0\|)
+
O_p\left(\frac{1}{\sqrt{T_0}}\right).
\]
This proves (ii).
\end{proof}

\begin{lemma}
\label{lem:score_replacement_routeA}
Under Assumptions~A--D, define
\[
a_{ts}:=
f_t^{0\prime}\Big(\frac{1}{T_0}F_0^{\prime}F_0\Big)^{-1}f_s^0,
\qquad
F_0:=(f_{p+1}^0,\ldots,f_T^0)' .
\]
Assume \(N/T_0\to 0\). Then
\begin{align}\label{eq:A8_routeA}
&\frac{1}{\sqrt{NT_0}}\sum_{t=p+1}^T
\Bigg[
\tilde W_t' M_{\hat\Lambda}
-
\Big\{\frac{1}{T_0}\sum_{s=p+1}^T a_{ts}\tilde W_s'\Big\}M_{\hat\Lambda}
\Bigg]\tilde\varepsilon_t
\nonumber\\
&\qquad=
\frac{1}{\sqrt{NT_0}}\sum_{t=p+1}^T
\Bigg[
\tilde W_t' M_{\Pi_0}
-
\Big\{\frac{1}{T_0}\sum_{s=p+1}^T a_{ts}\tilde W_s'\Big\}M_{\Pi_0}
\Bigg]\tilde\varepsilon_t
+o_p(1).
\end{align}
\end{lemma}

\begin{proof}[Proof of Lemma~\ref{lem:score_replacement_routeA}]
Let
\[
\Delta_M:=M_{\hat\Lambda}-M_{\Pi_0}
=
\Pi_0-P_{\hat\Lambda}.
\]
Then the left-hand side minus the right-hand side of \eqref{eq:A8_routeA} is
\begin{align*}
R_{NT_0}
&=
\frac{1}{\sqrt{NT_0}}\sum_{t=p+1}^T
\Bigg[
\tilde W_t' \Delta_M
-
\Big\{\frac{1}{T_0}\sum_{s=p+1}^T a_{ts}\tilde W_s'\Big\}\Delta_M
\Bigg]\tilde\varepsilon_t \\
&=:
R_{1,NT_0}-R_{2,NT_0},
\end{align*}
where
\[
R_{1,NT_0}
:=
\frac{1}{\sqrt{NT_0}}\sum_{t=p+1}^T
\tilde W_t' \Delta_M \tilde\varepsilon_t,
\qquad
R_{2,NT_0}
:=
\frac{1}{\sqrt{NT_0}}\sum_{t=p+1}^T
\Big\{\frac{1}{T_0}\sum_{s=p+1}^T a_{ts}\tilde W_s'\Big\}\Delta_M \tilde\varepsilon_t .
\]

By Cauchy--Schwarz,
\[
\|R_{1,NT_0}\|
\le
\|\Delta_M\|
\Big\|
\frac{1}{\sqrt{NT_0}}\sum_{t=p+1}^T \tilde W_t'\tilde\varepsilon_t
\Big\|,
\]
and
\[
\|R_{2,NT_0}\|
\le
\|\Delta_M\|
\Big\|
\frac{1}{\sqrt{NT_0}}\sum_{t=p+1}^T
\Big\{\frac{1}{T_0}\sum_{s=p+1}^T a_{ts}\tilde W_s'\Big\}\tilde\varepsilon_t
\Big\|.
\]

By Lemma~\ref{lem:rotation_projection_routeA}(ii),
\[
\|\Delta_M\|^2
=
\|P_{\hat\Lambda}-\Pi_0\|^2
=
O_p(\|\hat\beta-\beta_0\|)+O_p(T_0^{-1}).
\]
Since \(N/T_0\to 0\), Proposition~\ref{prop:consistency_ours}(i) implies
\[
\|\Delta_M\|=o_p(1).
\]

It therefore remains to show that the two bracketed terms are \(O_p(1)\).

\medskip
\noindent\textbf{Step 1: the unweighted score term.}
Define the deterministic array \(\{c^{(1)}_{ts}\}_{t,s=p+1}^T\) by
\[
c^{(1)}_{ts}:=\mathbf 1\{t=s\}.
\]
Then \(c^{(1)}_{ts}\) is uniformly bounded, and
\[
\sum_{s=p+1}^T c^{(1)}_{ts}\tilde W_s'\tilde\varepsilon_t
=
\tilde W_t'\tilde\varepsilon_t.
\]
Hence Assumption~C(vi) yields
\[
\frac{1}{NT_0}
\sum_{t,u=p+1}^T
\left|
E\!\left[
(\tilde W_t'\tilde\varepsilon_t)(\tilde W_u'\tilde\varepsilon_u)
\right]
\right|
\le K .
\]
Therefore
\[
E\Bigg\|
\frac{1}{\sqrt{NT_0}}\sum_{t=p+1}^T \tilde W_t'\tilde\varepsilon_t
\Bigg\|^2
\le K,
\]
and thus
\[
\Big\|
\frac{1}{\sqrt{NT_0}}\sum_{t=p+1}^T \tilde W_t'\tilde\varepsilon_t
\Big\|=O_p(1).
\]

\medskip
\noindent\textbf{Step 2: the weighted score term.}
Define the deterministic array \(\{c^{(2)}_{ts}\}_{t,s=p+1}^T\) by
\[
c^{(2)}_{ts}:=\frac{1}{T_0}a_{ts}.
\]
By Assumption~B(i), the coefficients \(a_{ts}\) are uniformly bounded, so
\(\{c^{(2)}_{ts}\}\) is also uniformly bounded. Moreover,
\[
\sum_{s=p+1}^T c^{(2)}_{ts}\tilde W_s'\tilde\varepsilon_t
=
\Big\{\frac{1}{T_0}\sum_{s=p+1}^T a_{ts}\tilde W_s'\Big\}\tilde\varepsilon_t.
\]
Applying Assumption~C(vi) again, we obtain
\[
\frac{1}{NT_0}
\sum_{t,u=p+1}^T
\left|
E\!\left[
\Big(\frac{1}{T_0}\sum_{s=p+1}^T a_{ts}\tilde W_s'\tilde\varepsilon_t\Big)
\Big(\frac{1}{T_0}\sum_{v=p+1}^T a_{uv}\tilde W_v'\tilde\varepsilon_u\Big)
\right]
\right|
\le K .
\]
Therefore
\[
E\Bigg\|
\frac{1}{\sqrt{NT_0}}\sum_{t=p+1}^T
\Big\{\frac{1}{T_0}\sum_{s=p+1}^T a_{ts}\tilde W_s'\Big\}\tilde\varepsilon_t
\Bigg\|^2
\le K,
\]
and hence
\[
\Big\|
\frac{1}{\sqrt{NT_0}}\sum_{t=p+1}^T
\Big\{\frac{1}{T_0}\sum_{s=p+1}^T a_{ts}\tilde W_s'\Big\}\tilde\varepsilon_t
\Big\|
=O_p(1).
\]

Combining Steps 1 and 2, we conclude that
\[
\|R_{1,NT_0}\|+\|R_{2,NT_0}\|=o_p(1),
\]
and therefore
\[
R_{NT_0}=o_p(1).
\]
This proves \eqref{eq:A8_routeA}.
\end{proof}

\begin{corollary}
\label{cor:linear_representation_routeA}
Under Assumptions~A--D and \(N/T_0\to 0\), we have
\begin{align*}
\sqrt{NT_0}\,(\hat\beta-\beta_0)
&=
D(\hat\Lambda)^{-1}\,
\frac{1}{\sqrt{NT_0}}\sum_{t=p+1}^T
\Bigg[
\tilde W_t' M_{\Pi_0}
-
\Big\{\frac{1}{T_0}\sum_{s=p+1}^T a_{ts}\tilde W_s'\Big\}M_{\Pi_0}
\Bigg]\tilde\varepsilon_t
+o_p(1),
\end{align*}
where
\[
a_{ts}=f_t^{0\prime}\Big(\frac{1}{T_0}F_0^{\prime}F_0\Big)^{-1}f_s^0,
\qquad
F_0=(f_{p+1}^0,\ldots,f_T^0)',
\]
and \(D(\cdot)\) is defined in Assumption~A(ii).
\end{corollary}

\begin{proof}[Proof of Corollary~\ref{cor:linear_representation_routeA}]
By Proposition~\ref{prop:beta_linear_expansion},
\begin{align*}
\sqrt{NT_0}\,(\hat\beta-\beta_0)
&=
D(\hat\Lambda)^{-1}
\frac{1}{\sqrt{NT_0}}\sum_{t=p+1}^T
\Bigg[
\tilde W_t' M_{\hat\Lambda}
-
\Big\{\frac{1}{T_0}\sum_{s=p+1}^T a_{ts}\tilde W_s'\Big\}M_{\hat\Lambda}
\Bigg]\tilde\varepsilon_t \\
&\quad
+\sqrt{\frac{N}{T_0}}\zeta_{NT_0}
+o_p(1).
\end{align*}
Under \(N/T_0\to 0\), the bias term is negligible, and by Lemma~\ref{lem:score_replacement_routeA},
the score term with \(M_{\hat\Lambda}\) can be replaced by the one with \(M_{\Pi_0}\) up to \(o_p(1)\).
Therefore
\begin{align*}
\sqrt{NT_0}\,(\hat\beta-\beta_0)
&=
D(\hat\Lambda)^{-1}\,
\frac{1}{\sqrt{NT_0}}\sum_{t=p+1}^T
\Bigg[
\tilde W_t' M_{\Pi_0}
-
\Big\{\frac{1}{T_0}\sum_{s=p+1}^T a_{ts}\tilde W_s'\Big\}M_{\Pi_0}
\Bigg]\tilde\varepsilon_t
+o_p(1),
\end{align*}
which proves the claim.
\end{proof}

\section*{Acknowledgments}
\if0\blind
The authors are grateful to Takuya Ishihara of Tohoku University and Daisuke Kurisu of the University of Tokyo for valuable comments. Ishihara provided helpful suggestions on the analysis of density-valued data, and Kurisu offered valuable comments on defining the inner product in an isometric way.
\fi

\section*{Funding}
\if0\blind
This work was supported by a JSPS KAKENHI Grant-in-Aid for Scientific Research (A), Grant Number 26H01952.
\fi

\section*{Disclosure Statement}
The authors report there are no competing interests to declare.

\section*{Data and Code Availability Statement}
Code and replication materials will be made available upon acceptance.
Access to the raw Ct-value data is subject to the terms of the data provider.

\bibliographystyle{apalike}
\bibliography{Bibliography}

\end{document}